\documentclass[11pt,oneside,english,reqno]{amsart}
\usepackage{geometry}
\usepackage[normalem]{ulem} 
\usepackage{color}
\usepackage{float}
\usepackage{amsthm}
\usepackage{amstext}
\usepackage{amssymb}
\usepackage{graphicx}
\usepackage[pdfusetitle, bookmarks = true,bookmarksnumbered = false,bookmarksopen = false,
 breaklinks = false,pdfborder = {0 0 1},backref = false,colorlinks = true]{hyperref}

\makeatletter
\numberwithin{equation}{section}
\numberwithin{figure}{section}
  \theoremstyle{definition}
  \newtheorem{defn}{\protect\definitionname}
  \theoremstyle{plain}
  \newtheorem{thm}{\protect\theoremname}
  \theoremstyle{plain}
  \newtheorem{lem}{\protect\lemmaname}
  \theoremstyle{plain}
  \newtheorem{cor}{\protect\corollaryname}
 \ifx\proof\undefined\
   \newenvironment{proof}[1][\proofname]{\par
     \normalfont\topsep6\p@\@plus6\p@\relax
     \trivlist
     \itemindent\parindent
     \item[\hskip\labelsep
           \scshape
       #1]\ignorespaces
   }{%
     \endtrivlist\@endpefalse
   }
   \providecommand{\proofname}{Proof}
 \fi

\usepackage[foot]{amsaddr}
\usepackage{amsthm}
\makeatother

\makeatother

\usepackage{subfigure}
\makeatother
\usepackage{tikz}
\usepackage{scalefnt}
\usepackage[english]{babel}
 \makeatletter
   \renewcommand{\subsection}{\@startsection{subsection}{1}{0mm}
   {\baselineskip}%
   {0.5\baselineskip}{\normalfont\normalsize\bfseries}}%
   \makeatother

\makeatother

\providecommand{\corollaryname}{Corollary}
\providecommand{\definitionname}{Definition}
\providecommand{\lemmaname}{Lemma}
\providecommand{\theoremname}{Theorem}

\newcommand{\iii}{\mathbb{I}}
\newcommand{\ket}[1]{\vert #1 \rangle}
\newcommand{\bra}[1]{\langle #1 \vert}

\begin{document}
\title{Quantum 3-SAT is $\text{QMA}_1$-Complete}
\author{David Gosset$^{\dagger}$}
\author{Daniel Nagaj$^{\star}$}
\address{$^\dagger$ Department of Combinatorics \& Optimization and Institute for Quantum Computing, University of Waterloo}
\address{$^\star$ Faculty of Physics, University of Vienna}
\email{dngosset@gmail.com, daniel.nagaj@univie.ac.at}
\begin{abstract}
Quantum satisfiability is a constraint satisfaction problem that generalizes classical boolean satisfiability. In the quantum $k$-SAT problem, each constraint is specified by a $k$-local projector and is satisfied by any state in its nullspace. Bravyi showed that quantum $2$-SAT can be solved efficiently on a classical computer and that quantum $k$-SAT with $k\geq 4$ is QMA$_1$-complete \cite{Bravyi06}. Quantum $3$-SAT was known to be contained in QMA$_1$ \cite{Bravyi06}, but its computational hardness was unknown until now. We prove that quantum $3$-SAT is QMA$_1$-hard, and therefore complete for this complexity class. 
\end{abstract}
\maketitle


\section{Introduction \label{sec:Intro}}

Satisfiability is a widely studied constraint satisfaction problem.
In this problem one is given a list of constraints on $n$ boolean variables and asked
to determine if there is an $n$-bit string which satisfies each of
them. When each constraint involves at most $k$ variables this problem
is called $k$-satisfiability, or $k$-SAT. Although $2$-SAT can
be solved efficiently on a classical computer, $3$-SAT is NP-complete.

Bravyi introduced a quantum generalization of satisfiability \cite{Bravyi06}.
In the quantum $k$-SAT problem the constraints act on an $n$-qubit
Hilbert space and we are asked to determine if there is a state which
satisfies all of them. Each constraint is specified by a $k$-local%
\footnote{ A $k$-local operator acts nontrivially on at most $k$ qubits and
as the identity on all other qubits.%
} projector and is satisfied by any state in its nullspace. 

Bravyi proved that quantum $2$-SAT is in P, quantum $k$-SAT is contained
in QMA$_{1}$ for any constant $k$, and quantum $k$-SAT for $k\geq4$
is QMA$_{1}$-complete. The computational complexity of quantum $3$-SAT
has been an open question. In this paper, we resolve this question
by proving that quantum 3-SAT is QMA$_{1}$-complete.

Less is known about quantum $k$-SAT than is known about its classical
counterpart. Random instances of quantum $k$-SAT have been studied
by Laumann et. al. as a function of the clause density $\alpha$ \cite{Laumann10,LaumannProd}.
As in the classical case, it is conjectured that a satisfiability
threshold $\alpha_{c}(k)$ exists, above which the probability that
a random instance is satisfiable approaches zero as $n\rightarrow\infty$
and below which this probability approaches one \cite{Laumann10}.
Some bounds on this threshold value have been proven using a quantum
version of the Lov{\'{a}}sz  local lemma \cite{quantumLovasz} and by using
graph-theoretic techniques \cite{BravyiMoore} but only the case $k = 2$
is fully understood \cite{Laumann10,JZZfrustrationfree}. 
Other previous work has focused on quantum satisfiability
with qudit variables of dimension $d>2$ \cite{Ramisline,New3local,Eldar08,CriticalityWithoutFrustration}
or in restricted geometries \cite{Ramisline,Treesqudits}.

Quantum satisfiability is relevant to
the study of \emph{frustration-free} Hamiltonians. A frustration-free
Hamiltonian is a local Hamiltonian (a sum of $k$-local Hermitian operators for some constant $k$) with groundstates that minimize the energy of each local term individually. Such Hamiltonians naturally arise in the study of quantum error correction  and play a central role in the field of Hamiltonian complexity, which touches upon condensed matter physics, quantum computation and computational complexity theory \cite{OsborneReview}.  We can view quantum $k$-SAT as the problem where one is asked to determine if a sum of $k$-local
projectors is frustration-free, that is to say, if its ground energy
is zero. 

The computational complexity of quantum $k$-SAT is naturally compared
with that of the $k$-Local Hamiltonian problem, which can be viewed
as the quantum analogue of MAX $k$-SAT. In this problem one is given 
a Hamiltonian which is a sum of $k$-local operators, along with constants
$a,b$ such that $a<b$. One is asked to determine if the ground energy 
of the given Hamiltonian  is less than $a$ (yes instance) or greater 
than $b$ (no instance), promised that one of these conditions holds. 
Note that, for yes instances of this problem, a ground state of the 
Hamiltonian need not have minimal energy for each $k$-local term; 
such a system can be frustrated. Because of the possibility of frustration,
the $k$-Local Hamiltonian problem can be computationally more difficult
than quantum $k$-SAT. Indeed, the $k$-local Hamiltonian problem
is QMA-complete for $k\geq2$ \cite{LocalHam}. On the other hand,
our result that quantum $3$-SAT is QMA$_{1}$-complete shows that
these problems are of comparable difficulty for $k\geq3$ (putting
aside the subtle differences between the definitions of QMA and QMA$_{1}$). 

In the next Section, we provide the basic definitions, describe our
results in more detail and give an overview of the remainder of the
paper.


\section{Definitions and Overview\label{sec:Preliminaries-and-Overview}}

We first define the complexity class QMA, or Quantum Merlin-Arthur.
This class gets its name from a scenario involving Merlin and Arthur,
who at the outset are both given an instance of a promise problem
encoded as a bit-string $X$. Arthur wishes to know the correct answer
to this problem (which is either yes or no) but his time and space
resources are bounded as polynomial functions of $|X|$. Merlin has
unbounded computational power and can easily obtain the correct answer.
Merlin wants to convince Arthur the answer is ``yes'', but Arthur
doesn't trust Merlin, so he asks for proof. Merlin hands over an $n$-qubit
quantum state $|W\rangle$ (called a witness) that Arthur uses to
verify the claim in the following way. He adjoins some number $n_{a}$
of ancilla qubits each in the state $|0\rangle$ to produce 
\[
|W\rangle|0\rangle^{\otimes n_{a}}
\]
 (the total number $n + n_{a}$ of qubits in this state must be upper
bounded by a polynomial in $|X|$), then applies a polynomial sized
verification circuit $U_{X}$ and then measures the first ancilla
qubit in the computational basis. If the measurement outcome is $1$,
he accepts Merlin's claim that $X$ is a yes instance. Arthur's acceptance
probability given the state $|W\rangle$ is therefore 
\begin{equation}
	\text{AP}\left(U_{X},|W\rangle\right)
	 =  \left\Vert 
		\left(\iii ^{\otimes n} \otimes |1\rangle\langle 1|_{(n + 1)}
			\otimes \iii^{\otimes(n_{a}-1)}\right)
			U_{X}|W\rangle|0\rangle^{\otimes n_{a}}
	\right\Vert^{2}.\label{eq:acceptanceprobability}
\end{equation}
For problems in the class QMA, if Merlin is being truthful he can
convince Arthur with probability at least $\frac{2}{3}$. On the other
hand, if Merlin is lying (i.e., the answer is actually ``no'') then
he can only fool Arthur with probability at most $\frac{1}{3}$. 
\begin{defn}
[\textbf{QMA}]

\noindent A promise problem $L_{\text{yes}}\cup L_{\text{no}}\subset\{0,1\}^{*}$
is contained in QMA if and only if there exists a uniform polynomial-size
quantum circuit family $U_{X}$ such that

\medskip{}

If $X\in L_{\text{yes}}$ there exists a state $|W\rangle$ such that
$\text{AP}\left(U_{X},|W\rangle\right)\geq\frac{2}{3}$ (completeness).

\medskip{}

If $X\in L_{\text{no}}$ then $\text{AP}\left(U_{X},|W\rangle\right)\leq\frac{1}{3}$
for any state $|W\rangle$ (soundness).

\smallskip{}

\end{defn}
As indicated above, the first property is called \emph{completeness}
and the second property \emph{soundness}. Here we have defined QMA
with constant completeness $\frac{2}{3}$ and soundness $\frac{1}{3}$.
Kitaev showed that these parameters can be amplified: we obtain an
equivalent definition with soundness $2^{-\Omega(|X|^{\alpha})}$
and completeness $1-2^{-\Omega(|X|^{\alpha})}$ for any constant $\alpha$
\cite{Kitaev:2002:CQC:863284} (see also \cite{MarriottWatrous,fastQMAamp}). 

QMA$_{\text{1}}$ is defined in a similar way to QMA with two modifications.
The first is ``perfect'' completeness -- for $X\in L_{\text{yes}}$,
Merlin can convince Arthur with probability exactly equal to $1$.
The second difference is that Arthur's verification circuit must consist
of a sequence of gates from a fixed universal gate set $\mathcal{G}$.
The definition of QMA$_{1}$ is not known to be independent of the
gate set used. In this paper we use the standard choice
\begin{equation}
	\mathcal{G} = \{\widehat{H},T,\text{CNOT}\},\label{eq:Gset}
\end{equation}
 where 
\[
	\widehat{H} = \frac{1}{\sqrt{2}}\left(\begin{array}{cc}
		1 & 1\\
		1 & -1
	\end{array}\right),\quad 
	T = \left(\begin{array}{cc}
		1 & 0\\
		0 & e^{\frac{i\pi}{4}}
	\end{array}\right),\quad
	\text{CNOT} = \left(\begin{array}{cccc}
		1 & 0 & 0 & 0\\
		0 & 1 & 0 & 0\\
		0 & 0 & 0 & 1\\
		0 & 0 & 1 & 0
	\end{array}\right).
\]

\begin{defn}
[\textbf{QMA}$_\textbf{1}$]\label{QMA_one}

A promise problem $L_{\text{yes}}\cup L_{\text{no}}\subset\{0,1\}^{*}$
is contained in QMA$_{\text{1}}$ if and only if there exists a uniform
polynomial-size quantum circuit family $U_{X}$ over the gate set
$\mathcal{G}$ such that

\medskip{}

If $X\in L_{\text{yes}}$ there exists a state $|W\rangle$ such that
$\text{AP}\left(U_{X},|W\rangle\right) = 1$ (perfect completeness).

\medskip{}

If $X\in L_{\text{no}}$ then $\text{AP}\left(U_{X},|W\rangle\right)\leq\frac{1}{3}$
for any state $|W\rangle$ (soundness).

\smallskip{}

\end{defn}

Just as with QMA, the soundness of a QMA$_{1}$ verification procedure
(taken to be $\frac{1}{3}$ in the above) can be amplified so that
it is very close to zero \cite{Kitaev:2002:CQC:863284}. 

Note that we could modify the definition of QMA$_{1}$ so that Arthur's
verification procedure may use intermediate measurements in the computational
basis as well as operations conditioned on these measurement outcomes.
However, expanding the set of allowed verification protocols in this
way we obtain an equivalent definition of QMA$_{1}$. To see this,
note that the controlled version of any unitary written as a product
of gates from $\mathcal{G}$ can also be written as a product of gates
from $\mathcal{G}$. Given a verification procedure which has some
intermediate measurements we can use controlled unitaries in place
of all conditional operations. This is the standard trick for deferring
measurements until the end of an algorithm; in this case we need not
ever measure any of the qubits except the ancilla which gives the
output of the computation.

We think of QMA$_{1}$ as being very similar to QMA, although the
precise relationship between these two classes has yet to be determined.
One might expect that QMA$ = $QMA$_{1}$ since the analogous classical
complexity classes MA and MA$_{1}$ are known to be equal \cite{Zachos87}.
This is an interesting open question in quantum complexity theory
but we do not explore it further here (see \cite{AaronsonQMA, JordanQCMA, QMAconstEPR} for recent developments). 
For us, the one-sided error in QMA$_{1}$ is primarily
an annoyance that prevents the use of standard techniques such as
converting between different universal gate sets.

Let us now turn our attention to quantum $3$-SAT. In this problem
we are given a Hamiltonian
\[
H = \sum_{i = 1}^{r}\Pi_{i}
\]
 that is a sum of $3$-local projectors $\Pi_{i}$ acting on an $n$-qubit
Hilbert space. We are promised that either $H$ has ground state energy
zero, or else its ground state energy is greater than a constant (which
we take without loss of generality to be $1$) and we are asked to
decide which is the case. 

Note that the matrix elements of a projector $\Pi_{i}$ in an instance
of quantum $3$-SAT cannot be specified as arbitrary complex numbers
with unlimited precision. In our definition of quantum $3$-SAT we
must constrain the set of allowed projectors in some way. One possible
choice would be to consider projectors where each matrix element is
given to some number of decimal places, but this choice is not well
suited to the one-sided error setting we are working in. In particular,
it is not clear whether quantum $3$-SAT, defined in this way, is
contained in QMA$_{1}$ (although it is contained in QMA). On the other hand, if instead we consider projectors where each matrix element is expressed as a rational number with denominator $2^s$ (for some fixed integer $s$) then the problem is contained in QMA$_{1}$. More generally, let us consider defining quantum $3$-SAT with some set of allowed projectors $\mathcal{P}$. Bravyi proved
that quantum $3$-SAT is contained in QMA$_{1}$ as long as $\mathcal{P}$
is chosen to satisfy the following condition: there is an efficient
algorithm which can be used to measure the eigenvalue of a projector
$\Pi\in\mathcal{P}$ in a given state $\ket{\psi}$ using the gate set
$\mathcal{G}$ \cite{Bravyi06,Bravyi_communication}. Of course, since
$\mathcal{G}$ is universal, such a measurement can always be \textit{approximated}
using gates from this set, but here we require an {\em exact} measurement
because of the one-sided error. We discuss this issue in more detail 
in  Appendix \ref{sec:containmentINqma1}, but we mention this here 
so the reader is aware of the subtlety. The situation is further 
complicated by the fact that this technical point was treated incorrectly in reference \cite{Bravyi06}
due to an error in the proof of Lemma 5 \cite{Bravyi_communication}.

In this work we define quantum $3$-SAT with a restricted set of projectors
$\mathcal{P}$ given below. While quantum $3$-SAT remains in QMA$_{1}$
for larger classes of projectors, restricting to a smaller set makes
our QMA$_{1}$-hardness result stronger. The specific set $\mathcal{P}$
that we use arises from technical considerations.
\begin{defn}
\label{Delta_proj}Let $\mathcal{{P}}$ be the set of $3$-local projectors
$\Pi$ which satisfy one of the following two conditions:

1. Every matrix element of $\Pi$ in the computational basis has the
form 
\begin{equation}
\frac{1}{4}\left(a + i\, b + \sqrt{2}\, c + i\sqrt{2}\, d\right)\label{eq:mat_el_ring}
\end{equation}
for $a,b,c,d\in\mathbb{Z}.$ 

2. There is a $3$-qubit unitary $U$ with matrix elements of the
form \eqref{eq:mat_el_ring} (in the computational basis) such that
$U\Pi U^{\dagger}$ is equal to 
\[
\left(\sqrt{\frac{1}{3}}|000\rangle-\sqrt{\frac{2}{3}}|001\rangle\right)\left(\sqrt{\frac{1}{3}}\langle 000|-\sqrt{\frac{2}{3}}\langle 001|\right)
\]
on $3$ of the qubits tensored with the identity on the remaining
qubits.
\end{defn}
We define quantum $3$-SAT with projectors from $\mathcal{{P}}.$
\begin{defn}
[\textbf{Quantum 3-SAT}]Given a collection $\{\Pi_{i}\::\: i = 1,\dots,r\}\subset\mathcal{P}$
of $3$-local projectors acting on $n$ qubits, we are asked to decide
if they correspond to a yes instance or a no instance (promised that
one is the case), where \medskip{}

Yes: There exists an $n$-qubit state $\ket{\psi}$ satisfying $\Pi_{i}\ket{\psi} = 0$
for all $i = 1,\dots,r$.\medskip{}

No: $\sum_{i}\langle\psi|\Pi_{i}\ket{\psi}\geq1$ for all $\ket{\psi}$.
\end{defn}
With the definitions given above, we prove that quantum $3$-SAT is
QMA$_{1}$-complete. In Appendix \ref{sec:containmentINqma1} we
show (following Bravyi \cite{Bravyi06,Bravyi_communication}) that
quantum $3$-SAT is contained in QMA$_{1}$. Our main result in this
paper is QMA$_{1}$-hardness of quantum $3$-SAT. To prove this, we
exhibit an efficiently computable mapping from a $g$-gate, $(n + n_{a})$-qubit
verification circuit that implements a unitary $U_{X}$ to a Hamiltonian
\[
H_{X} = \sum_{i}\Pi_{i,X}
\]
 which is a sum of $\Theta\left(n_{a} + g\right)$ 3-local projectors
$\Pi_{i,X}\in\mathcal{P}$ acting on $\Theta(n + n_{a} + g)$ qubits. Moreover, we prove
\begin{thm}
[\textbf{Completeness}]\label{thm:Completeness}
	$H_{X}$ has ground energy $0$ if and only if there exists 
	$|W\rangle$ such that $\text{AP}\left(U_{X},|W\rangle\right) = 1,$ 
\end{thm}
and
\begin{thm}
[\textbf{Soundness}]\label{thm:Soundness}
	If $\text{AP}\left(U_{X},|W\rangle\right)\leq\frac{1}{3}$
	for all $|W\rangle$, then $H_{X}$ has ground energy $\Omega\left(\frac{1}{g^{6}}\right)$. 
\end{thm}
Note that in our definition of quantum $3$-SAT we require that in
the ``no'' case the ground energy is greater than or equal to $1$,
whereas Theorem \ref{thm:Soundness} gives a bound of $\Omega\left(g^{-6}\right)$.
To form an instance of quantum $3$-SAT as defined above we repeat
each projector $\Pi_{i,X}$ in the instance a suitable number of times
(i.e., $\Theta(g^{6})$ times). This shows that any promise problem
in QMA$_{1}$ can be reduced to quantum $3$-SAT. Since quantum $3$-SAT
is also contained in QMA$_{1}$ we have proven that it is complete
for this complexity class. 

Our mapping from the verification circuit $U_{X}$ to the Hamiltonian
$H_{X}$ relies on two technical innovations. Like many previous works
in the field of Hamiltonian complexity, we use a ``clock construction''.
In this work we introduce a new one which has some special properties. Most previous
QMA- or QMA$_{1}$-hardness results use a circuit-to-Hamiltonian mapping
which is an immediate and simple application of the clock construction (the standard approach uses a Hilbert space with a computational
register and a clock register). In contrast, in this work we define a novel circuit-to-Hamiltonian mapping where the Hamiltonian
$H_{X}$ acts on a Hilbert space with a computational register along
with \textit{two }clock registers.

The paper is organized as follows. In Section \ref{sec:A-new-clock}
we introduce our new clock construction. In Section \ref{sec:Hamiltonians-with-two}
we describe Hamiltonians acting on two clock registers which illustrate
the main ideas of our proof that quantum $3$-SAT is QMA$_{1}$-hard.
Then in Section \ref{sub:Completeness_Section} we define the Hamiltonian
$H_{X}$ and in Section \ref{sub:The-zero-energy} we characterize
its zero energy eigenspace, establishing Theorem \ref{thm:Completeness}.
We prove Theorem \ref{thm:Soundness} in Section \ref{sub:Soundness_Section},
using a spectral bound due to Kitaev \cite{Kitaev:2002:CQC:863284}.


\section{A new clock construction\label{sec:A-new-clock}}

A clock construction is a local Hamiltonian along with a set of local
operators which act on its groundspace in a certain way. It can be used as a set of building blocks to
define more complicated Hamiltonians while keeping track of the groundspace. The novel clock construction that we introduce in this Section is a central part of our QMA$_1$ hardness result. In this Section we focus on the construction itself and discuss its properties. Then in Section \ref{sec:Hamiltonians-with-two} we show how to build interesting Hamiltonians using two clock registers. This idea is developed further in Section \ref{sec:Circuit-to-Hamiltonian-mapping} where we define our circuit-to-Hamiltonian mapping.

In Section \ref{sub:Feynman's-circuit-to-Hamiltonian} we give some
historical context, in order to motivate the notion of a clock construction.
Then in Section \ref{sub:Summary-of-the} we summarize the properties
of the new one that we introduce in this work. Finally,
in Section \ref{sec:Details-of-the} we provide all the details of
our construction. On a first read through this paper the reader may
wish to skip Section \ref{sec:Details-of-the} since most of our results
(with the exception of those presented in Appendix \ref{sec:Proof-of-Theorem})
are derived using only the properties described in Section \ref{sub:Summary-of-the}.


\subsection{Feynman's circuit-to-Hamiltonian mapping and Kitaev's clock construction\label{sub:Feynman's-circuit-to-Hamiltonian}}

Feynman showed how to map a sequence of unitary operators into a Hamiltonian
in an interesting way \cite{Feynman}. This mapping is often used
to convert statements about quantum circuits into statements about
Hamiltonians. Let us start with a quantum circuit $\mathcal{C}$ that
acts on $n$ qubits and consists of a product of $m$ one- and two-qubit
unitaries 
\[
W_{m-1}W_{m-2}\dots W_{0}.
\]
Now consider a Hilbert space with two registers. The first register
holds $n$ qubits and the second register holds an $(m + 1)$-level
system with orthonormal basis states $|t\rangle$ for $t = 0,\dots,m$.
Feynman's Hamiltonian is 
\begin{equation}
H_{\text{F}}(\mathcal{C}) = \sum_{t = 0}^{m-1}\frac{1}{2}\left(\iii \otimes |t\rangle\langle t| + \iii \otimes |t + 1\rangle\langle t + 1|-W_{t}^{\dagger}\otimes|t\rangle\langle t + 1|-W_{t}\otimes|t + 1\rangle\langle t|\right).\label{eq:Feynman's Ham}
\end{equation}
(Strictly speaking, this definition differs slightly from Feynman's original Hamiltonian--we have added the first two terms in parentheses). Note that each of the terms in the sum is a projector. Zero eigenvectors
of the $t$th term which have all of their support on the clock states
$|t\rangle$ and $|t + 1\rangle$ have the form 
\[
\ket{\psi}|t\rangle + W_{t}\ket{\psi}|t + 1\rangle
\]
where $\ket{\psi}$is an $n$-qubit state. Using this fact, one
can show that every state in the zero energy groundspace of $H_{F}^{\text{}}(\mathcal{C})$
is of the form
\[
|\text{Hist}_{\text{F}}(\phi)\rangle = \frac{1}{\sqrt{m + 1}}\left(|\phi\rangle|0\rangle + \sum_{t = 1}^{m}W_{t-1}W_{t-2}\dots W_{0}|\phi\rangle|t\rangle\right)
\]
where $|\phi\rangle$ is an $n$-qubit state. The state $|\text{Hist}_{\text{F}}\mbox{(\ensuremath{\phi})}\rangle$
is sometimes called a history state because it contains every intermediate
state $W_{t-1}W_{t-2}\dots W_{0}|\phi\rangle$ obtained during the course
of applying the circuit $\mathcal{C}$ to the initial state $|\phi\rangle.$
The second register is called the clock register and keeps track of
the number of unitaries that have been applied. 

The clock register in Feynman's construction is an $(m + 1)$-level
system. However, Kitaev noticed that a modification of Feynman's idea
can be used to map the circuit $\mathcal{C}$ into a\emph{ }local
Hamiltonian acting on a system made only from qubits. Again the Hilbert
space has two registers; the first holds $n$ qubits as before but
now the second register has $m$ qubits. The idea is to encode Feynman's
$m + 1$ clock states as $m$-bit unary strings
\begin{eqnarray*}
	|t\rangle_{\text{u}} = |\underbrace{1\cdots1}_{t}\underbrace{0\cdots0}_{m-t}\rangle.
\end{eqnarray*}
The $(m + 1)$-dimensional span of these clock states is called the
clock subspace. It is the groundspace of the 2-local, $m$-qubit
Hamiltonian 
\[
H_{\text{unary}}^{(m)} = \sum_{i = 1}^{m-1}|01\rangle\langle 01|_{i,i + 1}.
\]
 This is the first part of Kitaev's clock construction: a local Hamiltonian
acting on a system of qubits which has the desired number of groundstates.
The second part of the clock construction is a set of local operators
which act in a special way on states that have all of their support
in the clock subspace. For any unitary $U$ acting on the first register,
define projectors\footnote{For $t = 0$ and $t = m-1$ these operators should be defined in a slightly
different but self-explanatory way.%
} 
\begin{align}
	h_{t,t + 1}^{\text{u}}(U)
	 = 	\frac{1}{2} &\left( 
			\iii \otimes |100\rangle\langle 100|_{t,t + 1,t + 2} + \iii \otimes |110\rangle\langle 110|_{t,t + 1,t + 2} 
		\right) \label{eq:kitaev_transition_ops}\\
		 - \frac{1}{2} &\left(
			U^{\dagger}\otimes|100\rangle\langle 110|_{t,t + 1,t + 2} + U\otimes|110\rangle\langle 100|_{t,t + 1,t + 2}
		\right). \nonumber
\end{align}
These are called the {\em transition operators} for the clock.
Note that on the LHS of this expression the subscript labels the operator,
whereas on the RHS the subscript indicates which qubits the operator
acts on nontrivially. Writing $\Pi_{\text{unary}}^{(m)}$ for the
projector onto the clock subspace, we have 
\begin{align}
	\left(\iii \otimes \Pi_{\text{unary}}^{(m)}\right)
			h_{t,t + 1}^{\text{u}}(U)
			\left(\iii \otimes \Pi_{\text{unary}}^{(m)}\right)
	 =  &\frac{1}{2}\left(
			\iii \otimes |t\rangle\langle t|_{\text{u}} + \iii \otimes |t + 1\rangle\langle t + 1|_{\text{u}}
		\right) \label{eq:kitaev_in_subspace}\\
	- & \frac{1}{2} \left( U^{\dagger}\otimes|t\rangle\langle t + 1|_{\text{u}}
			+ U\otimes|t + 1\rangle\langle t|_{\text{u}}\right), \nonumber
\end{align}
which should be compared with \eqref{eq:Feynman's Ham}. Kitaev's Hamiltonian
is 
\begin{equation}
H_{\text{K}}(\mathcal{C}) = \iii \otimes  H_{\text{unary}}^{(m)} + \sum_{t = 0}^{m-1}h_{t,t + 1}^{\text{u}}(W_{t})\label{eq:KitaevHam}
\end{equation}
Since both terms in \eqref{eq:KitaevHam} are positive semidefinite,
every state in its zero energy groundspace is a zero eigenvector of
the first term. Within the nullspace of $\iii \otimes  H_{\text{unary}}^{(m)}$,
$H_{K}(\mathcal{C})$ acts in the same way as $H_{\text{F}}(\mathcal{C})$.
Thus, every state in the groundspace of \eqref{eq:KitaevHam} is a history
state 
\[
	\frac{1}{\sqrt{m + 1}}
	\left(
			|\phi\rangle|0\rangle_{\text{u}} + \sum_{t = 1}^{m}W_{t-1}W_{t-2}\dots W_{0}|\phi\rangle|t\rangle_{\text{u}}
	\right)
\]
 for some $n$-qubit state $|\phi\rangle.$ Furthermore, the Hamiltonian
$H_{\text{K}}(\mathcal{C})$ is $5$-local since each unitary $W_{j}$
is either one- or two-local. This circuit-to-Hamiltonian mapping was
used by Kitaev in his proof that $5$-local Hamiltonian is QMA-complete \cite{Kitaev:2002:CQC:863284}.

Let us now view Kitaev's clock construction as a mathematical object
of independent interest. In our view, the clock construction is specified by the
Hamiltonian $H_{\text{unary}}^{(m)}$ along with the transition operators $h_{t,t + 1}^{\text{u}}(U)$
which act as in equation \eqref{eq:kitaev_in_subspace}. 
This is the set of operators which Kitaev used to instantiate Feynman's 
Hamiltonian as a $5$-local Hamiltonian acting on a system of qubits. 
In Kitaev's construction the Hamiltonian $H_{\text{unary}}^{(m)}$
is a sum of 2-local projectors and the transition operators are $(k + 3)$-local
projectors where $k$ is the locality of the unitary $U$. In contrast,
in this Section we introduce a clock construction where the clock states are
groundstates of a Hamiltonian that is a sum of 3-local projectors
and where the transition operators are $\left(k + 2\right)$-local projectors.
If we na{\"{i}}vely replace Kitaev's clock construction with ours,  it allows us to instantiate Feynman's Hamiltonian as a $4$-local Hamiltonian. This is not good enough for our purposes since we are interested in quantum $3$-SAT where the projector terms are $3$-local. The circuit-to-Hamiltonian mapping that we define in this paper is not based on Feynman's Hamiltonian, and in particular only ever uses the transition operators $h_{t,t + 1}^{\text{u}}(U)$ for single-qubit unitaries $U$. Nevertheless we are able to encode a computation consisting of $1$- and $2$-qubit gates in the ground state of the resulting Hamiltonian. This nontrivial fact is one of the main contributions of our paper, and is the subject of Sections \ref{sec:Hamiltonians-with-two} and \ref{sec:Circuit-to-Hamiltonian-mapping}.


\subsection{Summary of the new clock construction \label{sub:Summary-of-the}}

Our construction, including explicit expressions for all the operators
mentioned below, is presented in full detail in Section \ref{sec:Details-of-the}.
Here we summarize its features. As discussed above, the crucial difference between our clock construction and Kitaev's is the locality of the transition operators. 

For any $N\in\{2,3,\dots\}$, we present a Hamiltonian 
\[
H_{\text{clock}}^{(N)}
\]
 which acts on the Hilbert space 
\begin{equation}
\mathcal{H}_{\text{clock}}^{(N)} = \left(\mathbb{C}^{2}\right)^{7N-3}\label{eq:clockspace}
\end{equation}
 of $7N-3$ qubits and which is a sum of 3-local projectors from the
set $\mathcal{{P}}$ given in Definition \ref{Delta_proj}. The zero
energy groundspace of $H_{\text{clock}}^{(N)}$ is spanned by orthonormal
states 
\[
|C_{i}\rangle,\quad i = 1,\dots,N.
\]
Now let $\mathcal{H}_{\text{comp}}$ be a computational register containing
some (arbitrary) number of qubits, and let $U$ be a unitary acting
on this register. We exhibit projectors 
\[
	h_{i,i + 1}(U)
\]
 for $i = 1,\dots,N-1,$ which act on 
\[
	\mathcal{H}_{\text{comp}}\otimes\mathcal{H}_{\text{clock }}^{(N)}.
\]
These are the transition operators \eqref{eq:kitaev_in_subspace}
for the clock. Specifically, they satisfy
\begin{align}
	\left(\iii \otimes \Pi_{\text{clock}}^{(N)}\right)h_{i,i + 1}(U)\left(\iii \otimes \Pi_{\text{clock}}^{(N)}\right)
		 = & \frac{1}{8}\left(
		 		\iii \otimes |C_{i}\rangle\langle 	C_{i}| + \iii \otimes |C_{i + 1}\rangle\langle C_{i + 1}| 
		 	\right) \nonumber\\
		-& \frac{1}{8}\left( 
			U^{\dagger}\otimes|C_{i}\rangle\langle C_{i + 1}|
			+ U\otimes|C_{i + 1}\rangle\langle C_{i}|
		\right),
	 \label{eq:hi_iplus1}
\end{align}
where $\Pi_{\text{clock}}^{(N)} = \sum_{i = 1}^{N}|C_{i}\rangle\langle C_{i}|$
projects onto the clock subspace. Comparing these operators with Kitaev's,
we see that the prefactor of $\frac{1}{2}$ in \eqref{eq:kitaev_in_subspace}
has been replaced with $\frac{1}{8}$ in \eqref{eq:hi_iplus1}. For
our purposes only the nullspaces of these operators are relevant,
and this prefactor does not affect the null space. 

A much more important difference is the locality. In our case, the
operator $h_{i,i + 1}(U)$ is a $(k + 2)$-local projector where $k$
is the locality of the unitary $U$. When $U$ is the identity the
projector $h_{i,i + 1}(1)$ acts nontrivially only on two qubits of
$\mathcal{H}_{\text{clock }}^{(N)}$ and we write 
\[
	h_{i,i + 1}\doteq h_{i,i + 1}(\iii).
\]
Thus, for a single-qubit unitary, $h_{i,i + 1}(U)$ is only 3-local
(in contrast with Kitaev's clock where $h_{i,i + 1}^{\text{u}}(U)$
would be $4$-local). Our circuit-to-Hamiltonian mapping, presented
in Sections \ref{sec:Hamiltonians-with-two} and \ref{sec:Circuit-to-Hamiltonian-mapping}
exploits this feature (it is partly inspired by the railroad switch idea
from \cite{railroad_switch}).

Finally, we also exhibit 1-local (single-qubit) projectors 
\begin{align}
		C_{\geq i}\qquad\text{and}\qquad C_{\leq i} \label{eq:CIS}
\end{align}
 for $i = 1,\dots,N$, whose role is to ``pick out'' clock states $|C_{j}\rangle$
with $j\geq i$ or $j\leq i$ respectively. They act on the Hilbert space \eqref{eq:clockspace}
of the clock and satisfy
\begin{align}
	\Pi_{\text{clock}}^{(N)}C_{\geq i}\Pi_{\text{clock}}^{(N)} &  = \frac{1}{2}|C_{i}\rangle\langle C_{i}| + \sum_{i<j\leq N}|C_{j}\rangle\langle C_{j}|\label{eq:1localproj}\\
\Pi_{\text{clock}}^{(N)}C_{\leq i}\Pi_{\text{clock}}^{(N)} &  = \sum_{1\leq j<i}|C_{j}\rangle\langle C_{j}| + \frac{1}{2}|C_{i}\rangle\langle C_{i}|\label{eq:1localproj2}
\end{align}
with the understanding that when $i = 1$ the first term in \eqref{eq:1localproj2}
is zero and when $i = N$ the second term in \eqref{eq:1localproj} is
zero. As in \eqref{eq:hi_iplus1}, only the nullspaces of the operators on the RHS of \eqref{eq:1localproj} and \eqref{eq:1localproj2} are important for our purposes. In particular, it is not significant that the $|C_{i}\rangle\langle C_{i}|$
terms have different prefactors, since the (positive) value of these
coefficients do not affect the nullspace.


\subsection{Details of the new clock construction\label{sec:Details-of-the}}

Here we present the details of the clock construction described in
Section \ref{sub:Summary-of-the}. The reader may safely skip this
Section on a first read through this paper. Subsequent Sections can
be understood using the summary presented in Section \ref{sub:Summary-of-the}. 

Our goal is to define $H_{\text{clock}}^{(N)}$ along
with the operators $C_{\geq i},C_{\leq i},h_{i,i + 1}(U)$ satisfying
the properties specified in Section \ref{sub:Summary-of-the}.
We present our construction in two steps. First, we define a Hamiltonian $H_{\text{triplet}}^{(2N)}$ 
that we call the ``triplet'' Hamiltonian. It already has many (but not all\footnote{We can't exhibit the type of operators in \eqref{eq:CIS}.}) of the properties that we desire
for our clock construction. 
In the second step, we construct $H_{\text{clock}}^{(N)}$ from the ``triplet'' Hamiltonian by appending another register and adding terms which act between the two registers. 

Let $N\in\{2,3,4,\dots\}$. We now define the triplet Hamiltonian $H_{\text{triplet}}^{(2N)}$
which is a sum of 3-local projectors and acts on the Hilbert space of $3(2N-1)$
qubits: 
\begin{equation}
	H_{\text{triplet}}^{(2N)} = H_{1} + H_{2} + H_{3},
	\label{eq:H_triplet}
\end{equation}
 where 
\begin{align}
	H_{1} &  =  \sum_{i = 0}^{2N-2}\left(|111\rangle\langle 111|
			 +  |010\rangle\langle 010| + |001\rangle\langle 001|
	 		 +  |\theta^-\rangle\langle\theta^-|\right)_{3i + 1,3i + 2,3i + 3}\label{eq:H1}\\
 &  + |000\rangle\langle 000|_{123}
  	 +  \left(|100\rangle\langle 100|
   	 +  |011\rangle\langle 011|\right)_{3(2N-2) + 1,3(2N-2) + 2,3(2N-2) + 3},\nonumber \\
	H_{2} &  =  \sum_{i = 0}^{2N-3}\left(|10\rangle\langle 10|
		 +  |01\rangle\langle 01|\right)_{3i + 2,3i + 3}
				\otimes\left(|1\rangle\langle 1|_{3i + 4}
		 +  |1\rangle\langle 1|_{3i + 5}
		 +  |1\rangle\langle 1|_{3i + 6}\right) \label{eq:H2}\\
 &  + \sum_{i = 0}^{2N-3}|00\rangle\langle 00|_{3i + 1,3i + 2}
 		\otimes\left(|1\rangle\langle 1|_{3i + 4}
 		 +  |1\rangle\langle 1|_{3i + 5} + |1\rangle\langle 1|_{3i + 6}\right), \nonumber \\
	H_{3} &  = \sum_{i = 0}^{2N-3}|\gamma\rangle\langle\gamma|_{3i + 3,3i + 4,3i + 5}\label{eq:H3}
\end{align}
act on triplets of qubits as in Figure~\ref{fig:H123}, with $\ket{\theta^\pm}$ and $\ket{\gamma}$ given by
\begin{align*}
	|\theta^\pm\rangle &  = \frac{1}{\sqrt{2}}\left(|100\rangle\pm|011\rangle\right),
	\qquad
	|\gamma\rangle = \frac{1}{\sqrt{3}}|100\rangle-\sqrt{\frac{2}{3}}|011\rangle.
\end{align*}

\begin{figure}
	\includegraphics[width=8cm]{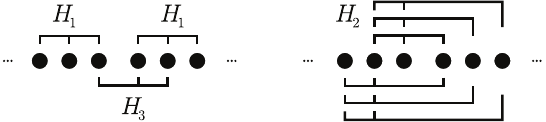}
	\caption{The terms $H_1$, $H_2$ and $H_3$ in the Hamiltonian $H_{\text{triplet}}^{(2N)}$ \eqref{eq:H_triplet}
	involve consecutive triplets of qubits.\label{fig:H123}}
\end{figure}

We now construct $2N$ orthonormal basis states which span the ground space of $H_{\text{triplet}}^{(2N)}$. 
First, consider $H_{1}.$ Observe that the groundspace of $H_{1}$
is spanned by the set of all products of three qubit states of the
form 
\[
	|\chi_{1}\rangle|\chi_{2}\rangle\dots|\chi_{2N-1}\rangle,
\]
 where 
\begin{align*}
	|\chi_{1}\rangle & \in\{|110\rangle,|101\rangle,|\theta^+ \rangle\},\\
	|\chi_{j}\rangle & \in\{|000\rangle,|110\rangle,|101\rangle,|\theta^+ \rangle\}
			\; \text{for}\; j = 2,\dots,2N-2,\\
	|\chi_{2N-1}\rangle & \in\{|000\rangle,|110\rangle,|101\rangle\}.
\end{align*}
 Let us now consider the action of $H_{2}$ on the groundspace of
$H_{1}$. The terms in $H_{2}$ assign an energy penalty to certain
products of three qubit states. The first term penalizes any state
where $|\chi_{i}\rangle\in\{|110\rangle,|101\rangle\}$ but $|\chi_{i + 1}\rangle\neq|000\rangle.$
The second term penalizes states where $|\chi_{i}\rangle = |000\rangle$
but $|\chi_{i + 1}\rangle\neq|000\rangle.$ The projector onto the
zero energy groundspace of $H_{1} + H_{2}$ is therefore given by 
\begin{equation}
	\Pi_{12} = 
	\sum_{i = 1}^{4N-2}|\hat{i}\rangle\langle\hat{i}|
	 + \sum_{k = 1}^{2N-2}|\tilde{k}\rangle\langle\tilde{k}|,\label{eq:Pi_12}
\end{equation}
 where 
\begin{align}
	|\widehat{1}\rangle &  = |110\rangle|000\rangle|000\rangle
			\dots |000\rangle,\label{eq:1_hat}\\
	|\widehat{2}\rangle &  = |101\rangle|000\rangle|000\rangle
			\dots |000\rangle\nonumber\\
	|\widehat{3}\rangle &  = |\theta^+\rangle\;|110\rangle|000\rangle
		\dots |000\rangle,\nonumber\\
	|\widehat{4}\rangle &  = |\theta^+\rangle\;|101\rangle|000\rangle
		\dots |000\rangle,\nonumber\\
 	& \;\,\vdots\nonumber\\
	|\widehat{4N-3}\rangle &  = |\theta^+ \rangle\,|\theta^+\rangle\;
			\dots \; |\theta^+\rangle\; |110\rangle,\nonumber\\
	|\widehat{4N-2}\rangle &  = |\theta^+\rangle\;|\theta^+\rangle\;
			\dots \; |\theta^+\rangle\;|101\rangle
	\label{eq:2R-2hat}
\end{align}
and
\begin{equation}
	|\tilde{k}\rangle = |\theta^+ \rangle^{\otimes k}|000\rangle^{\otimes2N-1-k} \label{eq:ktilde}
\end{equation}
for $k = 1,\dots,2N-2$.
Note that the states $\ket{\widehat{i}}$ each contain an ``active'' triplet $|110\rangle$
or $|101\rangle$ separating $|\theta^+ \rangle$ to the left and $|000\rangle$
to the right. In contrast, the states $|\tilde{k}\rangle$ do not have an active triplet separating the two
regions.

The role of $H_{3}$ is twofold. 
It assigns an energy penalty to each of the states $\{|\tilde{k}\rangle\}$
and to certain linear combinations of the states $\{|\hat{i}\rangle\}$. This can be seen by looking at the restriction of $H_{3}$ to the groundspace of $H_{1} + H_{2}$,
\begin{equation}
\Pi_{12}H_{3}\Pi_{12} = \frac{1}{3}\sum_{i = 2}^{2N-1}\left(|\widehat{2i-2}\rangle-|\widehat{2i-1}\rangle\right)\left(\langle\widehat{2i-2}|-\langle\widehat{2i-1}|\right) + \frac{1}{6}\sum_{k = 1}^{2N-2}|\tilde{k}\rangle\langle\tilde{k}|.\label{eq:H3_insubspace}
\end{equation}
We are now able to fully characterize the zero energy groundspace of
$H_{\text{triplet}}^{(2N)}$. It is spanned by the $2N$ states
\begin{align}
	|c_{1}\rangle &  = |\widehat{1}\rangle,\label{eq:c1}\\
	|c_{i}\rangle & 
		 = \frac{1}{\sqrt{2}}\left(|\widehat{2i-2}\rangle + |\widehat{2i-1}\rangle\right)
		\; \text{for} \; i = 2,\dots,2N-1,\label{eq:ci}\\
	|c_{2N}\rangle &  = |\widehat{4N-2}\rangle.\label{eq:cR}
\end{align}

We now consider some local operators and their actions in the zero
energy subspace of $H_{\text{triplet}}^{(2N)}.$ We will see that 
they have properties similar to those described in Section \ref{sub:Summary-of-the}.

Define $2$-local projectors
\begin{align}
	q_{1} &  = |10\rangle\langle 10|_{23},\nonumber \\
	q_{i} &  = |01\rangle\langle 01|_{3i-4,3i-3}\; \text{for}\; i = 2,\dots,2N,\label{eq:q_j}
\end{align}
and note that, projected to the groundspace of $H_{\text{triplet}}^{(2N)}$, we have
\begin{equation}
	\Pi^{(2N)}q_{i}\Pi^{(2N)} = \alpha_{i}|c_{i}\rangle\langle c_{i}|, \label{eq:action_of_qi}
\end{equation}
where $\Pi^{(2N)} = \sum_{i = 1}^{2N}|c_{i}\rangle\langle c_{i}|$ and
\begin{equation}
	\alpha_{i} = \begin{cases}
	1, & \text{ if }i\in\{1,2N\}\\
	\frac{1}{2}, & \text{ otherwise. }
	\end{cases}\label{eq:alpha_i}
\end{equation}
Now adjoin another register so the Hilbert space is 
\[
	\mathcal{H}_{\text{comp}}\otimes\left(\mathbb{C}^{2}\right)^{6N-3}
\]
and let $U$ be a unitary acting on $\mathcal{H}_{\text{comp}}$.
Define 
\begin{align}
	g_{i,i + 1}(U) =& 
		\frac{1}{\alpha_{i} + \alpha_{i + 1}}
		\left(
				\alpha_{i + 1}\,\iii \otimes |10\rangle\langle 10|_{3i-1,3i}
		 	+ \alpha_{i}\,\iii \otimes |01\rangle\langle 01|_{3i-1,3i}
		 \right) \label{eq:g(U)}\\
		-& \frac{\sqrt{\alpha_{i + 1}\alpha_{i}}}{\alpha_{i} + \alpha_{i + 1}} \left(
		 U^{\dagger}\otimes|10\rangle\langle 01|_{3i-1,3i}
		+ U\otimes|01\rangle\langle 10|_{3i-1,3i}
		\right) \nonumber 
\end{align}
for $i = 1,\dots,2N-1$. Note that $g_{i,i + 1}(U)$ is a $(k + 2)$-local
projector if $U$ is a $k$-qubit unitary, and furthermore 
\begin{align}
	\left(\iii \otimes \Pi^{(2N)}\right)g_{i,i + 1}(U)\left(\iii \otimes \Pi^{(2N)}\right) & 	
		 = \beta_{i}\left(
		 		\iii \otimes |c_{i}\rangle\langle c_{i}| 
		 		+ \iii \otimes |c_{i + 1}\rangle\langle c_{i + 1}| 
		 	\right)\nonumber \\
 	& - \beta_i \left(
 			U^{\dagger}\otimes|c_{i}\rangle\langle c_{i + 1}|
  		+ U\otimes|c_{i + 1}\rangle\langle c_{i}|
  \right),
  \label{eq:g_in_subspace}
\end{align}
 where 
\begin{equation}
	\beta_{i} = \frac{\alpha_{i + 1}\alpha_{i}}{\alpha_{i} + \alpha_{i + 1}} 
		= \begin{cases}
			\frac{1}{3}, & \text{ if }i\in\{1,2N-1\},\\
			\frac{1}{4}, & \text{ otherwise.}
	\end{cases}\label{eq:beta_i}
\end{equation}
Equation \eqref{eq:g_in_subspace} should be compared to \eqref{eq:hi_iplus1}.

Now we are ready to define $H_{\text{clock}}^{(N)}$ and the associated
operators which make up our clock construction. 
Consider the Hilbert space 
\[
	\left(\mathbb{C}^{2}\right)^{6N-3}\otimes\left(\mathbb{C}^{2}\right)^{N}
\]
and the Hamiltonian
\begin{equation}
	H_{\text{clock}}^{(N)}
		 =  H_{\text{triplet}}^{(2N)}\otimes \iii  + \iii\otimes 	H_{\text{unary}}^{(N)}
		 +  H_{\text{sync}} 
		 +  \sum_{i = 1}^{N}g_{2i-1,2i}(\sigma_{i}^{x}),
	\label{eq:H_N_clock}
\end{equation}
 where $\sigma_{i}^{x}$ is the $\sigma^{x}$ operator acting on the
$i$th qubit of the second register. Observe that each operator $g_{2i-1,2i}(\sigma_{i}^{x})$
(defined through \eqref{eq:g(U)}) is a 3-local projector which
acts on qubits $6i-4$ and $6i-5$ of the first register and qubit
$i$ of the second register. $H_{\text{unary}}^{(N)}$ acts on the
second register as 
\begin{equation}
	H_{\text{unary}}^{(N)} = \sum_{i = 1}^{N-1}|01\rangle\langle 01|_{i,i + 1}\label{eq:H_unary_N}
\end{equation}
and its ground state subspace is spanned by the unary states 
\begin{align}
	\ket{t}_{\text{u}}  =  \ket{\underbrace{1\cdots 1}_{t}\underbrace{0\cdots 0}_{N-t}}. 
	\label{eq:unarystates}
\end{align}

The Hamiltonian $H_{\text{sync}}$ acts on both registers as
\begin{equation}
	H_{\text{sync}}
	 =  q_{1}\otimes|1\rangle\langle 1|_{1}
	 +  \sum_{j = 2}^{2N-1}q_{j}\otimes\left(|0\rangle\langle 0|_{\left\lfloor \frac{j}{2}\right\rfloor }
	 +  |1\rangle\langle 1|_{\left\lfloor \frac{j}{2}\right\rfloor  + 1}\right)
	 +  q_{2N}\otimes|0\rangle\langle 0|_{N}.\label{eq:H_sync}
\end{equation}
 where $q_{j}$ is defined in \eqref{eq:q_j} and $\left\lfloor \cdot\right\rfloor $
is the floor function. Here the tensor product separates the first
$6N-3$ qubits from the last $N$ qubits. 

The zero energy groundspace of $H_{\text{triplet}}^{(2N)}\otimes \iii + \iii \otimes H_{\text{unary}}^{(N)} + H_{\text{sync}}$
is spanned by the states 
\begin{align}
	|e_{j}\rangle
	  = |c_{j}\rangle \left| \left\lfloor j/2 \right\rfloor \right\rangle_\text{u}
	 	 \label{eq:coupled_clocks}
\end{align}
 for $j = 1,\dots,2N$, where $|c_{j}\rangle$ are the ground states of
$H_{\text{triplet}}^{(2N)}$ as defined in equations \eqref{eq:c1}-\eqref{eq:cR}.
Let $\Pi_{e}^{(2N)}$ be the projector onto the subspace spanned by
$\{|e_{j}\rangle\}$. The last term in \eqref{eq:H_N_clock}, the operator $\sum_{i}g_{2i-1,2i}(\sigma_{i}^{x})$, acts within this space as 
\begin{equation}
	\Pi_{\text{e}}^{(2N)}
			\left(\sum_{i = 1}^{N}g_{2i-1,2i}(\sigma_{i}^{x})\right)
			\Pi_{e}^{(2N)} 
	= \sum_{i = 1}^{N}\beta_{2i-1}
		\left(|e_{2i-1}\rangle-|e_{2i}\rangle\right)
		\left(\langle e_{2i-1}|-\langle e_{2i}|\right)
	\label{eq:action_of_g_x}
\end{equation}
 where $\beta_{i}$ is defined in \eqref{eq:beta_i}. Taking
linear combinations of the states \eqref{eq:coupled_clocks} which are
zero eigenvectors of \eqref{eq:action_of_g_x}, we finally get a spanning set
of zero energy groundstates for $H_{\text{clock}}^{(N)}$: 
\begin{align}
	|C_{i}\rangle
		& = \frac{1}{\sqrt{2}} \big(
		 	|c_{2i-1}\rangle 
		 			|  i-1\rangle_\text{u}
			+ |c_{2i}\rangle 
		 			| i\rangle_\text{u}
		\big)\\
		 & = \frac{1}{\sqrt{2}} \big(
		 	|c_{2i-1}\rangle 
		 			|  \underbrace{1\cdots 1}_{i-1}
	 						\underbrace{0\cdots 0}_{N-i+1} 	\rangle
			+ |c_{2i}\rangle 
		 			|  \underbrace{1\cdots 1}_{i}
	 						\underbrace{0\cdots 0}_{N-i} 	\rangle 
	 	\big)
	\label{eq:C_i}
\end{align}
for $i = 1,\dots,N.$ 
As an example, we expand these states in the computational basis for $N=4$ in Figure~\ref{fig:clockstates} in Appendix~\ref{sec:clockreference}.

We now define the operators 
\begin{align}
	C_{\leq i} &  = \iii \otimes |0\rangle\langle 0|_{i}, \label{eq:C_leq}\\
	C_{\geq i} &  = \iii \otimes |1\rangle\langle 1|_{i}, \label{eq:C_geq}
\end{align}
which act on the $i$th qubit of the second register, for $i = 1,\dots,N$.
Using \eqref{eq:C_i} it is easy to verify that these operators
satisfy \eqref{eq:1localproj} and \eqref{eq:1localproj2} as claimed in Section \ref{sub:Summary-of-the}. 

Finally, we are ready to define the transition operators $h_{i,i + 1}(U)$ and to
show that they satisfy \eqref{eq:hi_iplus1}. We consider a
Hilbert space 
\[
	\mathcal{H}_{\text{comp}} \otimes \mathcal{H}_{\text{clock}}^{(N)} =
	\mathcal{H}_{\text{comp}} \otimes \left(\mathbb{C}^{2}\right)^{6N-3}\otimes\left(\mathbb{C}^{2}\right)^{N}
\]
and a unitary $U$ which acts on $\mathcal{H}_{\text{comp}}$. Define
\begin{equation}
h_{i,i + 1}(U) = g_{2i,2i + 1}(U)\otimes \iii \label{eq:h_intermsof_g}
\end{equation}
 for $i = 1,\dots,N-1.$ Here $g_{2i,2i + 1}(U)$ acts on the computational
register and two out of the $6N-3$ qubits in the second register.
Note that $h_{i,i + 1}(U)$ is a $(k + 2)$-local projector when $U$
is a $k$-qubit unitary. Using equations \eqref{eq:h_intermsof_g},
\eqref{eq:C_i}, \eqref{eq:g_in_subspace} and the fact that $\beta_{2i} = \frac{1}{4}$
for all $i = 1,\dots,N-1$ we confirm \eqref{eq:hi_iplus1}:
\begin{align*}
	\left(\iii \otimes\Pi_{\text{clock}}^{(N)}\right)
			h_{i,i + 1}(U)
			\left(\iii \otimes\Pi_{\text{clock}}^{(N)}\right)
	 &= \frac{1}{8}\left(
	 		\iii \otimes |C_{i}\rangle\langle C_{i}| 
	 		 + \iii \otimes|C_{i + 1}\rangle\langle C_{i + 1}|
	 		\right) \\
	&- \frac{1}{8}\left(
			U^{\dagger}\otimes|C_{i}\rangle\langle C_{i + 1}|
			+U\otimes|C_{i + 1}\rangle\langle C_{i}|
		\right).
\end{align*}
 We can also write an explicit expression for $h_{i,i + 1}(U)$ using
\eqref{eq:g(U)} and the fact that $\alpha_{2i} = \alpha_{2i + 1} = \frac{1}{2}$
for all $i = 1,\dots,N-1$:
\begin{align}
	h_{i,i + 1}(U) &  = 
		\frac{1}{2}\left(
			\iii \otimes |10\rangle\langle 10|_{6i-1,6i}\otimes \iii + \iii \otimes |01\rangle\langle 01|_{6i-1,6i}\otimes \iii \right)
					\nonumber \\
 & - \frac{1}{2}\left( 
 	U^{\dagger}\otimes|10\rangle\langle 01|_{6i-1,6i}\otimes \iii + U\otimes|01\rangle\langle 10|_{6i-1,6i}\otimes \iii
  \right).
  \label{eq:h_i_explicit}
\end{align}

Finally, note that $H_{\text{clock}}^{(N)}$ is a sum of 3-local projectors.
We now show that each of the projectors in the sum is from the set
$\mathcal{P}$ given in Definition \ref{Delta_proj}. Looking at equations
\eqref{eq:H1}-\eqref{eq:H3} we see that $H_{\text{triplet}}^{(2N)}$ is
a sum of projectors which are diagonal in the computational basis
(and which satisfy condition 1 in Definition \ref{Delta_proj}) along
with projectors of the form $|\theta^-\rangle\langle\theta^-|$ (which
satisfy condition 1 in Definition \ref{Delta_proj}) and projectors
of the form $|\gamma\rangle\langle\gamma|$ (which satisfy condition
2). $H_{\text{unary}}^{(N)}$ and $H_{\text{sync}}$ are sums of
classical projectors which satisfy condition 1. Now consider the terms
$g_{2i-1,2i}(\sigma_{i}^{x})$ for $i = 1,\dots,N$. Looking at equations
\eqref{eq:alpha_i} and \eqref{eq:g(U)} we see that, for $i\in\{2,\dots,N-1\}$
each of these terms is a projector which satisfies condition 1. We
consider the terms with $i = 1$ and $i = N$ separately. First, look at
$g_{1,2}(\sigma_{1}^{x})$ which acts on qubits $2$ and $3$ of the
first register and qubit $1$ of the second register as 
\begin{align*}
	 & \frac{1}{3} \left(
	 |10\rangle\langle 10|\otimes \iii 
	  + 2\, |01\rangle\langle 01| \otimes \iii 
	 	- \sqrt{2}\, |10 \rangle\langle 01|\otimes\sigma_{x}
	 	- \sqrt{2}\, |01 \rangle\langle 10|\otimes\sigma_{x} \right)\\
	 =  & \frac{1}{3} \left(
	 	|10\rangle\langle 10|
	 +  2\, |01\rangle\langle 01|
		-\sqrt{2}\, |10\rangle\langle 01|
		-\sqrt{2}\, |01\rangle\langle 10|
	\right)
	\otimes | + \rangle\langle  + | \\
  + &
 \frac{1}{3} \left(
 		|10\rangle\langle 10|
 		 + 2\,|01\rangle\langle 01|
 		 + \sqrt{2}\,|10\rangle\langle 01|
 		 + \sqrt{2}\,|01\rangle\langle 10|
 		 \right)
 			\otimes|-\rangle\langle -|
\end{align*}
 where $| + \rangle$ and $|-\rangle$ are the eigenstates of the Pauli
$\sigma^{x}$ operator. It is not hard to see that each of the two
terms on the RHS of this equation is a projector which satisfies condition
2 from Definition \ref{Delta_proj}. A similar decomposition can be
used to write $g_{2N-1,2N}(\sigma_{N}^{x})$ as a sum of two projectors
satisfying condition 2.

This completes our description of the clock construction with the properties outlined in Section \ref{sub:Summary-of-the}.


\section{Hamiltonians acting on two clock registers\label{sec:Hamiltonians-with-two}}

Here we develop the main ideas behind our circuit-to-Hamiltonian mapping,
using the new clock construction described in the previous section.
A key feature of our approach is that we use two clock registers,
with Hilbert space 
\begin{equation}
	\mathcal{H}_{\text{clock}}^{(N)}\otimes\mathcal{H}_{\text{clock}}^{(N)}.\label{eq:2d_space}
\end{equation}
Let's consider some local operators which act on this space. The Hamiltonians

\[
	\iii \otimes  H_{\text{clock}}^{(N)}\qquad\text{and}\qquad H_{\text{clock}}^{(N)} \otimes \iii 
\]
 are both sums of 3-local projectors as discussed in the previous
section. Since $C_{\leq i}$ and $C_{\geq i}$ are 1-local projectors,
we can form 2-local projectors by taking tensor products, e.g., 
\[
	C_{\leq i}\otimes C_{\geq j}.
\]
 Similarly, since the operators $h_{k,k + 1}$ are 2-local projectors,
terms such as 
\[
	h_{k,k + 1}\otimes C_{\leq i}
\]
 are 3-local projectors. For convenience and to ease notation later
on, we define the following sum of such terms 
\begin{align}
	S^{(k,k + 2)} &= C_{\leq k}\otimes C_{\geq(k + 2)} 
	 + h_{k,k + 1}\otimes C_{\leq(k + 1)}
	 + h_{(k + 1),(k + 2)}\otimes C_{\geq(k + 1)} 
	  \label{eq:S_k_k_2}\\
	 &+ C_{\geq(k + 2)}\otimes C_{\leq k}
	 + C_{\leq(k + 1)}\otimes h_{k,k + 1}
  + C_{\geq(k + 1)}\otimes h_{(k + 1),(k + 2)} \nonumber
\end{align}
for $k = 1,\dots,N-2.$ 

We begin by looking at a simple Hamiltonian which acts in the Hilbert
space \eqref{eq:2d_space}; this example introduces some notation and
conventions that we use later on. We then consider two examples where
the Hilbert space \eqref{eq:2d_space} is tensored with a computational
register. These examples contain the essential ideas behind our proof
in Section \ref{sec:Circuit-to-Hamiltonian-mapping} that quantum $3$-SAT is QMA$_{1}$-hard. 


\subsection{Warm up example}

As a warm-up, consider the following Hamiltonian acting on the space
\eqref{eq:2d_space} with $N = 9$: 
\begin{equation}
\iii \otimes  H_{\text{clock}}^{(9)} + H_{\text{clock}}^{(9)} \otimes \iii  + S^{(4,6)},\label{eq:example_eqn}
\end{equation}
with $S^{(4,6)}$ given by \eqref{eq:S_k_k_2}.
We will see how the zero energy groundspace of this operator can be
represented pictorially. 

Recall (from Section \ref{sub:Summary-of-the}) that $H_{\text{clock}}^{(9)}$
has $9$ orthonormal zero energy states $|C_{i}\rangle$ for $i = 1,\dots,9$.
The first two terms of \eqref{eq:example_eqn} 
\begin{equation}
\iii \otimes  H_{\text{clock}}^{(9)} + H_{\text{clock}}^{(9)} \otimes \iii \label{eq:9dim}
\end{equation}
 therefore have 81 zero energy ground states which we choose to represent
as a set of vertices arranged in a 2D grid, as shown in Figure \ref{fig:N = 00003D9_example}(a).
We adopt the convention that the vertex in the top left corner has
coordinates $(i,j) = (1,1)$, the $i$ coordinate increases moving to
the right and the $j$ coordinate increases moving downwards. The
vertex with coordinate $(i,j)$ is associated with the groundstate
$|C_{i}\rangle|C_{j}\rangle.$ 

We add $S^{(4,6)}$ to \eqref{eq:9dim} a few terms at a time. First
look at 
\[
\iii \otimes  H_{\text{clock}}^{(9)} + H_{\text{clock}}^{(9)} \otimes \iii  + C{}_{\leq4}\otimes C_{\geq6} + C_{\geq6}\otimes C_{\leq4}
\]
 which is just the first two terms of $S^{(4,6)}$ added to \eqref{eq:9dim}.
Using the expressions \eqref{eq:1localproj} and \eqref{eq:1localproj2}
we see that adding this term assigns an energy penalty to all the
states $|C_{i}\rangle|C_{j}\rangle$ with either $i\leq4$ and $j\geq6$
or $i\geq6$ and $j\leq4$. Eliminating the corresponding vertices
from Figure \ref{fig:N = 00003D9_example}(a) we get Figure \ref{fig:N = 00003D9_example}(b). 

Now look at the next term which is $h_{4,5}\otimes C_{\leq5}$. Using
equations \eqref{eq:hi_iplus1} and \eqref{eq:C_leq} we get 
\begin{align*}
	\left(\Pi_{\text{clock}}^{(9)}\otimes\Pi_{\text{clock}}^{(9)}\right)
		&\left( h_{4,5}\otimes C_{\leq5} \right)
	\left(\Pi_{\text{clock}}^{(9)}\otimes\Pi_{\text{clock}}^{(9)}\right) \\
		&= \frac{1}{8}\left(|C_{4}\rangle-|C_{5}\rangle\right)\left(\langle C_{4}|-\langle C_{5}|\right)
		 \otimes \sum_{j = 1}^{5}\left(1-\frac{1}{2}\delta_{j,5}\right)|C_{j}\rangle\langle C_{j}|.
\end{align*}
 From this we see that states $|C_{4}\rangle|C_{j}\rangle$ and $|C_{5}\rangle|C_{j}\rangle$
for $j = 1,\dots,5$ are not zero energy states for this term although
their uniform superpositions $\frac{1}{\sqrt{2}}\left(|C_{4}\rangle + |C_{5}\rangle\right)|C_{j}\rangle$
are. We represent the groundspace of 
\[
\iii \otimes  H_{\text{clock}}^{(9)} + H_{\text{clock}}^{(9)} \otimes \iii  + C{}_{\leq4}\otimes C_{\geq6} + C_{\geq6}\otimes C_{\leq4} + h_{4,5}\otimes C_{\leq5}
\]
 as the graph in Figure~\ref{fig:N = 00003D9_example}(c), where now
ground states are in one-to-one correspondence with the\emph{ connected
components }of the graph. The ground state corresponding to a given
connected component $\mathcal{J}$ is the uniform superposition 
\[
\sum_{(i,j)\in\mathcal{J}}|C_{i}\rangle|C_{j}\rangle
\]
(up to normalization).
The next three terms modify the picture in a similar way and the groundspace
of $\iii \otimes  H_{\text{clock}}^{(9)} + H_{\text{clock}}^{(9)} \otimes \iii  + S^{(4,6)}$
is represented as the graph shown in Figure~\ref{fig:N = 00003D9_example}(d). 

\begin{center}
\begin{figure}[H]
\center
\subfigure[]
{
\begin{tikzpicture}[scale = 0.25,vertex/.style = {circle,draw = black,fill = black,inner sep = 0.5pt,minimum size  = 0mm}]
\tikzstyle{every node} = [draw];
\foreach \i in {1,...,9}
{
\foreach \j in {1,...,9}
{
      \draw (\i,\j) node[vertex]{};
}
}
\end{tikzpicture}
}
\hspace{1.5cm}
\subfigure[]
{
\begin{tikzpicture}[scale = 0.25,vertex/.style = {circle,draw = black,fill = black,inner sep = 0.5pt,minimum size  = 0mm}]
\tikzstyle{every node} = [draw];
\foreach \i in {1,...,5}
{
\foreach \j in {5,...,9}
{
      \draw (\i,\j) node[vertex]{};
}
}
\foreach \i in {5,...,9}
{
\foreach \j in {1,...,5}
{
      \draw (\i,\j) node[vertex]{};
}
}
\end{tikzpicture}
}
\hspace{1.5cm}
\subfigure[]
{
\begin{tikzpicture}[scale = 0.25,vertex/.style = {circle,draw = black,fill = black,inner sep = 0.5pt,minimum size  = 0mm}]
\tikzstyle{every node} = [draw];
\foreach \i in {1,...,5}
{
\foreach \j in {5,...,9}
{
      \draw (\i,\j) node[vertex]{};
}
}
\foreach \i in {5,...,9}
{
\foreach \j in {1,...,5}
{
      \draw (\i,\j) node[vertex]{};
}
}
\foreach \i in {5,...,9}
{
\draw (4,\i)--(5,\i);
}
\end{tikzpicture}
}
\hspace{1.5cm}
\subfigure[]
{
\begin{tikzpicture}[scale = 0.25,vertex/.style = {circle,draw = black,fill = black,inner sep = 0.5pt,minimum size  = 0mm}]
\tikzstyle{every node} = [draw];
\foreach \i in {1,...,5}
{
\foreach \j in {5,...,9}
{
      \draw (\i,\j) node[vertex]{};
}
}
\foreach \i in {5,...,9}
{
\foreach \j in {1,...,5}
{
      \draw (\i,\j) node[vertex]{};
}
}
\draw (4,6) --  + (1,0) --  + (1,-1)--  + (2,-1)--  + (2,-2)--  + (1,-2)--  + (1,-1)--  + (0,-1)--  + (0,0);

\foreach \i in {5,...,9}
{
\draw (4,\i)--(5,\i);
\draw (\i,4)--(\i,5);
}
\foreach \i in {1,...,4}
{
\draw (5,\i)--(6,\i);
\draw (\i,5)--(\i,6);
}
\end{tikzpicture}
}

\caption{The groundspaces of 
	(a) $\iii \otimes  H_{\text{clock}}^{(9)} + H_{\text{clock}}^{(9)} \otimes \iii $,
	\;
	(b) $\iii \otimes  H_{\text{clock}}^{(9)} + H_{\text{clock}}^{(9)} \otimes \iii  
				+ C{}_{\leq4}\otimes C_{\geq6} + C_{\geq6}\otimes C_{\leq4}$,
	\;
	(c) $\iii \otimes  H_{\text{clock}}^{(9)} + H_{\text{clock}}^{(9)} \otimes \iii  
				+ C{}_{\leq4}\otimes C_{\geq6} + C_{\geq6}\otimes C_{\leq4} 
				+ h_{4,5}\otimes C_{\leq5}$,
	\;
	and (d) $\iii \otimes  H_{\text{clock}}^{(9)} + H_{\text{clock}}^{(9)} \otimes \iii  + S^{(4,6)}$.
In these graphs each connected component is associated with a ground
state of the Hamiltonian. \label{fig:N = 00003D9_example}}
\end{figure}

\par\end{center}

\subsection{A single-qubit unitary}

Next, consider an example with two clock registers with $N = 6$ and
a computational register containing a single qubit. The Hilbert space
is 
\[
\mathbb{C}^{2}\otimes\mathcal{H}_{\text{clock}}^{(6)}\otimes\mathcal{H}_{\text{clock}}^{(6)}.
\]
Let $U$ be a (single-qubit) unitary acting on the computational qubit
and define
\begin{equation}
	H_{1\text{q}}(U) = \iii \otimes \iii \otimes H_{\text{clock}}^{(6)} 
		+ \iii \otimes  H_{\text{clock}}^{(6)} \otimes \iii  
		+ \iii \otimes  S^{(1,3)} 
		+ \iii \otimes  S^{(4,6)} 
		+ H_{U},\label{eq:H(U)}
\end{equation}
 where 
\begin{equation}
	H_{U} = h_{34}(U) \otimes \iii  + \iii \otimes \iii \otimes h_{34}.
	\label{eq:Hu}
\end{equation}
Here $h_{34}(U)$ acts nontrivially on the computational qubit and
two qubits of first clock register. 

We analyze the groundspace of \eqref{eq:H(U)} in two steps. First,
we represent the groundspace of the sum of the first four terms using
a picture, as in the previous example. Then we consider the action
of $H_{U}$ on this space and obtain the zero energy states for \eqref{eq:H(U)}. 

First, consider 
\[
	\iii \otimes  H_{\text{clock}}^{(6)} + H_{\text{clock}}^{(6)} \otimes \iii  + S^{(1,3)} + S^{(4,6)},
\]
which acts in the space $\mathcal{H}_{\text{clock}}^{(6)}\otimes\mathcal{H}_{\text{clock}}^{(6)}$
and note (using the graphical representation discussed in the previous
example) that its nullspace can be represented as Figure
\ref{fig:1qubit_example}. In the Figure we label vertices of the graph as $(i,j)$ with the top left vertex labeled $(1,1)$, $i$ increasing to the right and $j$ increasing downward.  A ground state is associated with each connected component $\mathcal{K},\mathcal{L},\mathcal{M},\mathcal{N}$
as shown in the Figure, given by 
\begin{equation}
	|\mathcal{S}\rangle = \sum_{(i,j)\in\mathcal{S}}|C_{i}\rangle|C_{j}\rangle,
	\label{eq:S_defn}
\end{equation}
where $\mathcal{S}\in\{\mathcal{K},\mathcal{L},\mathcal{M},\mathcal{N}\}$.
In this paper we will often work with unnormalized states such as these.

\begin{figure}[H]
\center
\begin{tikzpicture}[scale = 0.25,vertex/.style = {circle,draw = black,fill = black,inner sep = 0.5pt,minimum size  = 0mm}]
\foreach \i in {1,...,4}
{
\foreach \j in {-1,...,-4}
{
      \draw (\i,\j) node[vertex]{};
}
}
\foreach \i in {-1,...,-4}
{
\draw (1,\i)--(2,\i);
\draw (3,\i)--(4,\i);
}
\foreach \i in {1,...,4}
{
\draw (\i,-1)--(\i,-2);
\draw (\i,-3)--(\i,-4);
}

\draw (0,0) node[vertex]{}--(1,0) node[vertex]{}--(1,-1) node[vertex]{}--(0,-1) node[vertex]{}--(0,0);
\draw (1,-1)--(2,-1)--(2,-2)--(1,-2)--(1,-1);

\draw (3,-3) --  + (1,0) --  + (1,-1)--  + (2,-1)--  + (2,-2)--  + (1,-2)--  + (1,-1)--  + (0,-1)--  + (0,0);

\draw (4,-5) node[vertex]{};
\draw (5,-4) node[vertex]{};
\draw (5,-5) node[vertex]{};
\draw (6,-4.5) node[]{$\mathcal{N}$};
\draw (5,-0.5) node[]{$\mathcal{L}$};
\draw (-1,-0.5) node[]{$\mathcal{K}$};
\draw (0,-4.5) node[]{$\mathcal{M}$};
\end{tikzpicture}

\caption{The four groundstates of $H_{\text{clock}}^{(6)} \otimes \iii  + \iii \otimes  H_{\text{clock}}^{(6)} + S^{(1,3)} + S^{(4,6)}$
are associated with the four connected components of this graph, which
we label $\mathcal{K},\mathcal{L},\mathcal{M},\mathcal{N}$.\label{fig:1qubit_example}}
\end{figure}
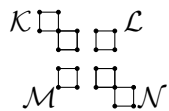

The groundstates of $H_{1\text{q}}(U)$ are superpositions of 
\begin{equation}
	|z\rangle|\mathcal{K\rangle}, \quad
	|z\rangle|\mathcal{L\rangle}, \quad
	|z\rangle|\mathcal{M\rangle}, \quad
	|z\rangle|\mathcal{N}\rangle, \quad z\in\{0,1\}.
	\label{eq:z_K_L_M}
\end{equation}
which have zero energy for $H_{U}$. We solve for them as follows.
First, note that
\begin{equation}
	|z\rangle|\mathcal{K\rangle}, \quad 
	U|z\rangle|\mathcal{L\rangle}, \quad
	|z\rangle|\mathcal{M\rangle}, \quad 
	U|z\rangle|\mathcal{N}\rangle
	\label{eq:z_U_K_L_M}
\end{equation}
for $z\in\{0,1\}$ span the same space as \eqref{eq:z_K_L_M}. This
basis is convenient because $H_{U}$ does not connect states with
$z = 0$ to states with $z = 1$. We evaluate the matrix elements of $H_{U}$
between these unnormalized states using \eqref{eq:hi_iplus1}.
For each $z\in\{0,1\}$, $H_{U}$ acts as a $4\times4$ matrix within
the space spanned by the four states \eqref{eq:z_U_K_L_M} (since it
does not connect states with different $z$). This matrix is the same
for $z = 0$ and $z = 1$ and is given by
\[
	\frac{1}{4}\left(
		\begin{array}{rrrr}
			2 & -1 & -1 & 0\\
			-1 & 2 & 0 & -1\\
			-1 & 0 & 2 & -1\\
			0 & -1 & -1 & 2
		\end{array}
	\right),
\]
with the ordering of basis states as in \eqref{eq:z_U_K_L_M}.
The unique zero eigenvector of this matrix is the all-ones vector.
This means that the groundspace of $H_{1\text{q}}(U)$ is spanned by the two states
\[
	|z\rangle|\mathcal{K}\rangle 
	+ U|z\rangle|\mathcal{L}\rangle 
	+ |z\rangle|\mathcal{M}\rangle 
	+ U|z\rangle|\mathcal{N}\rangle, \quad z\in\{0,1\}.
\]
Note that to solve for these zero energy eigenvectors it was sufficient to consider the matrix elements of $H_U$ in the unnormalized basis \eqref{eq:z_U_K_L_M}. 

Now considering superpositions of these two states we see that every state in
the groundspace of $H_{1\text{q}}(U)$ has the form 
\begin{equation}
	|\phi\rangle|\mathcal{K}\rangle 
	+ U|\phi\rangle|\mathcal{L}\rangle 
	+ |\phi\rangle|\mathcal{M}\rangle 
	+ U|\phi\rangle|\mathcal{N}\rangle
	\label{eq:phi_single_qubit}
\end{equation}
for some single-qubit state $|\phi\rangle.$ In this example we view
the state $|C_{1}\rangle|C_{1}\rangle$ (corresponding to the top
left vertex in Figure~\ref{fig:1qubit_example}) as the initial state
of the two clocks, and we view the state $|C_{6}\rangle|C_{6}\rangle$
(the bottom right vertex) as the final state. We interpret \eqref{eq:phi_single_qubit}
as a history state for the computation that consists of applying $U$
to the state $|\phi\rangle$.


\subsection{A two-qubit unitary}

Now consider an example where $N = 9$ and the computational register
contains two qubits. The Hilbert space is 
\[
	\left(\mathbb{C}^{2}\right)^{2} \otimes 
	\mathcal{H}_{\text{clock}}^{(9)} \otimes
	\mathcal{H}_{\text{clock}}^{(9)}.
\]
Define 
\begin{equation}
	H_{2\text{q}} = \iii \otimes \iii \otimes H_{\text{clock}}^{(9)} 
		+ \iii \otimes  H_{\text{clock}}^{(9)} \otimes \iii  
		+ \iii \otimes  S^{(1,3)} 
		+ \iii \otimes  S^{(7,9)} 
		+ H_{V},
	\label{eq:H(V)}
\end{equation}
 where 
\begin{equation}
H_{V} = H_{\text{horizontal}} + H_{\text{vertical}}\label{eq:Hv}
\end{equation}
 and $H_{\text{horizontal}}$ involves transitions of the first clock
register whereas $H_{\text{vertical}}$ involves transitions of the
second clock register. Labeling the first computational (control) qubit $a$
and the second (target) one $b$, we define 
\begin{align}
	H_{\text{horizontal}}   =\, &
			|0\rangle\langle 0|_{a}\otimes h_{34} \otimes \iii  
			+ \iii \otimes  h_{34}\otimes C_{\geq7} 
			+ \iii \otimes  h_{56}\otimes C_{\leq3} 
			+ h_{45}(B_{b}) \otimes \iii  \nonumber\\
   + \, & |0\rangle\langle 0|_{a}\otimes h_{67} \otimes \iii  
 		+ \iii \otimes  h_{67}\otimes C_{\geq7}
 		+ \iii \otimes  h_{56}\otimes C_{\geq7},
 		\label{eq:H_horiz}
\end{align}
\begin{align}
	H_{\text{vertical}}   = \, &
		|1\rangle\langle 1|_{a} \otimes \iii \otimes h_{34} 
		+ \iii \otimes  C_{\geq7}\otimes h_{34} 
		+ \iii \otimes  C_{\leq3}\otimes h_{56} 
		+ h_{45}(\sigma_{b}^{z}) 
		\nonumber\\
 	 + \,&|1\rangle\langle 1|_{a} \otimes \iii \otimes h_{67} 
 		+ \iii \otimes  C_{\geq7}\otimes h_{67}
 		+ \iii \otimes  C_{\geq7}\otimes h_{56}.
 	\label{eq:H_vert}
\end{align}

Here the single-qubit unitaries which act on qubit $b$ are 
\begin{equation}
	\sigma^{z} = \left(\begin{array}{cc}
		1 & 0\\
		0 & -1
	\end{array}\right) \qquad \text{and} \qquad 
	B = \frac{1}{\sqrt{2}}\left(\begin{array}{cc}
		1 & i\\
		i & 1
	\end{array}\right).
	\label{eq:B_sigma_z}
\end{equation}
In \eqref{eq:H_vert} the operator $h_{45}(\sigma_{b}^{z})$
acts nontrivially on the computational qubit $b$ as well as the second clock register
(and acts as the identity on the first clock register). 

Note that $H_{2\text{q}}$ is a sum of 3-local projectors. We now discuss
its groundspace.

Look at the first four terms in \eqref{eq:H(V)} which act
as
\begin{equation}
	\iii \otimes  H_{\text{clock}}^{(9)} 
	+ H_{\text{clock}}^{(9)} \otimes \iii  
	+ S^{(1,3)} + S^{(7,9)}
	\label{eq:H_diag_twoqubit}
\end{equation}
 on the two clock registers. Using our graphical notation, the zero
energy groundspace of \eqref{eq:H_diag_twoqubit} can be represented
as the black graph shown in Figure~\ref{fig:colorset_defn} (a) and
(b). A ground state is associated with each of the 25 connected components
of this graph (as discussed in the Figure caption). Now adjoining the two-qubit computational register,
we get 100 basis vectors for the nullspace of 
\begin{equation}
\iii \otimes \iii \otimes H_{\text{clock}}^{(9)} + \iii \otimes  H_{\text{clock}}^{(9)} \otimes \iii  + \iii \otimes  S^{(1,3)} + \iii \otimes  S^{(7,9)},\label{eq:first_four_terms}
\end{equation}
four for each connected component. States in the nullspace of $H_{2\text{q}}$
are superpositions of these 100 basis vectors that also have zero
energy for $H_{V}$, that is to say, zero eigenvectors of the matrix
\begin{align}
	\langle\mathcal{J}_{2}|\bra{y'}\bra{x'}H_{V}\ket{x}\ket{y}|\mathcal{J}_{1}\rangle \label{eq:mat_HV}
\end{align}
where $\mathcal{J}_{1}$ and $\mathcal{J}_{2}$ are connected components
of the graph in Figure~\ref{fig:colorset_defn} and $x,y,x',y'\in\{0,1\}$. (Here $|\mathcal{J}_{1}\rangle$ and $|\mathcal{J}_{2}\rangle$ are defined through \eqref{eq:S_defn}). One could now proceed to solve for the nullspace of $H_{2\text{q}}$ by
explicitly constructing the matrix elements \eqref{eq:mat_HV} and analyzing
the resulting $100\times100$ matrix. However, computing the matrix elements \eqref{eq:mat_HV} is a tedious exercise. To save space and time, we take a different approach here. We provide a Lemma which characterizes
the nullspace of $H_{2\text{q}}$, and we describe a simple way that the
reader can verify our claim.

It will be helpful to use the sets of vertices $\mathcal{R}_{0},\mathcal{G}_{0},\mathcal{Y}_{0},\mathcal{B}_{0}$
and $\mathcal{R}_{1},\mathcal{G}_{1},\mathcal{Y}_{1},\mathcal{B}_{1}$
depicted in Figure~\ref{fig:colorset_defn}. For each set we define
an unnormalized state through \eqref{eq:S_defn} (now letting
$\mathcal{S}$ be any set of vertices). For example, 
\[
	|\mathcal{Y}_{0}\rangle = 
		  |C_{2}\rangle|C_{7}\rangle 
		+ |C_{3}\rangle|C_{7}\rangle 
		+ |C_{4}\rangle|C_{7}\rangle 
		+ |C_{2}\rangle|C_{8}\rangle 
		+ |C_{3}\rangle|C_{8}\rangle 
		+ |C_{4}\rangle|C_{8}\rangle.
\]
 Note that each of the states 
\begin{equation}
|\mathcal{R}_{0}\rangle,|\mathcal{G}_{0}\rangle,|\mathcal{Y}_{0}\rangle,|\mathcal{B}_{0}\rangle\qquad\text{and}\qquad|\mathcal{R}_{1}\rangle,|\mathcal{G}_{1}\rangle,|\mathcal{Y}_{1}\rangle,|\mathcal{B}_{1}\rangle\label{eq:R_0_R_1}
\end{equation}
have zero energy for \eqref{eq:H_diag_twoqubit}. The following Lemma
characterizes the groundspace of $H_{2\text{q}}$.
\begin{lem}
\label{lem:H_V}The groundspace of $H_{2\text{q}}$ is spanned by 
\begin{eqnarray*}
|\psi_{V}^{xy}\rangle &  =  & |x\rangle|y\rangle|\mathcal{R}_{x}\rangle + \left(Q|x\rangle|y\rangle\right)|\mathcal{G}_{x}\rangle + \big(\tilde{Q}|x\rangle|y\rangle\big)|\mathcal{Y}_{x}\rangle + \left(V|x\rangle|y\rangle\right)|\mathcal{B}_{x}\rangle
\end{eqnarray*}
 for $x,y\in\{0,1\}$, where the two-qubit unitaries $Q,\tilde{Q},$
and $V$ are given by 
\begin{align}
	Q &  = |0\rangle\langle 0|\otimes B + |1\rangle\langle 1|\otimes\sigma^{z}, \nonumber \\
	\tilde{Q} &  = |0\rangle\langle 0|\otimes	\big( B^{\dagger}\sigma^{z}B \big) 
		+ |1\rangle\langle 1|\otimes\left(\sigma^{z}B\sigma^{z}\right), \nonumber \\
	V &  = |0\rangle\langle 0|\otimes\left(\sigma^{z}B\right) 
		+ |1\rangle\langle 1|\otimes\left(B\sigma^{z}\right),
		\label{eq:V_gate}
\end{align}
with the single-qubit unitaries $B$ and $\sigma^{z}$ as in \eqref{eq:B_sigma_z}.
\end{lem}
Note that, since $|\psi_{V}^{xy}\rangle$ has support only on states
\eqref{eq:R_0_R_1} of the clock registers, it has zero energy for \eqref{eq:first_four_terms}.
Using equations \eqref{eq:hi_iplus1}, \eqref{eq:1localproj}, and \eqref{eq:1localproj2}
the reader can verify that each state $|\psi_{V}^{xy}\rangle$ also
has zero energy for $H_{V}$. It remains to show that these four states
span the groundspace of $H_{2\text{q}}.$ We recommend using a computer
to verify this fact. To do this, one can numerically diagonalize a
specific $324\times324$ matrix: the restriction of $H_{2\text{q}}$ to
the space spanned by 
\[
	|z_{1}\rangle|z_{2}\rangle|C_{i}\rangle|C_{j}\rangle
\]
 with $i,j = 1,\dots,9$ and $z_{1},z_{2}\in\{0,1\}$. It is easy to compute the matrix elements of  $H_{2\text{q}}$ in this basis using equations \eqref{eq:hi_iplus1}, \eqref{eq:1localproj}, and \eqref{eq:1localproj2}. We have
included in our arxiv submission an ancillary file (a Matlab script)
which numerically diagonalizes this matrix and confirms that the states
$|\psi_{V}^{xy}\rangle$ span the nullspace of $H_{2\text{q}}$. 

The reader may find it helpful to look at Figure \ref{fig:extrafigure2} from Appendix \ref{sec:clockreference} in order to understand how the states $|\psi_{V}^{xy}\rangle$ arise as ground states of the Hamiltonian $H_{2q}$.

\begin{center}
\begin{figure}[tb]
\center
\subfigure[]
{
\begin{tikzpicture}[scale = 0.5,vertex/.style = {circle,draw = black,fill = black,inner sep = 0.5pt,minimum size  = 0mm}]

\foreach \i in {1,...,7}
{
\foreach \j in {-1,...,-7}
{
      \draw (\i,\j) node[vertex]{};
}
}
\foreach \i in {-1,...,-7}
{
\draw (1,\i)--(2,\i);
\draw (6,\i)--(7,\i);
}
\foreach \i in {1,...,7}
{
\draw (\i,-1)--(\i,-2);
\draw (\i,-6)--(\i,-7);
}

\draw (0,0) node[vertex]{}--(1,0) node[vertex]{}--(1,-1) node[vertex]{}--(0,-1) node[vertex]{}--(0,0);
\draw (1,-1)--(2,-1)--(2,-2)--(1,-2)--(1,-1);

\draw (6,-6) --  + (1,0) --  + (1,-1)--  + (2,-1)--  + (2,-2)--  + (1,-2)--  + (1,-1)--  + (0,-1)--  + (0,0);

\draw (7,-8) node[vertex]{};
\draw (8,-7) node[vertex]{};
\draw (8,-8) node[vertex]{};
\draw[thick,dotted,color = red](-0.25,0.25)--(1.25,0.25)--(1.25,-0.75)--(3.25,-0.75)--(3.25,-2.25)--(0.75,-2.25)--(0.75,-1.25)--(-0.25,-1.25)--(-0.25,0.25);
\draw [thick,dotted,color = green]
(3.75,-0.75)--(7.25,-0.75)--(7.25,-3.25)--(4.75,-3.25)--(4.75,-2.25)--(3.75,-2.25)--(3.75,-0.75);
\draw [thick,dotted,color = blue]
(3.75,-5.75)--(4.75,-5.75)--(4.75,-3.75)--(7.25,-3.75)--(7.25,-6.75)--(8.25,-6.75)--(8.25,-8.25)--(6.75,-8.25)--(6.75,-7.25)--(3.75,-7.25)--(3.75,-5.75);
\draw [thick,dotted,color = yellow]
(3.25,-5.75)--(0.75,-5.75)--(0.75,-7.25)--(3.25,-7.25)--(3.25,-5.75);
\draw (-0.25,-2) node[]{$\mathcal{R}_0$};
\draw (8,-0.5) node[] {$\mathcal{G}_0$};
\draw (8,-6) node[] {$\mathcal{B}_0$};
\draw (-0.25,-7) node[] {$\mathcal{Y}_0$};
\end{tikzpicture}
}
\hspace{1.5cm}
\subfigure[]
{
\begin{tikzpicture}[scale = 0.5,vertex/.style = {circle,draw = black,fill = black,inner sep = 0.5pt,minimum size  = 0mm}]

\foreach \i in {1,...,7}
{
\foreach \j in {-1,...,-7}
{
      \draw (\i,\j) node[vertex]{};
}
}
\foreach \i in {-1,...,-7}
{
\draw (1,\i)--(2,\i);
\draw (6,\i)--(7,\i);
}
\foreach \i in {1,...,7}
{
\draw (\i,-1)--(\i,-2);
\draw (\i,-6)--(\i,-7);
}

\draw (0,0) node[vertex]{}--(1,0) node[vertex]{}--(1,-1) node[vertex]{}--(0,-1) node[vertex]{}--(0,0);
\draw (1,-1)--(2,-1)--(2,-2)--(1,-2)--(1,-1);

\draw (6,-6) --  + (1,0) --  + (1,-1)--  + (2,-1)--  + (2,-2)--  + (1,-2)--  + (1,-1)--  + (0,-1)--  + (0,0);

\draw (7,-8) node[vertex]{};
\draw (8,-7) node[vertex]{};
\draw (8,-8) node[vertex]{};
\draw[thick,dotted,color = red](-0.25,0.25)--(1.25,0.25)--(1.25,-0.75)--(2.25,-0.75)--(2.25,-3.25)--(0.75,-3.25)--(0.75,-1.25)--(-0.25,-1.25)--(-0.25,0.25);
\draw [thick,dotted,color = green]
(0.75,-3.75)--(0.75,-7.25)--(3.25,-7.25)--(3.25,-4.75)--(2.25,-4.75)--(2.25,-3.75)--(0.75,-3.75);
\draw [thick,dotted,color = blue]
(3.75,-4.75)--(5.75,-4.75)--(5.75,-3.75)--(7.25,-3.75)--(7.25,-6.75)--(8.25,-6.75)--(8.25,-8.25)--(6.75,-8.25)--(6.75,-7.25)--(3.75,-7.25)--(3.75,-4.75);
\draw [thick,dotted,color = yellow]
(5.75,-3.25)--(5.75,-0.75)--(7.25,-0.75)--(7.25,-3.25)--(5.75,-3.25);
\draw (-0.25,-2) node[]{$\mathcal{R}_1$};
\draw (8,-0.5) node[] {$\mathcal{Y}_1$};
\draw (8,-6) node[] {$\mathcal{B}_1$};
\draw (-0.25,-7) node[] {$\mathcal{G}_1$};
\end{tikzpicture}
}

\caption{The black graph (drawn twice for clarity) depicts the groundspace
of $\iii \otimes  H_{\text{clock}}^{(9)} + H_{\text{clock}}^{(9)} \otimes \iii  + S^{(1,3)} + S^{(7,9)}$.
The top left vertex is labeled $(1,1)$ and the bottom right vertex
$(9,9).$ Each connected component corresponds to a ground state,
given by the uniform superposition of states $|C_{i}\rangle|C_{j}\rangle$
with $(i,j)$ in the component. In (a) we have defined sets of vertices
$\mathcal{R}_{0},\mathcal{G}_{0},\mathcal{Y}_{0},\mathcal{B}_{0}$
and in (b) we have defined sets $\mathcal{R}_{1},\mathcal{G}_{1},\mathcal{Y}_{1},\mathcal{B}_{1}$.\label{fig:colorset_defn}}
\end{figure}
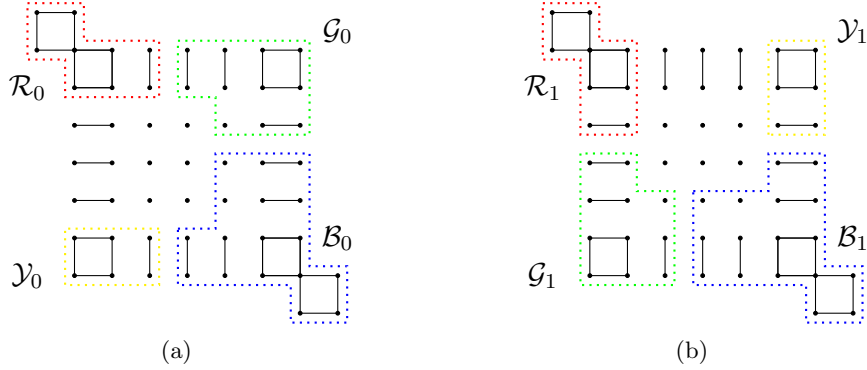

\par\end{center}

Using Lemma \ref{lem:H_V} we see that any state in the groundspace
of $H_{2\text{q}}$ is a superposition
\begin{equation}
\sum_{x,y\in\{0,1\}}\alpha_{xy}|\psi_{V}^{xy}\rangle = |\phi\rangle|C_{1}\rangle|C_{1}\rangle + |\text{other\ensuremath{\rangle}} + \left(V|\phi\rangle\right)|C_{9}\rangle|C_{9}\rangle,\label{eq:superposition_psiV}
\end{equation}
where 
\[
|\phi\rangle = \sum_{x,y\in\{0,1\}}\alpha_{xy}|x\rangle|y\rangle,
\]
and $|\text{other}\rangle$ has no support on clock states $|C_{1}\rangle|C_{1}\rangle$
or $|C_{9}\rangle|C_{9}\rangle$. We view $|C_{1}\rangle|C_{1}\rangle$
as the initial state of the two clocks and $|C_{9}\rangle|C_{9}\rangle$ as the
final state of the two clocks, and we interpret \eqref{eq:superposition_psiV} as
a history state for the computation that consists of applying the
two-qubit unitary $V$ from \eqref{eq:V_gate} to the state
$|\phi\rangle$.

Finally, we show that the two-qubit unitary $V$ is an entangling gate. 
To see this, note that by multiplying it with single-qubit $T$ and Hadamard gates we obtain the CNOT gate:

\begin{align}
	\left(T^{2}\otimes \big( T^{6}\hat{H}T^{2} \big)\right)V &  = \text{CNOT}.
	\label{eq:CNOT_V}
\end{align}

The reader may already see where this is going. In this Section we
exhibited Hamiltonians $H_{1\text{q}}(U)$ and $H_{2\text{q}}$ which are sums of 3-local
projectors and which have ground states that can be viewed as history
states for any one-qubit computation and a specific two-qubit computation respectively.
Now we show how to put these ideas together to make a quantum $3$-SAT
Hamiltonian that is associated with a sequence of one- and two-qubit
gates.


\section{Quantum 3-SAT is QMA$_{1}$-hard \label{sec:Circuit-to-Hamiltonian-mapping}}

In this Section we prove that quantum $3$-SAT is QMA$_{1}$-hard,
following the strategy outlined in Section \ref{sec:Preliminaries-and-Overview}.
We exhibit our circuit-to-Hamiltonian mapping and we prove Theorems
\ref{thm:Completeness} and \ref{thm:Soundness}. 

Recall from Section \ref{sec:Preliminaries-and-Overview} that we
consider a verification circuit which implements a unitary $U_{X}$
on $n + n_{a}$ qubits, $n_{a}$ of which are ancillas initialized to
$|0\rangle$ at the beginning of the computation. It is expressed
as a product of $g$ gates from the set $\{\widehat{H},T,\text{CNOT}\}$.
We begin by rewriting this circuit in a canonical form.

Equation \eqref{eq:CNOT_V} expresses the CNOT gate as a product of
$\hat{H}$ and $T$ gates and the two-qubit 
gate $V$ \eqref{eq:V_gate}. Using this identity we efficiently rewrite the given
circuit so that it is a product of $\Theta(g)$ gates from the set
\[
	\{ \widehat{H},T,V \}.
\]
For convenience, we then rewrite the circuit so that it is expressed
as a product with $M = \Theta(g)$ single-qubit gates alternating with
$M$ two-qubit $V$ gates: 
\begin{equation}
	U_{X} = V_{a_{M-1}b_{M-1}}U^{M-1}\dots V_{a_{1}b_{1}}U^{1}V_{a_{0}b_{0}}U^{0},
	\label{eq:U_C}
\end{equation}
where each single-qubit gate $U^{0},U^{2},\dots U^{M-1}$ is either 
$\widehat{H}$, $T$ or the identity,
and where
\[
	a_{0},\dots,a_{M-1},b_{0},\dots,b_{M-1}\in[n + n_{a}]
\]
are the labels of the qubits on which the $V$ gates act. An arbitrary
circuit expressed using this gate set can always be efficiently rewritten
in this form. To see this, note that we can always place a unitary
equal to the identity in between any two consecutive two qubit gates. Furthermore, using the fact that 
\[
	V^{2} = \iii,
\]
we can replace the product of any two consecutive single-qubit gates
$U^{1}U^{2}$ by $U^{1} \cdot V_{ab} \cdot \iii \cdot V_{ab} \cdot U^{2}$, with the gate $V_{ab}$ acting on some pair of qubits $a,b$.


\subsection{The Hamiltonian $H_{X}$\label{sub:Completeness_Section}}

We define a Hamiltonian $H_{X}$ which we associate with the verification
circuit \eqref{eq:U_C} and which is a sum of 3-local projectors from
the set $\mathcal{P}$ in Definition \ref{Delta_proj}. It acts on the
Hilbert space 
\begin{equation}
	\mathcal{H}_{\text{comp}} \otimes
	\mathcal{H}_{\text{clock}}^{(9M + 3)} \otimes
	\mathcal{H}_{\text{clock}}^{(9M + 3)},
	\label{eq:Hilb_space}
\end{equation}
 where $\mathcal{H}_{\text{comp}}$ is a computational register containing
$n + n_{a}$ qubits. Recalling \eqref{eq:clockspace}, note
that each of the two clock registers contains $63M + 18$ qubits. 

First, consider the Hilbert space 
\[
	\mathcal{H}_{\text{clock}}^{(9M + 3)}\otimes\mathcal{H}_{\text{clock}}^{(9M + 3)}
\]
of the two clock registers and define the following operator acting
on this space 
\begin{equation}
	H_{\text{diag}}^{(M)} 
		= H_{\text{\text{clock}}}^{(9M + 3)} \otimes \iii  
		+ \iii \otimes  H_{\text{clock}}^{(9M + 3)} 
		+ \sum_{j = 0}^{M-1}\left(S^{(9j + 1,9j + 3)} 
		+ S^{(9j + 4,9j + 6)}\right) 
		+ S^{(9M + 1,9M + 3)},
	\label{eq:H_diag}
\end{equation}
 where $S^{(k,k + 2)}$ is defined in \eqref{eq:S_k_k_2}. 

Let us pause for a moment and explain why we consider this operator.
Using the graphical notation developed in Section \ref{sec:Hamiltonians-with-two},
the groundspace of \eqref{eq:H_diag} can be represented as the graph drawn
in black in Figure~\ref{fig:groundstates}. For the moment, let us
focus on the graph drawn in black and ignore all other aspects
of the Figure. As described in the caption, a basis for the
groundspace is in $1$-$1$ correspondence with the connected components
of the graph. Note that the graphs from Figures \ref{fig:1qubit_example}
and \ref{fig:colorset_defn} each appear $M$ times along the diagonal.
This corresponds to the fact that the verification circuit contains
$M$ one-qubit gates and $M$ two-qubit 
gates $V$.

We obtain $H_{X}$ by adding terms to $\iii \otimes  H_{\text{diag }}^{(M)}$.
We add terms for each one- and two-qubit gate and we add terms which
check the initial and final state of the computation. Specifically,
let

\begin{align}
	H_{X} &  = \iii \otimes  H_{\text{diag}}^{(M)} 
		+ \sum_{j = 0}^{M-1}\left(
				  H_{U}^{j} 
				+ H_{V}^{j}
			\right) 
		+ H_{\text{init}} 
		+ H_{\text{end}},
	\label{eq:H_C}
\end{align}
 which acts on the Hilbert space \eqref{eq:Hilb_space}. Here 
\begin{align}
	H_{\text{init}} &  = \sum_{i = 1}^{n_{a}}|1\rangle\langle 1|_{(i + n)}\otimes C_{\leq1}\otimes C_{\leq1}, \label{eq:H_init_H_end}\\
	H_{\text{end}} &= |0\rangle\langle 0|_{(n + 1)}\otimes C_{\geq\left(9M + 3\right)}\otimes C_{\geq\left(9M + 3\right)} \nonumber
\end{align}
ensure that each of the ancilla qubits is in the $|0\rangle$ state
when the clock state is $|C_{1}\rangle|C_{1}\rangle$, and that the
first qubit in the ancilla register is in the state $|1\rangle$ when
the clock state is $|C_{9M + 3}\rangle|C_{9M + 3}\rangle$. The operators
\begin{equation}
	H_{U}^{j} = h_{9j + 3,9j + 4}(U^{j}) \otimes \iii  
		+ \iii \otimes \iii \otimes h_{9j + 3,9j + 4}
	\label{eq:H_Uj_def}
\end{equation}
 are defined by analogy with \eqref{eq:Hu} and involve the
single-qubit unitaries $\{U^{j}\}$. We also define
\begin{equation}
	H_{V}^{j} = H_{\text{horizontal}}^{j} 
		+ H_{\text{vertical}}^{j}
	\label{eq:H_Vj_def}
\end{equation}
by analogy with \eqref{eq:Hv}:
\begin{align}
	H_{\text{horizontal}}^{j} &  = 
		|0\rangle\langle 0|_{a_{j}}\otimes h_{9j + 6,9j + 7} \otimes \iii  
		\,\, + \iii \otimes  h_{9j + 6,9j + 7}\otimes C_{\geq 9j + 10} 
		\; + \iii \otimes  h_{9j + 8,9j + 9}\otimes C_{\leq 9j + 6} 
		 \nonumber \\
 &  + |0\rangle\langle 0|_{a_{j}}\otimes h_{9j + 9,9j + 10} \otimes \iii
 		+ \iii \otimes  h_{9j + 9,9j + 10}\otimes C_{\geq 9j + 10}
 		+ \iii \otimes  h_{9j + 8,9j + 9}\otimes C_{\geq 9j + 10} 
 		 \nonumber \\
 &  + h_{9j + 7,9j + 8}(B_{b_{j}}) \otimes \iii
		 \label{eq:H_horizontalj}\\
	H_{\text{vertical}}^{j} &  = 
		|1\rangle\langle 1|_{a_{j}} \otimes \iii \otimes h_{9j + 6,9j + 7} 
		\,\, + \iii \otimes  C_{\geq 9j + 10}\otimes h_{9j + 6,9j + 7} 
		\; + \iii \otimes  C_{\leq 9j + 6}\otimes h_{9j + 8,9j + 9} \nonumber\\
 &  + \ket{1}\langle 1|_{a_{j}} \otimes \iii \otimes h_{9j + 9,9j + 10}
 		+ \iii \otimes  C_{\geq 9j + 10}\otimes h_{9j + 9,9j + 10} 
 		+ \iii \otimes  C_{\geq 9j + 10}\otimes h_{9j + 8,9j + 9} 
 		\nonumber \\
 &  + h_{9j + 7,9j + 8}(\sigma_{b_{j}}^{z}).
 	 \label{eq:H_verticalj}
\end{align}

Note that $H_{X}$ is a sum of $3$-local projectors. Furthermore,
each projector in the sum is of the form given in Definition \ref{Delta_proj}.
(To see this, first recall that in Section \ref{sec:Details-of-the}
we showed that $H_{\text{clock}}^{(9M + 3)}$ is a sum of projectors
from the set $\mathcal{P}$. For the rest of the terms in $H_{X}$, one
can use the definitions of $C_{\leq i}$ and $C_{\geq i}$ and $h_{i,i + 1}(U)$
from equations \eqref{eq:C_leq}, \eqref{eq:C_geq} and \eqref{eq:h_i_explicit}
to check that they are projectors satisfying condition 1 from Definition
\ref{Delta_proj}. Note that for the term $h_{9j + 3,9j + 4}(U^{j}) \otimes \iii $
appearing in \eqref{eq:H_Uj_def} we use the fact that $U^{j}$
is either $\iii$, $\hat{H}$ or $T$.

We now characterize the groundspace of $H_{X}$.

\begin{center}
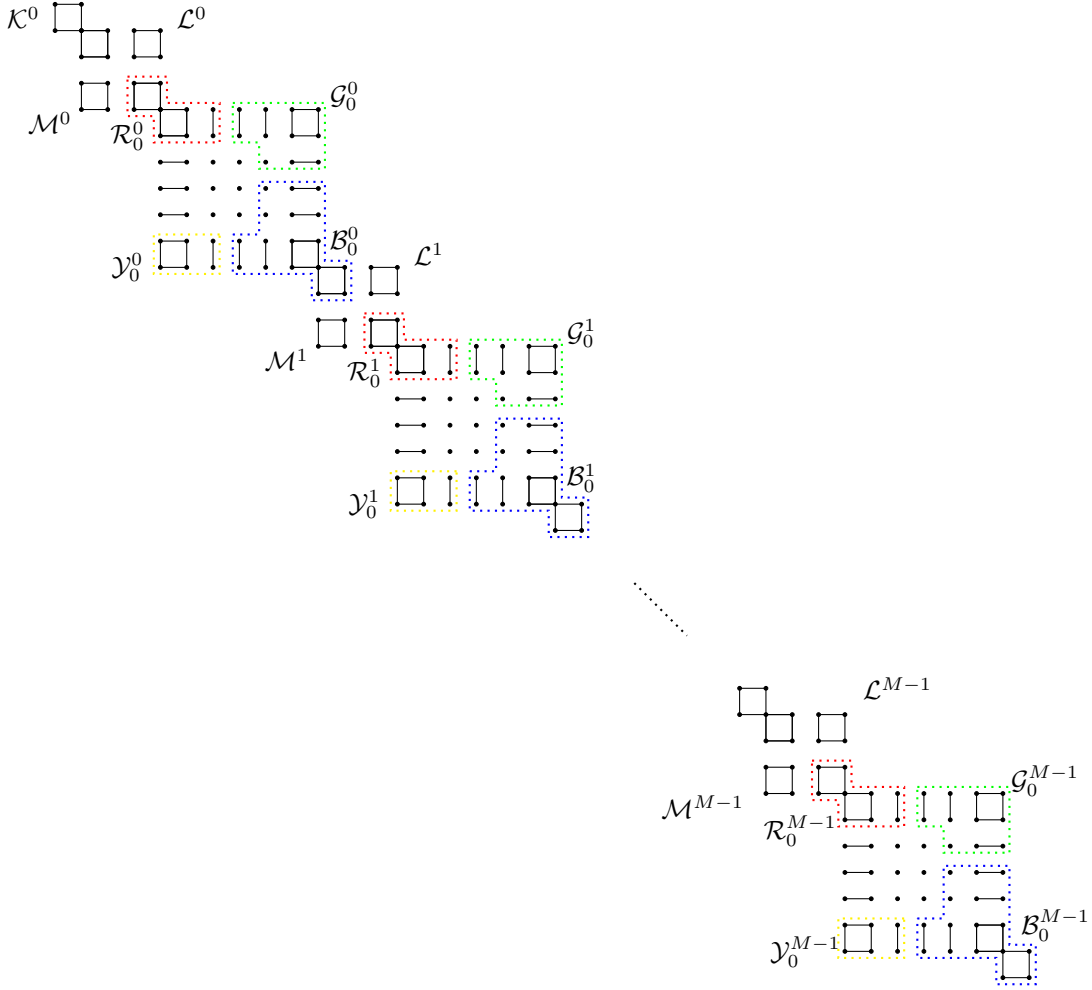
\begin{figure}[t]
\center
{\scalefont{0.9}
\begin{tikzpicture}[scale = 0.35,vertex/.style = {circle,draw = black,fill = black,inner sep = 0.5pt,minimum size  = 0mm}]

\foreach \i in {1,...,4}
{
\foreach \j in {-1,...,-4}
{
      \draw (\i,\j) node[vertex]{};
}
}
\foreach \i in {-1,...,-4}
{
\draw (1,\i)--(2,\i);
\draw (3,\i)--(4,\i);
}
\foreach \i in {1,...,4}
{
\draw (\i,-1)--(\i,-2);
\draw (\i,-3)--(\i,-4);
}

\foreach \i in {4,...,10}
{
\foreach \j in {-4,...,-10}
{
      \draw (\i,\j) node[vertex]{};
}
}

\foreach \i in {-4,...,-10}
{
\draw (4,\i)--(5,\i);
\draw (9,\i)--(10,\i);
}
\foreach \i in {4,...,10}
{
\draw (\i,-4)--(\i,-5);
\draw (\i,-9)--(\i,-10);
}

\foreach \i in {10,...,13}
{
\foreach \j in {-10,...,-13}
{
      \draw (\i,\j) node[vertex]{};
}
}
\foreach \i in {-10,...,-13}
{
\draw (10,\i)--(11,\i);
\draw (12,\i)--(13,\i);
}
\foreach \i in {10,...,13}
{
\draw (\i,-10)--(\i,-11);
\draw (\i,-12)--(\i,-13);
}
\foreach \i in {13,...,19}
{
\foreach \j in {-13,...,-19}
{
      \draw (\i,\j) node[vertex]{};
}
}
\foreach \i in {-13,...,-19}
{
\draw (13,\i)--(14,\i);
\draw (18,\i)--(19,\i);
}
\foreach \i in {13,...,19}
{
\draw (\i,-13)--(\i,-14);
\draw (\i,-18)--(\i,-19);
}
\foreach \i in {27,...,30}
{
\foreach \j in {-27,...,-30}
{
      \draw (\i,\j) node[vertex]{};
}
}
\foreach \i in {-27,...,-30}
{
\draw (27,\i)--(28,\i);
\draw (29,\i)--(30,\i);
}
\foreach \i in {27,...,30}
{
\draw (\i,-27)--(\i,-28);
\draw (\i,-29)--(\i,-30);
}
\foreach \i in {30,...,36}
{
\foreach \j in {-30,...,-36}
{
      \draw (\i,\j) node[vertex]{};
}
}

\foreach \i in {-30,...,-36}
{
\draw (30,\i)--(31,\i);
\draw (35,\i)--(36,\i);
}
\foreach \i in {30,...,36}
{
\draw (\i,-30)--(\i,-31);
\draw (\i,-35)--(\i,-36);
}
\draw (-1.2,-0.5) node[]{$\mathcal{K}^0$};
\draw (5.2,-0.5) node[]{$\mathcal{L}^0$};
\draw (-0.2,-4.5) node[]{$\mathcal{M}^0$};
\draw (14.2,-9.5) node[]{$\mathcal{L}^1$};
\draw (8.8,-13.5) node[]{$\mathcal{M}^1$};
\draw (24.6,-30.5) node[]{$\mathcal{M}^{M-1}$};
\draw (32,-26) node[]{$\mathcal{L}^{M-1}$};

\begin{scope}[shift = {(3,-3)}]
\draw[thick,dotted,color = red](-0.25,0.25)--(1.25,0.25)--(1.25,-0.75)--(3.25,-0.75)--(3.25,-2.25)--(0.75,-2.25)--(0.75,-1.25)--(-0.25,-1.25)--(-0.25,0.25);
\draw [thick,dotted,color = green]
(3.75,-0.75)--(7.25,-0.75)--(7.25,-3.25)--(4.75,-3.25)--(4.75,-2.25)--(3.75,-2.25)--(3.75,-0.75);
\draw [thick,dotted,color = blue]
(3.75,-5.75)--(4.75,-5.75)--(4.75,-3.75)--(7.25,-3.75)--(7.25,-6.75)--(8.25,-6.75)--(8.25,-8.25)--(6.75,-8.25)--(6.75,-7.25)--(3.75,-7.25)--(3.75,-5.75);
\draw [thick,dotted,color = yellow]
(3.25,-5.75)--(0.75,-5.75)--(0.75,-7.25)--(3.25,-7.25)--(3.25,-5.75);
\draw (-0.25,-2) node[]{$\mathcal{R}_0^{0}$};
\draw (8,-0.5) node[] {$\mathcal{G}_0^{0}$};
\draw (8,-6) node[] {$\mathcal{B}_0^{0}$};
\draw (-0.25,-7) node[] {$\mathcal{Y}_0^{0}$};
\end{scope}
\begin{scope}[shift = {(12,-12)}]
\draw[thick,dotted,color = red](-0.25,0.25)--(1.25,0.25)--(1.25,-0.75)--(3.25,-0.75)--(3.25,-2.25)--(0.75,-2.25)--(0.75,-1.25)--(-0.25,-1.25)--(-0.25,0.25);
\draw [thick,dotted,color = green]
(3.75,-0.75)--(7.25,-0.75)--(7.25,-3.25)--(4.75,-3.25)--(4.75,-2.25)--(3.75,-2.25)--(3.75,-0.75);
\draw [thick,dotted,color = blue]
(3.75,-5.75)--(4.75,-5.75)--(4.75,-3.75)--(7.25,-3.75)--(7.25,-6.75)--(8.25,-6.75)--(8.25,-8.25)--(6.75,-8.25)--(6.75,-7.25)--(3.75,-7.25)--(3.75,-5.75);
\draw [thick,dotted,color = yellow]
(3.25,-5.75)--(0.75,-5.75)--(0.75,-7.25)--(3.25,-7.25)--(3.25,-5.75);
\draw (-0.25,-2) node[]{$\mathcal{R}_0^{1}$};
\draw (8,-0.5) node[] {$\mathcal{G}_0^{1}$};
\draw (8,-6) node[] {$\mathcal{B}_0^{1}$};
\draw (-0.25,-7) node[] {$\mathcal{Y}_0^{1}$};
\end{scope}
\begin{scope}[shift = {(29,-29)}]
\draw[thick,dotted,color = red](-0.25,0.25)--(1.25,0.25)--(1.25,-0.75)--(3.25,-0.75)--(3.25,-2.25)--(0.75,-2.25)--(0.75,-1.25)--(-0.25,-1.25)--(-0.25,0.25);
\draw [thick,dotted,color = green]
(3.75,-0.75)--(7.25,-0.75)--(7.25,-3.25)--(4.75,-3.25)--(4.75,-2.25)--(3.75,-2.25)--(3.75,-0.75);
\draw [thick,dotted,color = blue]
(3.75,-5.75)--(4.75,-5.75)--(4.75,-3.75)--(7.25,-3.75)--(7.25,-6.75)--(8.25,-6.75)--(8.25,-8.25)--(6.75,-8.25)--(6.75,-7.25)--(3.75,-7.25)--(3.75,-5.75);
\draw [thick,dotted,color = yellow]
(3.25,-5.75)--(0.75,-5.75)--(0.75,-7.25)--(3.25,-7.25)--(3.25,-5.75);
\draw (-0.7,-2.4) node[]{$\mathcal{R}_0^{M-1}$};
\draw (8.6,-0.5) node[] {$\mathcal{G}_0^{M-1}$};
\draw (9,-6) node[] {$\mathcal{B}_0^{M-1}$};
\draw (-0.5,-7) node[] {$\mathcal{Y}_0^{M-1}$};
\end{scope}
\draw (0,0) node[vertex]{}--(1,0) node[vertex]{}--(1,-1) node[vertex]{}--(0,-1) node[vertex]{}--(0,0);
\draw (1,-1)--(2,-1)--(2,-2)--(1,-2)--(1,-1);

\draw (3,-3) --  + (1,0) --  + (1,-1)--  + (2,-1)--  + (2,-2)--  + (1,-2)--  + (1,-1)--  + (0,-1)--  + (0,0);
\draw (9,-9) --  + (1,0) --  + (1,-1)--  + (2,-1)--  + (2,-2)--  + (1,-2)--  + (1,-1)--  + (0,-1)--  + (0,0);
\draw (12,-12) --  + (1,0) --  + (1,-1)--  + (2,-1)--  + (2,-2)--  + (1,-2)--  + (1,-1)--  + (0,-1)--  + (0,0);
\draw (18,-18) --  + (1,0) --  + (1,-1)--  + (2,-1)--  + (2,-2)--  + (1,-2)--  + (1,-1)--  + (0,-1)--  + (0,0);
\draw (26,-26) --  + (1,0) --  + (1,-1)--  + (2,-1)--  + (2,-2)--  + (1,-2)--  + (1,-1)--  + (0,-1)--  + (0,0);
\draw (35,-35) --  + (1,0) --  + (1,-1)--  + (2,-1)--  + (2,-2)--  + (1,-2)--  + (1,-1)--  + (0,-1)--  + (0,0);

\draw (20,-19) node[vertex]{};
\draw (19,-20) node[vertex]{};
\draw (20,-20) node[vertex]{};

\draw (26,-26) node[vertex]{};
\draw (27,-26) node[vertex]{};
\draw (26,-27) node[vertex]{};

\draw (37,-36) node[vertex]{};
\draw (37,-37) node[vertex]{};
\draw (36,-37) node[vertex]{};
\draw[dotted,thick] (22,-22)--(24,-24);
\end{tikzpicture}
}

\caption{A basis for the groundspace of $H_{\text{diag}}^{(M)}$ is in $1$-$1$
correspondence with the connected components of the graph drawn in
black. The vertices are labeled $(i,j)$ with the top left vertex
labeled $(1,1)$ and the bottom right vertex $(9M + 3,9M + 3)$. The ground
state associated with a connected component $\mathcal{J}$ is the
uniform superposition $\sum_{(i,j)\in\mathcal{J}}|C_{i}\rangle|C_{j}\rangle$.
\label{fig:groundstates}}
\end{figure}

\par\end{center}


\subsection{The zero energy groundspace of $H_{X}$ (Proof of Theorem \ref{thm:Completeness})\label{sub:The-zero-energy}}

We now show that a zero-energy ground state of $H_{X}$ exists 
if and only if there exists a witness $\ket{W}$ 
which the original verifier $U_X$ accepts with certainty. We begin by defining some sets of vertices
in the graph \ref{fig:groundstates}. 
For each copy $j = 0,\dots,M-1$ of the graph in Figure \ref{fig:colorset_defn} 
that appears in Figure \ref{fig:groundstates}, we define sets 
\[
	\mathcal{R}_{0}^{j},\mathcal{G}_{0}^{j},\mathcal{Y}_{0}^{j},\mathcal{B}_{0}^{j}
	\quad \text{and} \quad 
	R_{1}^{j},\mathcal{G}_{1}^{j},\mathcal{Y}_{1}^{j},\mathcal{B}_{1}^{j}.
\]
Likewise, for each copy  $j = 0,\dots,M-1$ 
of the graph in Figure~\ref{fig:1qubit_example} that appears in Figure~\ref{fig:groundstates}  we define sets 
\[ \mathcal{L}^{j},\mathcal{M}^{j} \] in Figure \ref{fig:groundstates}, 
and for the copy with $j=0$ (in the top left) we also define $\mathcal{K}^{0}$ as shown in the Figure.
For each of these sets, we define an associated (unnormalized) state through \eqref{eq:S_defn}. 

To analyze the groundspace of $H_{X}$, we add the terms in equation
\eqref{eq:H_C} one at a time, computing the zero energy states of the
resulting operator at each step.

We start with 
\begin{equation}
	\iii \otimes  H_{\text{diag}}^{(M)},
	\label{eq:1timesHdiag}
\end{equation}
which (as discussed in the previous Section) has nullspace spanned
by states of the form 
\begin{equation}
	|z\rangle|\mathcal{J}\rangle = 	|z\rangle\sum_{(i,j)\in\mathcal{J}}|C_{i}\rangle|C_{j}\rangle,
	\label{eq:phi_J-1}
\end{equation}
 where $\mathcal{J}\subset[9M + 3]\otimes[9M + 3]$ is a connected component
of the graph drawn in black in Figure~\ref{fig:groundstates} and
$z$ is an $\left(n + n_{a}\right)$-bit string. 

Now consider 
\begin{equation}
\iii \otimes  H_{\text{diag}}^{(M)} + \sum_{j = 0}^{M-1}H_{V}^{j}.\label{eq:first_op}
\end{equation}
As the reader might expect, we are going to use Lemma \ref{lem:H_V}
to solve for the zero energy states. We begin by considering the action
of $H_{V}^{j}$ in the basis \eqref{eq:phi_J-1}. 

Look at the graph in black in Figure~\ref{fig:colorset_defn} and
note that there are $M$ copies of this graph along the diagonal in
Figure \ref{fig:groundstates}. Each copy $j = 0,\dots,M-1$ contains
25 connected components $\mathcal{J}$. The operator $H_{V}^{j}$
only has support on states $|z\rangle|\mathcal{J}\rangle$ when $\mathcal{J}$
is one of the 25 connected components in the $j$th copy. To see this,
look at the definition of $H_{V}^{j}$ in \eqref{eq:H_Vj_def}-\eqref{eq:H_verticalj} and note that 
\begin{equation}
	H_{V}^{j}|z\rangle|\mathcal{J}\rangle = 0 \quad \text{whenever} \quad
		\mathcal{J}\cap\{9j + 6,\dots,9j + 10\}
		\otimes
		\{9j + 6,\dots,9j + 10\} = \emptyset.
	\label{eq:H_vj_constraint}
\end{equation}
This implies that the matrix element 
\begin{equation}
	\sum_{j = 0}^{M-1} \langle\mathcal{J}_{2}|\langle z_{2}|H_{V}^{j}|z_{1}\rangle|\mathcal{J}_{1}\rangle
	\label{eq:mat_els_Hvj}
\end{equation}
 is nonzero only when $\mathcal{J}_{1}$ and $\mathcal{J}_{2}$ are
both contained in the same copy $j$. The matrix with entries \eqref{eq:mat_els_Hvj}
is therefore block diagonal with a nonzero block for each $j = 0,\dots,M-1$.
In addition to these blocks, there are ``blocks'' of size 1 -- the states

\begin{equation}
	|z\rangle|\mathcal{K}^{0}\rangle, \qquad
	|z\rangle|\mathcal{M}^{j}\rangle, \qquad
	|z\rangle|\mathcal{L}^{j}\rangle\label{eq:phi_k_m_l}
\end{equation}
for $j = 0,\dots,M-1$ and $z\in\{0,1\}^{n + n_{a}}$ which have zero energy
for \eqref{eq:first_op}. Let us now solve for the zero eigenvectors
of \eqref{eq:first_op} within each nonzero block. The block corresponding
to a given value $j$ is a $\left(2^{n + n_{a}}\cdot25\right)\times\left(2^{n + n_{a}}\cdot25\right)$
matrix with entries
\[
	\langle\mathcal{J}_{2}|\langle z_{2}|H_{V}^{j}|z_{1}\rangle|\mathcal{J}_{1}\rangle,
\]
where $\mathcal{J}_{1}$ and $\mathcal{J}_{2}$ are from the corresponding
set of 25 connected components. Recall that $H_{V}^{j}$ acts nontrivially
on qubits $a_{j}$ and $b_{j}$ and as the identity on the remaining
$2^{n + n_{a}-2}$ qubits of the first register. Using this fact we
see that the above matrix (the block labeled by $j$) further decomposes into $2^{n + n_{a}-2}$
identical blocks each of which has size $4\cdot25\times4\cdot25$.
Each of these $100\times100$ blocks is a matrix that we have already
encountered in Section \ref{sec:Hamiltonians-with-two}, the matrix
with entries given by \eqref{eq:mat_HV}. Lemma \ref{lem:H_V}
characterizes the zero energy eigenvectors of this matrix. Applying
Lemma~\ref{lem:H_V},
we get zero eigenvectors of \eqref{eq:first_op} 

\begin{align}
	\sum_{x,y\in\{0,1\}}\Bigg[ 
			& \left(|xy\rangle\langle xy|_{a_{j}b_{j}}\right)|z\rangle|\mathcal{R}_{x}^{j}\rangle 
			+ \left(Q|xy\rangle\langle xy|_{a_{j}b_{j}}\right)|z\rangle|\mathcal{G}_{x}^{j}\rangle \nonumber\\
 			& + \left(\tilde{Q}|xy\rangle\langle xy|_{a_{j}b_{j}}\right)|z\rangle|\mathcal{Y}_{x}^{j}\rangle 
 			+ \left(V|xy\rangle\langle xy|_{a_{j}b_{j}}\right)|z\rangle|\mathcal{B}_{x}^{j}\rangle
 		\Bigg].
 	\label{eq:psi_states}
\end{align}
Here the projector $|xy\rangle\langle xy|_{a_{j}b_{j}}$ acts nontrivially
only on qubits $a_{j}$ and $b_{j}$ of the computational register
and $z$ is an $\left(n + n_{a}\right)$-bit string. Note that since
$|z\rangle$ is a computational basis state, only one of the terms
in the sum over $x,y$ is nonzero; we have written the state in this
way to ease understanding later on. Letting $z$ range over all $(n + n_{a})$-bit
strings and $j = 0,\dots,M-1$, the states \eqref{eq:phi_k_m_l} and \eqref{eq:psi_states}
span the groundspace of \eqref{eq:first_op}. 

Now consider 
\begin{equation}
	\iii \otimes  H_{\text{diag}}^{(M)} 
		+ \sum_{j = 0}^{M-1}H_{V}^{j} 
		+ \sum_{j = 0}^{M-1}H_{U}^{j}.
	\label{eq:intermediate_Ham}
\end{equation}
The third term in \eqref{eq:intermediate_Ham} couples the ground states of \eqref{eq:first_op}.
To solve for the zero energy states of \eqref{eq:intermediate_Ham}, we compute the action of 
\begin{equation}
	\sum_{j = 0}^{M-1}H_{U}^{j}
	\label{eq:sum_of_Hus}
\end{equation}
 within the groundspace of the first two terms. We now exhibit a basis for the ground space of \eqref{eq:first_op} in which the operator \eqref{eq:sum_of_Hus} has a simple form.  Define unitaries
\begin{align*}
	O^{0} &  = \iii \qquad \text{and} \qquad 
	O^{k} = \text{ }V_{a_{k-1}b_{k-1}}U^{k-1}V_{a_{k-2}b_{k-2}}U^{j-2}\dots U^{0}
\end{align*}
for $k = 1,\dots,M-1$ and states
\begin{equation}
	|\mathcal{K}^{0}(\phi)\rangle = |\phi\rangle|\mathcal{K}^{0}\rangle, \qquad
	|\mathcal{M}^{j}(\phi)\rangle = O^{j}|\phi\rangle|\mathcal{M}^{j}\rangle, \qquad
	|\mathcal{L}^{j}(\phi)\rangle = U^{j}O^{j}|\phi\rangle|\mathcal{L}^{j}\rangle,
	\label{eq:Basis_K0}
\end{equation}
and 
\begin{align}
	|\psi_{V}^{j}(\phi)\rangle = \sum_{x,y\in\{0,1\}}\Bigg[
		& \left(|xy\rangle\langle xy|_{a_{j}b_{j}}\right)
				U^{j}O^{j}|\phi\rangle|\mathcal{R}_{x}^{j}\rangle 
		+ \left(Q|xy\rangle\langle xy|_{a_{j}b_{j}}\right)
				U^{j}O^{j}|\phi\rangle|\mathcal{G}_{x}^{j}\rangle\\
 		& + \left(\tilde{Q}|xy\rangle\langle xy|_{a_{j}b_{j}}\right)
 				U^{j}O^{j}|\phi\rangle|\mathcal{Y}_{x}^{j}\rangle 
 		+ \left(V|xy\rangle\langle xy|_{a_{j}b_{j}}\right)
 				U^{j}O^{j}|\phi\rangle|\mathcal{B}_{x}^{j}\rangle
 	\Bigg].
 	\label{eq:psiV}
\end{align}
Here we let $|\phi\rangle$
range over some (arbitrary) complete orthonormal
basis $\Lambda$ for the $(n + n_{a})$ qubit register and $j = 0,\dots,M-1$.
The states \eqref{eq:Basis_K0}-\eqref{eq:psiV} (with $j = 0,\dots,M-1$
and $|\phi\rangle\in\Lambda$) are linearly independent superpositions
of \eqref{eq:phi_k_m_l} and \eqref{eq:psi_states} and therefore span
the groundspace of \eqref{eq:first_op}. Their normalizations are:
\begin{align}
	\langle\mathcal{K}^{0}(\phi)|\mathcal{K}^{0}(\phi)\rangle &= 7, \nonumber\\
	\langle\mathcal{M}^{j}(\phi)|\mathcal{M}^{j}(\phi)\rangle 
				= \langle\mathcal{L}^{j}(\phi)|\mathcal{L}^{j}(\phi)\rangle &= 4, \label{eq:norm1}\\
	\langle\psi_{V}^{j}(\phi)|\psi_{V}^{j}(\phi)\rangle &= 43. \nonumber
\end{align}

The operator \eqref{eq:sum_of_Hus} acts on this basis in a simple way.
It only connects states with the same $|\phi\rangle\in\Lambda$ and
is therefore block diagonal (with $2^{n + n_{a}}$ blocks).

We compute
the matrix elements within a block using equations \eqref{eq:H_Uj_def}
and \eqref{eq:hi_iplus1}. For example, 
\begin{align*}
	\sum_{j = 0}^{M-1} \langle\mathcal{K}^{0}(\phi)|H_{U}^{j}|\mathcal{M}^{0}(\phi)\rangle 
		&  = \langle\mathcal{K}^{0}(\phi)|H_{U}^{0}|\mathcal{M}^{0}(\phi)\rangle\\
 		&  = \langle\mathcal{K}^{0}(\phi)|
 				\left(\iii \otimes \iii \otimes h_{3,4}\right)
 				|\mathcal{M}^{0}(\phi)\rangle 
   = -\frac{1}{4}.
\end{align*}
Continuing in this manner, we compute all matrix elements of \eqref{eq:sum_of_Hus}
between states \eqref{eq:Basis_K0}-\eqref{eq:psiV}. The resulting matrix
is the same for each $|\phi\rangle$ and is equal to $\frac{1}{4}L$
where $L$ is the Laplacian\footnote{Recall that the Laplacian matrix of a simple graph with vertex set $V$ and
edge set $E$ is the $|V|\times|V|$ matrix with 
\[
L_{u,v}=\begin{cases}
-1 & \text{ if }(u,v)\in E\\
d(u) & \text{ if }u=v\\
0 & \text{otherwise}
\end{cases}
\]
 where $d(u)$ is the degree of vertex $u$.}
of the graph in Figure~\ref{fig:laplaciangraph}.
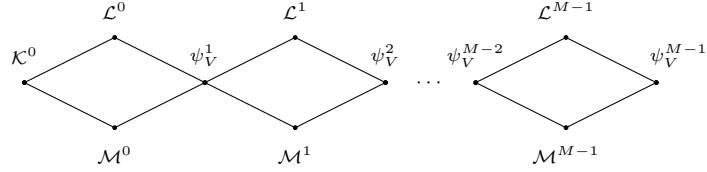
\begin{figure}
\center
{\scalefont{0.65}
\begin{tikzpicture}[scale = 0.6,vertex/.style = {circle,draw = black,fill = black,inner sep = 0.5pt,minimum size  = 0mm}]

\draw (0,0)node[vertex]{}--(2,1)node[vertex]{}--(4,0)node[vertex]{}--(6,1)node[vertex]{}--(8,0)node[vertex]{};
\draw (0,0)node[vertex]{}--(2,-1)node[vertex]{}--(4,0)node[vertex]{}--(6,-1)node[vertex]{}--(8,0)node[vertex]{};
\draw (10,0) node[vertex]{}--(12,1)node[vertex]{}--(14,0)node[vertex]{};
\draw (10,0) node[vertex]{}--(12,-1)node[vertex]{}--(14,0)node[vertex]{};
\draw (9,0)node[]{$\cdots$};
\draw (0,0.6) node[]{$\mathcal{K}^{0}$};
\draw (2,1.6) node[]{$\mathcal{L}^{0}$};
\draw (2,-1.6) node[]{$\mathcal{M}^{0}$};
\draw (4,0.6) node[]{$\psi_{V}^{1}$};
\draw (6,1.6) node[]{$\mathcal{L}^{1}$};
\draw (6,-1.6) node[]{$\mathcal{M}^{1}$};
\draw (8,0.6) node[]{$\psi_{V}^{2}$};
\draw (10,0.6)node[]{$\psi_{V}^{M-2}$};
\draw (12,1.6) node[]{$\mathcal{L}^{M-1}$};
\draw (12,-1.6) node[]{$\mathcal{M}^{M-1}$};
\draw (14.5,0.6)node[]{$\psi_{V}^{M-1}$};
\end{tikzpicture}
}

\caption{The Hamiltonian \eqref{eq:sum_of_Hus} is block diagonal when written
in the basis \eqref{eq:Basis_K0}-\eqref{eq:psiV}. Each of the $2^{n + n_{a}}$
blocks corresponds to a different $(n + n_{a})$-qubit state $|\phi\rangle$
from an orthonormal basis $\Lambda$. The matrix for each block is
the same, equal to $\frac{1}{4}L$ where $L$ is the Laplacian of this graph, which has $3M+1$ vertices.}

\label{fig:laplaciangraph}
\end{figure}
 The Laplacian matrix of a connected graph has a unique eigenvector
with eigenvalue zero: the all ones vector. This fact means that for
each $|\phi\rangle\in\Lambda$ there is a unique zero energy state
of \eqref{eq:intermediate_Ham} given by the uniform superposition
\begin{align}
	|\text{Hist(\ensuremath{\phi})}\rangle 
	&  = \frac{1}{\sqrt{51M + 7}}\left(
			|\mathcal{K}^{0}(\phi)\rangle 
			+ \sum_{j = 0}^{M-1}\left(
					|\mathcal{L}^{j}(\phi)\rangle 
					+ |\mathcal{M}^{j}(\phi)\rangle 
					+ |\psi_{V}^{j}(\phi)\rangle
			\right)
		\right),
	\label{eq:history_states_Hc}
\end{align}
 where we used \eqref{eq:norm1} to compute the normalization.
Letting $|\phi\rangle$ range over all states in the basis $\Lambda$
we get a spanning basis for the groundspace of \eqref{eq:intermediate_Ham}.
Moreover, we claim that \emph{every} state in the groundspace of \eqref{eq:intermediate_Ham}
is of the form 
\[
	|\text{Hist}(\psi)\rangle
\]
 for some $(n + n_{a})$-qubit state $\ket{\psi}.$ This is because,
for any two basis states $|\phi\rangle,|\phi'\rangle\in\Lambda$ and
complex numbers $\alpha,\beta$ 
\[
	\alpha|\text{Hist(\ensuremath{\phi})}\rangle + \beta|\text{Hist}(\phi')\rangle 
	= |\text{Hist}(\alpha\phi + \beta\phi')\rangle,
\]
 which can be seen using the definitions \eqref{eq:Basis_K0}-\eqref{eq:psiV}.

Now consider the conditions under which a state $|\text{Hist}(\psi)\rangle$
in the groundspace of \eqref{eq:intermediate_Ham} also has zero energy
for both $H_{\text{init}}$ and $H_{\text{end}}$, the final two terms in \eqref{eq:H_C}.

We have 
\begin{align}
	\langle\text{Hist}(\psi)|H_{\text{init}}|\text{Hist}(\psi)\rangle 
	&  = \frac{1}{\left(51M + 7\right)}
			 \langle\mathcal{K}^{0}(\psi)|
			 		\sum_{i = 1}^{n_{a}} \ket{1}\langle 1|_{n + i}
					\otimes C_{\leq1}\otimes C_{\leq1}
					|\mathcal{K}^{0}(\psi)\rangle
					\label{eq:H_init_Hist}\\
 &  = \frac{1}{\left(51M + 7\right)}
 			 \langle\psi|\langle C_{1}|\langle C_{1}|
 			 		\sum_{i = 1}^{n_{a}} \ket{1}\langle 1|_{n + i}
 					\otimes C_{\leq1}\otimes C_{\leq1}
 					\ket{\psi}|C_{1}\rangle|C_{1}\rangle \nonumber \\
 &  = \frac{1}{4\left(51M + 7\right)}
 				\langle\psi| \sum_{i = 1}^{n_{a}}  \ket{1}\langle 1|_{n + i}\ket{\psi},\nonumber 
\end{align}
where in the last line we used \eqref{eq:1localproj2}. This
is equal to zero if and only if $\ket{\psi} = |W\rangle|0\rangle^{\otimes n_{a}}$
for some $n$-qubit state $|W\rangle$. Similarly, 
\begin{equation}
\langle\text{Hist}(\psi)|H_{\text{end}}|\text{Hist}(\psi)\rangle = \frac{1}{4\left(51M + 7\right)}\langle\psi|U_{X}^{\dagger}\left(|0\rangle\langle 0|_{n + 1}\right)U_{X}\ket{\psi},\label{eq:H_end_Hist}
\end{equation}
which is zero if and only if the $(n + 1)$th qubit of $U_{X}\ket{\psi}$
is in the state $\ket{1}$ with certainty. We have therefore proven
that $H_{X}$ has a zero energy eigenstate if and only if there exists
an $n-$qubit state $|W\rangle$ satisfying 
\begin{equation}
\left\Vert \left(1^{\otimes n}\otimes|0\rangle\langle 0|_{(n + 1)} \otimes \iii ^{\otimes(n_{a}-1)}\right)U_{X}|W\rangle|0\rangle^{\otimes n_{a}}\right\Vert ^{2} = 0.\label{eq:condition-1}
\end{equation}
This establishes Theorem \ref{thm:Completeness}.


\subsection{Proof of Theorem \ref{thm:Soundness} \label{sub:Soundness_Section}}

We now consider the case where 
\begin{equation}
	\text{AP}(U_{X},|W\rangle)\leq\frac{1}{3}\quad\text{for all}\quad|W\rangle.
	\label{eq:condition_of_thm2}
\end{equation}
By Theorem \ref{thm:Completeness} we know that in this case $H_{X}$
has no zero eigenvalues. Here we prove that its ground energy is $\Omega\left(\frac{1}{M^{6}}\right)$
which gives us Theorem \ref{thm:Soundness} since $M = \Theta(g)$.

We use a spectral bound due to Kitaev \cite{Kitaev:2002:CQC:863284}.
Write 
\[
	\gamma(H)
\]
for the smallest nonzero eigenvalue of a finite-dimensional, positive
semidefinite operator $H$. 
\begin{lem}\label{lem:KitLemma}
[Kitaev's geometric Lemma. \cite{Kitaev:2002:CQC:863284}, p.147]
\label{lem:geom_lemma_rephrased}Let $H = H_{A} + H_{B}$ with $H_{A}\geq0$
and $H_{B}\geq0$. Let $S$ be the nullspace of $H_{A}$ and let $\Pi_{B}$
be the projector onto the nullspace of $H_{B}$. Suppose $H$ has
no zero eigenvalues. Then
\begin{align}
	\gamma(H)\geq\min\{\gamma(H_{A}),\gamma(H_{B})\}\cdot\left(1-\sqrt{c}\right),
	\label{eq:KitLemma}
\end{align}
 where 
\begin{align}
	c = \max_{|v\rangle\in S\,:\,\langle v|v\rangle = 1}\langle v|\Pi_{B}|v\rangle.
	\label{eq:cdef}
\end{align}
\medskip{}
 
\end{lem}
Note that we use the notation $H\geq0$ to indicate that all eigenvalues
of a matrix $H$ are real and nonnegative.

Write $H_X =H_A + H_B$, with
\begin{align*}
	H_{A} &  = \iii \otimes  H_{\text{diag}}^{(M)} 
		+ \sum_{j = 0}^{M-1}\left(H_{V}^{j} + H_{U}^{j}\right), \\
H_{B} &  = H_{\text{init}} + H_{\text{end}}.
\end{align*}
We use Lemma~\ref{lem:KitLemma} to lower bound $\gamma(H_{X})$
which is equal to the ground energy of $H_{X}$.

To apply Lemma~\ref{lem:KitLemma},
we need lower bounds on $\gamma(H_{A})$ and $\gamma(H_{B})$.
To bound $\gamma(H_{B})$ we use equations \eqref{eq:C_geq} and \eqref{eq:C_leq},
which say that the operators 
\[
	C_{\leq1}\otimes C_{\leq1}
	\quad
	\text{and}
	\quad 
	C_{\geq(9M + 3)}\otimes C_{\geq(9M + 3)}
\]
are diagonal in the computational basis. 
Looking at the expressions
for $H_{\text{init}}$ and $H_{\text{end}}$ from \eqref{eq:H_init_H_end},
we see that $H_{B}$ is a sum of projectors diagonal in the computational
basis and hence $\gamma(H_{B})\geq1$. 
We prove the following lower bound on $\gamma(H_{A})$
in Appendix \ref{sec:Proof-of-Theorem}:
\begin{lem}
\label{lem:complicated_gap_bound}
\[
	\gamma(H_A) = \gamma\left(\iii \otimes  H_{\textnormal{diag}}^{(M)} 
		+ \sum_{j = 0}^{M-1}\left(H_{V}^{j} 
		+ H_{U}^{j}\right)\right) 
	= \Omega\left(\frac{1}{M^{5}}\right).
\]
\end{lem}
We use these bounds in \eqref{eq:KitLemma} to obtain
\begin{equation}
	\gamma(H_{X})\geq\Omega\left(\frac{1}{M^{5}}\right)\cdot\left(1-\sqrt{c}\right),
	\label{eq:H_x_bnd}
\end{equation}
with $c$ from \eqref{eq:cdef}.

We now upper bound $c$. Write $\Pi_{\text{init}}$ and $\Pi_{\text{end}}$
for the projectors onto the zero eigenspaces of $H_{\text{init}}$
and $H_{\text{end}}$ respectively. We have 
\begin{align}
	\Pi_{\text{init}} 
			&  = \left(1^{\otimes n}\otimes|0\rangle\langle 0|^{\otimes n_{a}}\right)\otimes C_{\leq1}\otimes C_{\leq1} 
					+ \iii \otimes (1-C_{\leq1}\otimes C_{\leq1}),
					\label{eq:Piinit} \\
	\Pi_{\text{end}} 
			&  = \ket{1}\langle 1|_{n + 1}\otimes C_{\geq 9M + 3}\otimes C_{\geq 9M + 3} 
					+ \iii \otimes \left(1-C_{\geq 9M + 3}\otimes C_{\geq 9M + 3}\right).
					\label{eq:Piend}
\end{align}
Since these two projectors are both diagonal in the computational
basis, $\Pi_{B}$ is simply 
\[
	\Pi_{\text{init}}\cdot\Pi_{\text{end}}.
\]
Recall from 
Section~\ref{sub:The-zero-energy} that every state in the nullspace
of $H_{A}$ is of the form $|\text{Hist}(\psi)\rangle$ for some $(n + n_{a})$-qubit
state $\ket{\psi}$. Thus
\begin{align}
	c &  = \max_{\ket{\psi}\,:\,\langle\psi\ket{\psi}=1}
			\langle\text{Hist}(\psi)| 
					\Pi_{\text{init}} \cdot \Pi_{\text{end}}
					|\text{Hist}(\psi)\rangle\nonumber \\
 & \leq\max_{\ket{\psi}\,:\,\langle\psi\ket{\psi}=1}
 		\langle\text{Hist}(\psi)|
 					\Pi_{\text{init}}
 					|\text{Hist}(\psi)\rangle^{\frac{1}{2}}
 	 \cdot \langle\text{Hist}(\psi)|
 	 				\Pi_{\text{end}}
 	 				|\text{Hist}(\psi)\rangle^{\frac{1}{2}}
	\label{eq:c_equation}
\end{align}
 by the Cauchy-Schwarz inequality. Using equations \eqref{eq:history_states_Hc},
\eqref{eq:Piinit} and \eqref{eq:Piend} we compute 
\begin{align*}
	\langle\text{Hist}(\psi)|\Pi_{\text{init}}|\text{Hist}(\psi)\rangle 
	&  = \frac{1}{4\left(51M + 7\right)}\langle\psi|\iii \otimes |0\rangle\langle 0|^{\otimes n_{a}}\ket{\psi} 
		+ 1 -\frac{1}{4\left(51M + 7\right)}, \\
	\langle\text{Hist}(\psi)|\Pi_{\text{end}}|\text{Hist}(\psi)\rangle 
	&  = \frac{1}{4\left(51M + 7\right)}\langle\psi|U_{X}^{\dagger}\ket{1}\langle 1|_{n + 1}U_{X}\ket{\psi} 
		+ 1 -\frac{1}{4\left(51M + 7\right)}.
\end{align*}

Write 
\[
	\ket{\psi} = |\psi_{1}\rangle + |\psi_{2}\rangle,
\]
 where 
\[
	|\psi_{1}\rangle = \iii \otimes |0\rangle\langle 0|^{\otimes n_{a}}\ket{\psi}, \qquad
	|\psi_{2}\rangle = \left(1-\iii \otimes |0\rangle\langle 0|^{\otimes n_{a}}\right)\ket{\psi}.
\]
Let $f = \langle\psi_{1}|\psi_{1}\rangle$, which implies
$\langle\psi_{2}|\psi_{2}\rangle = 1-f$. Then
\begin{equation}
	\langle\text{Hist}(\psi)|\Pi_{\text{init}}|\text{Hist}(\psi)\rangle 
		= 1- \frac{1}{4\left(51M + 7\right)}\left(1-f\right)
	\label{eq:psi_piinit_psi}
\end{equation}
Furthermore, $|\psi_{1}\rangle$ is proportional to a state of the
form $|W\rangle|0\rangle^{\otimes n_{a}}$ and so 
\[
	\langle\psi_{1}|U_{X}^{\dagger}\ket{1}\langle 1|_{n + 1}U_{X}|\psi_{1}\rangle
	\leq
	\frac{1}{3}\langle\psi_{1}|\psi_{1}\rangle = \frac{1}{3}f
\]
by equation \eqref{eq:condition_of_thm2}. Hence 
\begin{align}
	\langle\psi|U_{X}^{\dagger}\ket{1}\langle 1|_{n + 1}U_{X}\ket{\psi} 
		& \leq \frac{1}{3}f 
			+ \langle\psi_{2}|U_{X}^{\dagger}\ket{1}\langle 1|_{n + 1}U_{X}|\psi_{2}\rangle 
			+ 2\left|\langle\psi_{1}|U_{X}^{\dagger}\ket{1}\langle 1|_{n + 1}U_{X}|\psi_{2}\rangle\right| \nonumber \\
 		& \leq \frac{1}{3}f + \left(1-f\right) 
 			+ 2\left(\langle\psi_{1}|U_{X}^{\dagger}
 					\ket{1}\langle 1|_{n + 1}
 					U_{X}|\psi_{1}\rangle\langle\psi_{2}|\psi_{2}\rangle\right)^{\frac{1}{2}} \nonumber \\
 & \leq 1 - \frac{2}{3}f 
 			+ 2\sqrt{\frac{1}{3}f\left(1-f\right)}.
 \label{eq:psi_Ux_psi}
\end{align}
Plugging equations \eqref{eq:psi_piinit_psi} and \eqref{eq:psi_Ux_psi}
into \eqref{eq:c_equation} gives
\begin{align*}
	c & \leq \max_{0\leq f\leq1}\left(
				1-\frac{1}{4\left(51M + 7\right)}\left(1-f\right)
			\right)^{\frac{1}{2}}
			\left(
				1-\frac{1}{4\left(51M + 7\right)}\left(\frac{2}{3}f-2\sqrt{\frac{1}{3}f\left(1-f\right)}\right)
			\right)^{\frac{1}{2}}\\
 		& \leq\max_{0\leq f\leq1}\left(
 				1-\frac{1}{4\left(51M + 7\right)}\left(1-\frac{1}{3}f-2\sqrt{\frac{1}{3}f\left(1-f\right)}\right) 
 				+ \frac{1}{16\left(51M + 7\right)^{2}}\cdot\frac{2}{3}
 			\right)^{\frac{1}{2}}
\end{align*}
 Now using the fact that 
\[
	\min_{0\leq f\leq1}\left(1-\frac{1}{3}f - 2\sqrt{\frac{1}{3}f\left(1-f\right)}\right) 
	= 0.23\dots \geq \frac{1}{5},
\]
we get
\begin{align*}
	c & \leq \left(
				1-\frac{1}{4\left(51M + 7\right)}\cdot\frac{1}{5} + \frac{1}{16\left(51M + 7\right)^{2}}\cdot\frac{2}{3}
			\right)^{\frac{1}{2}}\\
 & \leq 1 - \frac{1}{4\left(51M + 7\right)}\cdot\frac{1}{10} 
 			+ \frac{1}{16\left(51M + 7\right)^{2}}\cdot\frac{1}{3}
\end{align*}
 (using $\sqrt{1 + x}\leq1 + \frac{x}{2}$). Plugging this into equation
\eqref{eq:H_x_bnd} gives 
\[
	\gamma\left(H_{X}\right) = \Omega\left(\frac{1}{M^{6}}\right),
\]
 which completes the proof. 


%


\section{Acknowledgments}

We thank Sergey Bravyi, Daniel Gottesman, Rajat Mittal, and John Watrous
for helpful discussions. We thank Sergey Bravyi for clarifying his
proof that quantum $k$-SAT is contained in QMA$_{1}$ \cite{Bravyi_communication}.
We thank John Watrous for suggesting the use of ``exact heralded
measurements'' and for helpful suggestions that led us to the specific
measurement procedure used in Section \ref{sec:containmentINqma1}.
DG was supported in part by NSERC. DN gratefully acknowledges support by the EU project QUERG.

\bibliographystyle{plain}	
\bibliography{q3sat_refs}


\appendix

\section{Quantum $3$-SAT is contained in QMA$_{1}$\label{sec:containmentINqma1}}

In this Section we show that quantum $3$-SAT is contained in QMA$_{1}$.
Bravyi gave a protocol that shows that quantum $k$-SAT is contained
in QMA$_{1}$ (for all constant $k$) but there is an error in the
proof of Lemma 5 of that paper (on which the result relies) \cite{Bravyi06}.
Bravyi's proof can be patched up, but this requires some small modifications
to the original definitions of QMA$_{1}$ and quantum $k$-SAT
\cite{Bravyi_communication}. Specifically, Bravyi's proof is valid
with the definition of QMA$_{1}$ given in Definition \ref{QMA_one}
and under a condition on the set of allowed projectors which may appear
in an instance of quantum $k$-SAT \cite{Bravyi_communication}. Here
we give a self-contained proof (following Bravyi \cite{Bravyi06}) that quantum $3$-SAT is contained
in QMA$_{1}$, with the definitions given in Section \ref{sec:Preliminaries-and-Overview}.

To show that quantum $3$-SAT is in QMA$_{1}$ we provide a verification
procedure that uses ancillae prepared in the state $|0\rangle$, measurements
in the computational basis, as well as operations controlled on the
measurement outcomes. All gates are from the set $\mathcal{G}$ \eqref{eq:Gset}.

As discussed in Section \ref{sec:Preliminaries-and-Overview}, such a
verification procedure can be efficiently converted into a QMA$_{1}$
verification procedure which uses only one measurement in the computational
basis (the final measurement which determines acceptance).

Consider the following simple verification procedure. Given an instance
of quantum $3$-SAT specified as a collection $\{\Pi_{i}:\, i = 1,\dots,r\}$
of projectors and a witness state $|W\rangle$, Arthur first chooses
a random integer $j\in\{1,\dots,r\}$. Note that using gates from $\mathcal{G}$,
Arthur can make the random choice very close to uniformly distributed.
For example he can choose $j$ so that
\[
	\text{Pr}[j = i] = \frac{1}{r} + \mathcal{O}\left(r2^{-r}\right)
\]
for each $i\in\{1,\dots,r\}$. (To do this he can prepare $r$ ancillae
in the state $|0\rangle^{\otimes r}$, apply Hadamard gates to each
of them, and then measure each in the computational basis. He can
divide the $2^{r}$ possible measurement outcomes into $r$ sets $i = 1,\dots,r$,
where an outcome from the $i$th set occurs with probability $p_{i} = \frac{1}{r} + \epsilon_{i}$
and $|\epsilon_{i}|\leq r2^{-r}$.) Having chosen $j$ randomly in
this way, Arthur measures the eigenvalue of the projector $\Pi_{j}$
in the state $|W\rangle$ and outputs ``yes'' if and only if the
measurement outcome is zero. For a yes instance, Merlin can give Arthur
a witness state $|W\rangle$ which is in the nullspace of each of
the projectors $\Pi_{i}$, and in this case Arthur outputs ``yes''
with certainty. On the other hand, for a no instance, any state $|W\rangle$ satisfies 
\[
	\sum_{i = 1}^{r} \langle W|\Pi_{i}|W\rangle\geq 1,
\]
 and so the probability that Arthur's measurement outcome is $1$
is given by 
\begin{align*}
	\sum_{i = 1}^{r}\text{Pr}[j = i] \, \langle W|\Pi_{i}|W\rangle 
		=  & \frac{1}{r}\sum_{i = 1}^{r}\langle W|\Pi_{i}|W\rangle 
			+ \mathcal{O}\left(r^{2}2^{-r}\right)\\
	\geq & \frac{1}{r} + \mathcal{O}\left(r^{2}2^{-r}\right).
\end{align*}
The probability that Arthur outputs ``yes'' in this case is upper
bounded by 
\[
	1-\frac{1}{r} + \mathcal{O}\left(r^{2}2^{-r}\right).
\]
This protocol has perfect completeness, and soundness $1-\frac{1}{r} + \mathcal{O}\left(r^{2}2^{-r}\right)$.
The soundness can be amplified (for example using the protocol described
in reference \cite{Kitaev:2002:CQC:863284}) to achieve the desired
constant $\frac{1}{3}$ which appears in our definition of QMA$_{1}$.

The above verification procedure requires that we can exactly measure
any projector $\Pi\in\mathcal{P}$ using the gate set $\mathcal{G}$. As
discussed below, this is only true for projectors satisfying condition
1. in Definition \ref{Delta_proj}. In the following, we describe a
modification of this verification procedure that weakens this requirement
and which proves that quantum $3$-SAT (with projectors from $\mathcal{P}$)
is contained in QMA$_{1}$.

First, let us show that an exact measurement can be performed for projectors
$\Pi\in\mathcal{P}$ satisfying condition 1 in Definition \ref{Delta_proj}.
For such projectors the controlled unitary 
\begin{equation}
	U_{\Pi} = \Pi\otimes\sigma^{x} + \left(1-\Pi\right) \otimes \iii \label{eq:controlled_U}
\end{equation}
is efficiently decomposable as sequences of $\{\widehat{H},T,\text{CNOT\}}$
gates (possibly using an ancilla initialized in the state $|0\rangle$).
This follows from reference \cite{Giles}, in which it is shown that
such a decomposition can be efficiently computed for any $3$-local
projector with matrix elements (in the computational basis) of the
form 
\begin{equation}
	\frac{1}{2^{s}}\left(a + i\,b + \sqrt{2}\, c + i\sqrt{2}\, d\right)
	\label{eq:matrix_elements_ring}
\end{equation}
for any $n$-independent constant $s\in\{2,3,4,\dots\}$ and integers
$a,b,c,d$. Note that a measurement
of $\Pi$ in a state $\ket{\psi}$ is performed by applying $U_{\Pi}$
to $\ket{\psi}|0\rangle$ and then measuring the last qubit in
the computational basis.

For projectors satisfying condition 2 it is possible to implement
a weaker type of exact measurement that we call an ``exact heralded
measurement''. This is a probabilistic implementation of the controlled
unitary \eqref{eq:controlled_U} using the gate set $\mathcal{G}$,
where the success or failure of the procedure is determined by the
outcome of a single qubit measurement of an ancilla in the computational
basis. The measurement is ``heralded'' because, after using the
algorithm, we are told (via the state of the ancilla qubit) whether
or not it succeeded. In the next Section we provide an exact heralded
measurement algorithm that can be run using any number of steps $m$
and achieves success probability $1-\mathcal{O}\left(4^{-m}\right)$. 

The following modification of the above verification procedure shows
that quantum $3$-SAT is contained in QMA$_{1}$ (with the definitions
given in Section \ref{sec:Preliminaries-and-Overview}). Arthur first
randomly chooses an integer $j\in\{1,\dots,r\}$ as described above.
If $\Pi_{j}$ is a projector satisfying condition 1 from Definition
\ref{Delta_proj} then he measures the projector using the exact measurement
procedure described above, and outputs ``yes'' if and only if the
measurement outcome is zero. On the other hand, if $\Pi_{j}$ satisfies
condition 2, then Arthur uses the exact heralded measurement algorithm
with $m = r$ steps. The probability that the measurement succeeds is
$1-\mbox{\ensuremath{\mathcal{O}}}\left(4^{-r}\right)$. In this case
Arthur outputs ``yes'' if either a) the 
heralded algorithm fails
or b)
the measurement succeeds and the measurement outcome is $0$. This
protocol has perfect completeness, and the very small failure probability
does not appreciably affect the soundness, which can still be amplified
to the constant $\frac{1}{3}$.


\subsection{An exact heralded measurement}

Consider a 3-local projector of the form 
\begin{equation}
	\Pi = U\left(
			\sqrt{\frac{1}{3}}|000\rangle-\sqrt{\frac{2}{3}}|001\rangle
		\right)
		\left(
			\sqrt{\frac{1}{3}}\langle 000|-\sqrt{\frac{2}{3}}\langle 001|
		\right)
		U^{\dagger},
		\label{eq:Pi_proj}
\end{equation}
where $U$ is a 3-qubit unitary with matrix elements of the form \eqref{eq:mat_el_ring}.
In this Section we describe an algorithm which uses gates from the
set $\{\widehat{H},T,\text{CNOT}\}$, ancillae prepared in the state
$|0\rangle,$ and measurements in the computational basis. The algorithm
takes as input a four-qubit state $|\alpha\rangle$ and outputs a
single classical bit which indicates either ``success'' or ``failure''
along with a four-qubit state $|\beta\rangle.$ If the algorithm outputs
``success'' then 
\[
	|\beta\rangle = U_{\Pi}|\alpha\rangle
\]
 and we have successfully and exactly implemented $U_{\Pi}$ (from \eqref{eq:controlled_U}).
The algorithm we describe can be performed using
any number of steps $m\in\{1,2,3,\dots\}$ and uses $\Theta(m)$ gates
along with $\Theta(m)$ ancillae and measurements in the computational
basis. The success probability of the $m$-step algorithm is 
\[
	p_{\text{success}}(m) = 1-\mathcal{O}\left(\frac{1}{4^{m}}\right).
\]
 We thank John Watrous for his help with developing this algorithm. 

\begin{figure}
	\includegraphics[width=8cm]{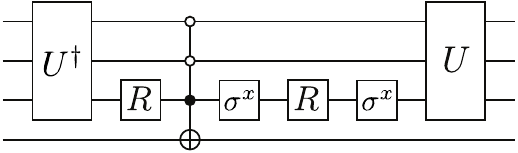}
	\caption{A circuit that implements the controlled unitary $U_{\Pi}$ from equation
	\eqref{eq:controlled_U} for a projector of the form \eqref{eq:Pi_proj}.
	Here the controlled operation applies $\sigma^{x}$ to the fourth
	qubit if the first three qubits are in the state $|001\rangle$. 
	The single qubit unitary $R$ is given by \eqref{eq:W_unitary}.
	\label{fig:U_Pi_circuit}}
\end{figure}
Our algorithm is based on the circuit in Figure~\ref{fig:U_Pi_circuit}
which implements $U_{\Pi}$. In the Figure, the four-qubit controlled
unitary is
\[
	|001\rangle\langle 001|\otimes\sigma^{x} 
	+ \left(\iii-|001\rangle\langle 001|\right) \otimes \iii 
\]
and 
\begin{equation}
	R = \frac{1}{\sqrt{3}}\left(\begin{array}{cc}
		\sqrt{2} & -1\\
		1 & \sqrt{2}
	\end{array}\right).\label{eq:W_unitary}
\end{equation}
Note that, with the exception of the single-qubit $R$ gates, all
gates which appear in the circuit in Figure~\ref{fig:U_Pi_circuit}
can be expressed as products of $\widehat{H},T$ and $\text{CNOT}$
gates (their matrix elements are of the form \eqref{eq:mat_el_ring}).
To implement each of the $R$ gates, 
we use a heralded $m$-step procedure described below. Each $R$ gate succeeds with probability 
\[
	1-\frac{1}{4^{m}}.
\]
The overall probability that we apply the circuit in Figure~\ref{fig:U_Pi_circuit}
is then $\left(1-\frac{1}{4^{m}}\right)^{2}$.


\subsection*{An $m$-step algorithm to implement $R$}

Consider the following 2-qubit unitary

\[
	A = \frac{1}{2}\left(\begin{array}{rrrr}
		\sqrt{2} & -1 & 1 & 0\\
		1 & \sqrt{2} & 0 & 1\\
		1 & 0 & -\sqrt{2} & -1\\
		0 & 1 & 1 & -\sqrt{2}
	\end{array}\right).
\]
Since each
of its matrix elements is of the form \eqref{eq:matrix_elements_ring},
$A$ can be implemented exactly using the gate set $\{\widehat{H},T,\text{CNOT}\}$
using one ancilla initialized in the $|0\rangle$ state \cite{Giles}.
That is to say, the results of reference \cite{Giles} can be used
to obtain a sequence of $\widehat{H}$, $T$ and $\text{CNOT }$gates
implementing a 3-qubit unitary $\tilde{U}$ such that 
\[
	\tilde{U}|0\rangle\ket{\psi} = |0\rangle\left(A\ket{\psi}\right)
\]
for all $2$-qubit states $\ket{\psi}$. Let $|\phi\rangle$ be
a one-qubit state that we wish to apply $R$ to. The $m$-step algorithm
consists of applying the following procedure until either a) the algorithm
terminates or b) step 2 has been repeated $m$ times. 
In case a) the
algorithm succeeds and in case b) it fails.
\begin{enumerate}
\item Apply $\tilde{U}$ to the state
	\begin{equation}
		|00\rangle|\phi\rangle\label{eq:startstate}
	\end{equation}
 to get 
\[
	\tilde{U}|00\rangle|\phi\rangle = |0\rangle\left(A|0\rangle|\phi\rangle\right).
\]

\item Measure the second qubit in the computational basis to obtain an outcome
$z\in\{0,1\}$. Looking at the top left $2\times2$ submatrix of $A$
we see that with probability 
\[
	p = \left\Vert \left(|0\rangle\langle 0| \otimes \iii \right)A|0\rangle|\phi\rangle\right\Vert ^{2} = \frac{3}{4}
\]
 we get the outcome $z = 0$. In this case the post-measurement state
is 
\[
	|00\rangle R|\phi\rangle.
\]
 In this case we are done--the unitary $R$ has been applied to the
state $|\phi\rangle$. On the other hand, if we get the outcome $z = 1$
then the post-measurement state is
\[
	|01\rangle \ket{\phi}.
\]
We then discard the two ancillas, adjoin two new ones in the state $|0\rangle$ and start
over from step 1.
\end{enumerate}
The algorithm fails if in each of the $m$ repetitions of step 2 we
measure $z = 1$. This occurs with probability $\frac{1}{4^{m}}$.


\section{Supplementary Figures}
\label{sec:clockreference}

\begin{figure}[H]
	\includegraphics[width=15cm]{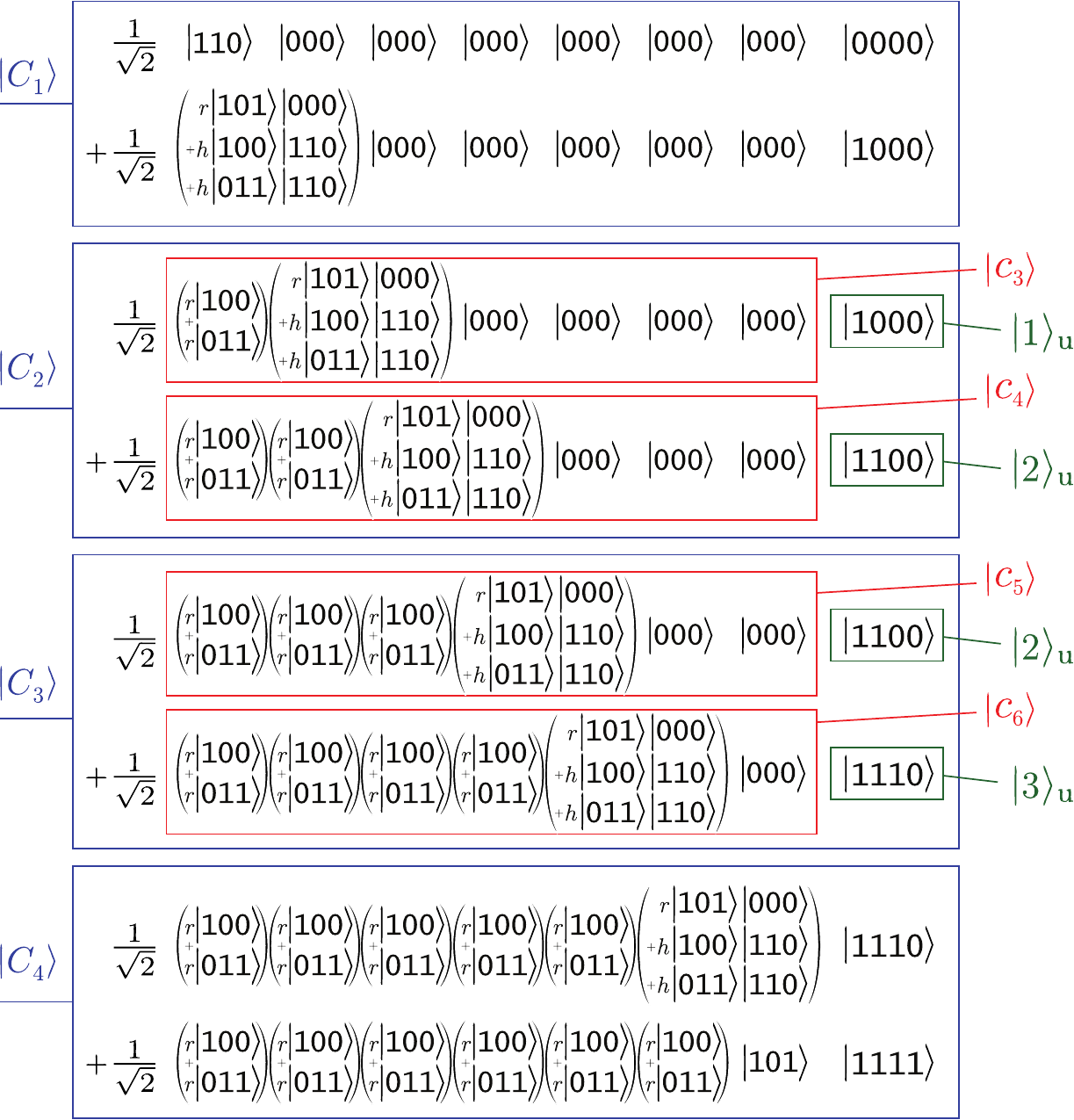}
	\caption{
	The clock states $|C_1\rangle, |C_2\rangle,|C_3\rangle,|C_4\rangle$ for $N=4$, expanded in the computational basis. These states are defined through equation \eqref{eq:C_i}. We use the shorthand $r=1/\sqrt{2}$ and $h=1/2$. 
	\label{fig:clockstates}}
\end{figure}
\newpage

\begin{figure}[H]
	\includegraphics[width=12cm]{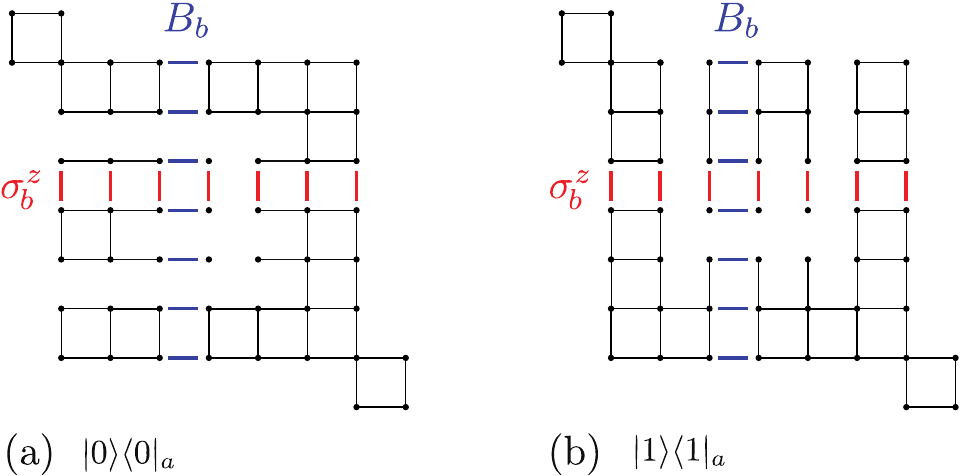}
	\caption{The terms  $h_{45}(B_{b}) \otimes \iii$ and $h_{45}(\sigma^z_b)$ from $H_{\text{horizontal}}$ \eqref{eq:H_horiz} and $H_{\text{vertical}}$ \eqref{eq:H_vert} couple the states $\ket{\mathcal{R}_0}, \ket{\mathcal{G}_0},  \ket{\mathcal{Y}_0}, \ket{\mathcal{B}_0}$ or $\ket{\mathcal{R}_1}, \ket{\mathcal{G}_1},  \ket{\mathcal{Y}_1}, \ket{\mathcal{B}_1}$ depending on the state of the control qubit. a) When the control qubit is $\ket{0}$, the state $\ket{0}\ket{y}\ket{\mathcal{R}_0}$ is coupled to $\ket{0}B \ket{y}\ket{\mathcal{G}_0}$, which is coupled to  $\ket{0}\sigma^z B \ket{y}\ket{\mathcal{B}_0}$ which is coupled to $\ket{0}B^{\dagger}\sigma^z B \ket{y}\ket{\mathcal{Y}_0}$.
	b) When the control qubit is $\ket{1}$, the state $\ket{1}\ket{y}\ket{\mathcal{R}_1}$ is coupled  $\ket{1}\ket{y}\sigma^z\ket{\mathcal{G}_1}$ which is coupled to  $\ket{1}B \sigma^z \ket{y}\ket{\mathcal{B}_1}$ which is coupled to $\ket{1}\ket{y}\sigma^z B \sigma^z\ket{\mathcal{Y}_1}$.
}
\label{fig:extrafigure2}
\end{figure}


\section{Proof of Lemma \ref{lem:complicated_gap_bound} \label{sec:Proof-of-Theorem}}

In this Section we prove Lemma \ref{lem:complicated_gap_bound}. We
build up the Hamiltonian 
\[
	\iii \otimes  H_{\text{diag}}^{(M)} + \sum_{j = 0}^{M-1}H_{V}^{j} + \sum_{j = 0}^{M-1}H_{U}^{j}
\]
a few terms at a time, lower-bounding its smallest nonzero eigenvalue
at every step. The main technical tools we use are three Corollaries
of Kitaev's geometric Lemma \cite{Kitaev:2002:CQC:863284} which we
prove below. In this Section we will make use of, and refer repeatedly
to, the details of the clock construction from Section \ref{sec:Details-of-the}.


\subsection{Notation and spectral bounds}

Let us start by establishing some notation. Recall that, if $H$ is
positive semidefinite and not identically zero, we write 
\[
\gamma(H)
\]
 for its smallest nonzero eigenvalue. If $H$ is a Hamitonian acting
on a Hilbert space $\mathcal{H}$ and $S$ is a subspace of $\mathcal{H}$,
write 
\[
H|_{S}
\]
for the restriction of $H$ to states in $S$ (it can be written as
a $\text{dim}(S)\times\text{dim}(S)$ matrix). In this paper all Hilbert
spaces and operators are finite-dimensional.

The following Corollary is a variant of Lemma \ref{lem:geom_lemma_rephrased}
which (unlike the Lemma) can be used when $H$ has a nonempty nullspace. 
\begin{cor}
\label{lem:Kitaev's-geometric-lemma}Let $H = H_{A} + H_{B}$ where $H_{A}\geq0$
and $H_{B}\geq0$ each have nonempty nullspaces. Let $\Gamma$ be
the subspace of states in the nullspace of $H_{A}$ that are orthogonal
to the nullspace of $H$, and let $\Pi_{B}$ be the projector onto
the nullspace of $H_{B}$. Then 
\[
	\gamma\left(H\right)\geq\min\{\gamma(H_{A}),\gamma(H_{B})\}\cdot\left(1-\sqrt{d}\right)
\]
where 
\[
	d = \left\Vert \Pi_{B}|_{\Gamma}\right\Vert  = \max_{|v\rangle\in\Gamma\,:\,\langle v|v\rangle = 1}\langle v|\Pi_{B}|v\rangle.
\]
\end{cor}
\begin{proof}
Let $P$ be the projector onto the nullspace of $H$ (if it is empty,
set $P = 0$). Define
\begin{align*}
	H_{A}' &  = H_{A} + \gamma(H_{A})\cdot P.
\end{align*}
Note that every state in the nullspace of $H$ is also in the nullspace
of $H_{A}$ and so $[P,H_{A}] = 0$. This implies that $\gamma(H_{A}') = \gamma(H_{A}).$
Noting that $H' = H_{A}' + H_{B}$ has no zero eigenvalues, we apply Lemma
\ref{lem:geom_lemma_rephrased} to get 
\[
	\gamma(H)\geq\min\{\gamma(H_{A}),\gamma(H_{B})\}\cdot\left(1-\sqrt{c}\right),
\]
 where 
\[
	c = \left\Vert \Pi_{B}|_{S'}\right\Vert 
\]
 and $S'$ is the groundspace of $H_{A}'.$ Noting that $S' = \Gamma$
(and hence $c = d$) completes the proof.
\end{proof}
We derive two additional Corollaries by upper bounding the parameter
$d$ in two different ways. Let $S$ be the nullspace of $H_{A}$,
and suppose that $H_{B}|_{S}$ is not the zero matrix. Note that 
\[
	\gamma(H_{B}|_{S}) = \min_{|v\rangle\in\Gamma\,:\,\langle v|v\rangle = 1}\langle v|H_{B}|v\rangle,
\]
 since the zero eigenspace of $H_{B}|_{S}$ is equal to the nullspace
of $H$. Accordingly, for any normalized state $|v\rangle\in\Gamma$,
\[
	\gamma(H_{B}|_{S})
	\leq \langle v|H_{B}|v\rangle
	\leq \langle v|\left(1-\Pi_{B}\right)|v\rangle\left\Vert H_{B}\right\Vert, 
\]
and so 
\[
	d \leq 1-\frac{\gamma(H_{B}|_{S})}{\left\Vert H_{B}\right\Vert }.
\]
Using this bound in Corollary \ref{lem:Kitaev's-geometric-lemma}
and using the fact that $1-\sqrt{1-x}\geq\frac{x}{2}$ for $x\in[0,1]$
gives the following Corollary.
\begin{cor}
\label{cor:corollary_bound}Let $H = H_{A} + H_{B}$ where $H_{A}\geq0$
and $H_{B}\geq0$ each have nonempty nullspaces. Write $S$ for the
nullspace of $H_{A}$ and suppose $H_{B}|_{S}$ is not the zero matrix.
Then 
\[
	\gamma\left(H\right)\geq\min\{\gamma(H_{A}),\gamma(H_{B})\}\cdot\frac{\gamma(H_{B}|_{S})}{2\left\Vert H_{B}\right\Vert }.
\]
\end{cor}
We now derive a different bound on $d$. For any normalized state
$|v\rangle\in\Gamma$, 
\begin{eqnarray*}
	\gamma(H_{B}|_{S}) 
	\leq \langle v|\left(1-\Pi_{B}\right)H_{B}|v\rangle
	\leq \left(
			\langle v|\left(1-\Pi_{B}\right)|v\rangle
			\langle v|H_{B}^{2}|v\rangle
		\right)^{\frac{1}{2}}
\end{eqnarray*}
using the Cauchy-Schwarz inequality. Squaring both sides and rearranging,
we get
\[
	\langle v|\Pi_{B}|v\rangle
	\leq 1-\frac{\gamma(H_{B}|_{S})^{2}}{\langle v|H_{B}^{2}|v\rangle}
	\leq 1-\frac{\gamma(H_{B}|_{S})^{2}}{F},
\]
 where 
\begin{equation}
	F = \max_{|v\rangle\in S\,:\,\langle v|v\rangle = 1}\langle v|H_{B}^{2}|v\rangle.\label{eq:F_defn}
\end{equation}
 Since this holds for all normalized $|v\rangle\in\Gamma$, we get
\[
	d\leq1-\frac{\gamma(H_{B}|_{S})^{2}}{F}.
\]
Here we need $F\neq0$, which is true whenever $H_{B}|_{S}$ is not
the zero matrix. Using this bound in Corollary \ref{lem:Kitaev's-geometric-lemma}
gives
\begin{cor}
\label{cor:F_corollary}Let $H = H_{A} + H_{B}$ where $H_{A}\geq0$ and
$H_{B}\geq0$ each have nonempty nullspaces. Write $S$ for the nullspace
of $H_{A}$ and suppose $H_{B}|_{S}$ is not the zero matrix. Then
\[
	\gamma\left(H\right)\geq\min\{\gamma(H_{A}),\gamma(H_{B})\}\cdot\frac{\gamma(H_{B}|_{S})^{2}}{2F},
\]
with $F$ defined in \eqref{eq:F_defn}.
\end{cor}


\subsection{Applying the spectral bounds}

We begin by proving a lower bound on the smallest nonzero eigenvalue
of the Hamiltonian $H_{\text{triplet}}^{(2N)}$ from \eqref{eq:H_triplet}.
\begin{lem}
\label{lem:basiclemma}
\[
	\gamma\left(H_{\textnormal{triplet}}^{(2N)}\right)\geq\frac{1}{48}.
\]
\end{lem}
\begin{proof}
Recall that 
\[
H_{\text{triplet}}^{(2N)} = H_{1} + H_{2} + H_{3}
\]
with $H_{1}$, $H_{2}$ and $H_{3}$ given in \eqref{eq:H1},
\eqref{eq:H2}, \eqref{eq:H3}. We first show that $\gamma\left(H_{1} + H_{2}\right)$
is lower bounded by a constant. Let

\begin{align*}
	H_{a} &  = \left(H_{1} + H_{2}\right)
			- \sum_{i = 0}^{2N-2}|\theta^-\rangle\langle\theta^-|_{3i + 1,3i + 2,3i + 3}, \qquad 	
	H_{b} = \sum_{i = 0}^{2N-2}|\theta^-\rangle\langle\theta^-|_{3i + 1,3i + 2,3i + 3}
\end{align*}
and note (by looking at equations \eqref{eq:H1}, \eqref{eq:H2}) that
$H_{a}$ is diagonal in the computational basis and is a sum of projectors.
Because of this, its eigenvalues are nonnegative integers and so $\gamma\left(H_{a}\right)\geq1.$
We now construct a convenient basis for the ground state subspace
of $H_{a}.$ Recall that the groundspace of $H_{1} + H_{2}$ is spanned
by states $\{|\widehat{i}\rangle\}$ and $\{|\tilde{k}\rangle\}$
from equations \eqref{eq:1_hat}-\eqref{eq:2R-2hat} and \eqref{eq:ktilde}.
We can associate a set of ground states of $H_{a}$ with each of these
states. For example, look at 
\[
	|\widehat{3}\rangle = |\theta^+ \rangle|110\rangle|000\rangle\dots|000\rangle.
\]
 Note that $|\widehat{3}\rangle$ is a zero energy ground state of
$H_{a}$, but so is 
\[
	|\theta^-\rangle|110\rangle|000\rangle\dots|000\rangle.
\]
Similarly, for each of the states $|\widehat{i}\rangle$ for $i\in\{1,\dots,4N-2\}$,
we can construct additional zero energy ground states of $H_{a}$
by replacing any subset of the $|\theta^+ \rangle$ states which appear
in the tensor product by $|\theta^-\rangle$ states. Likewise, we
obtain $2^{k}$ groundstates of $H_{a}$ from each state $|\tilde{k}\rangle$
for $k\in\{1,\dots,2N-2\}$ by making such replacements. Taking all
of the states constructed in this way, we get a spanning set of ground
states of $H_{a}.$ This basis for the groundspace of $H_{a}$ is
convenient because every basis state is also an eigenstate of $H_{b}$
with an integer eigenvalue. This shows that the projector $\Pi_{a}$
onto the groundspace of $H_{a}$ commutes with $H_{b}$ and hence
also with $H_{a} + H_{b}.$ Any eigenstate $|w\rangle$ of $H_{a} + H_{b}$
satisfying $\Pi_{a}|w\rangle = |w\rangle$ has an integer eigenvalue
as discussed above. On the other hand, an eigenstate $|w'\rangle$
with $\Pi_{a}|w'\rangle = 0$ has energy at least $\gamma(H_{a}) = 1$
since $H_{b}\geq0.$ 
Therefore,
\begin{equation}
	\gamma(H_{1} + H_{2}) = \gamma(H_{a} + H_{b})\geq1.\label{eq:bound_H1H2}
\end{equation}

We now use Corollary \ref{cor:F_corollary} with $H_{A} = H_{1} + H_{2}$
and $H_{B} = H_{3}$ to bound the smallest nonzero eigenvalue of $H_{\text{triplet}}^{(2N)}$.
First, note that $H_{3}$ is a sum of commuting projectors and therefore
$\gamma(H_{3}) = 1$. Looking at \eqref{eq:H3_insubspace}, and
writing $S$ for the groundspace of $H_{1} + H_{2}$, we see that $\gamma\left(H_{3}|_{S}\right) = \frac{1}{6}.$
Let us now look at the second power of $H_{3}$:

\begin{eqnarray*}
	H_{3}^{2} &  =  & 
		H_{3} + \sum_{i\neq j}|\gamma\rangle\langle\gamma|_{3i + 3,3i + 4,3i + 5}
		\otimes |\gamma\rangle\langle\gamma|_{3j + 3,3j + 4,3j + 5},
\end{eqnarray*}
noting that every state in the groundspace of $H_{1} + H_{2}$ is
annihilated by the second term in the above expression (one can easily
verify this fact for each of the states $|\hat{i}\rangle$ and $|\tilde{k}\rangle$
defined in equations \eqref{eq:1_hat}-\eqref{eq:2R-2hat} and \eqref{eq:ktilde}).
This implies that for all $|v\rangle\in S$, 
\[
	\langle v|H_{3}^{2}|v\rangle = \langle v|H_{3}|v\rangle,
\]
 and hence 
\[
	F = \max_{|v\rangle\in S}\langle v|H_{3}^{2}|v\rangle
	 = \max_{|v\rangle\in S}\langle v|H_{3}|v\rangle
	 = \left\Vert H_{B}|_{S}\right\Vert 
	 = \frac{2}{3},
\]
 where we used \eqref{eq:H3_insubspace} to compute $\left\Vert H_{B}|_{S}\right\Vert $.
Now applying Corollary \ref{cor:F_corollary} and using \eqref{eq:bound_H1H2}
we get 
\begin{align}
	\gamma\left(H_{\text{triplet}}^{(2N)}\right) & 
	\geq \frac{1}{2}\cdot\frac{1}{36}\cdot\frac{3}{2}\cdot1
	 = \frac{1}{48}.
	 \label{eq:triplet_clock_gap-1}
\end{align}
\end{proof}
Next, we prove a lower bound on the smallest nonzero eigenvalue of
$H_{\text{clock}}^{(N)}$ from \eqref{eq:H_N_clock}.
\begin{lem}
\label{lem:clock_lemma}
\[
	\gamma(H_{\text{\textnormal{clock}}}^{(N)})\geq\frac{1}{2048}.
\]
\end{lem}
\begin{proof}
We first bound 
\[
	\gamma\left(
			H_{\text{triplet}}^{(2N)} \otimes \iii  
			+ \iii \otimes  H_{\text{unary}}^{(N)} 
			+ H_{\text{sync}}
		\right),
\]
 where $H_{\text{unary}}^{(N)}$ and $H_{\text{sync}}$ are defined
in equations \eqref{eq:H_unary_N} and \eqref{eq:H_sync}. 
We will use Corollary \ref{cor:F_corollary} with 
\[
	H_{A} = H_{\text{triplet}}^{(2N)} \otimes \iii  + \iii \otimes  H_{\text{unary}}^{(N)}, \qquad 
	H_{B} = H_{\text{sync}}.
\]
 Note that the two terms in $H_{A}$ act on different registers and
therefore commute. The second term has nonnegative integer eigenvalues
so 
\[
	\gamma(H_{A})\geq\gamma(H_{\text{triplet}}^{(2N)} \otimes \iii )\geq\frac{1}{48}
\]
 by Lemma \ref{lem:basiclemma}. $H_{B}$ is a sum of commuting projectors
and therefore satisfies $\gamma(H_{B}) = 1.$ The zero energy groundspace
$S$ of $H_{A}$ is spanned by states
\begin{equation}
	|s_{ij}\rangle = |c_{j}\rangle \ket{\underbrace{1\cdots 1}_{i} \underbrace{0\cdots 0}_{N-i}} \label{eq:basis_for_s}
\end{equation}
 for $i = 0,\dots,N$ and $j = 1,\dots2N$, where $\{|c_{j}\rangle\}$ are
the ground states of $H_{\text{triplet}}^{(2N)}$ given in  \eqref{eq:c1}-\eqref{eq:cR}. Using equations \eqref{eq:action_of_qi}
and \eqref{eq:alpha_i}, we see that $H_{B}|_{S}$ is diagonal in the
basis \eqref{eq:basis_for_s}, with $\gamma(H_{B}|_{S}) = \frac{1}{2}$.
Now consider $H_{B}^{2}.$ Noting that 
\[
	q_{j}q_{k}|c_{m}\rangle = 0\quad\text{whenever}\; j\neq k,
\]
 and $q_{j}^{2} = q_{j}$, we get 
\[
	H_{B}^{2}|s_{ij}\rangle
	 = \left(q_{1}\otimes\ket{1}\langle 1|_{1}
		 + \sum_{j = 2}^{2N-1}q_{j}\otimes\left(|0\rangle\langle 0|_{\left\lfloor \frac{j}{2}\right\rfloor }
		 + \ket{1}\langle 1|_{\left\lfloor \frac{j}{2}\right\rfloor  + 1}\right)^{2}
		 + q_{2N}\otimes|0\rangle\langle 0|_{N}\right)|s_{ij}\rangle.
\]
Thus $H_{B}^{2}|_{S}$ is diagonal in the basis \eqref{eq:basis_for_s}
and we can evaluate its diagonal matrix elements using the above expression.
The largest such matrix element is 
\[
	\left\Vert H_{B}^{2}|_{S}\right\Vert  = 1.
\]
 Applying Corollary \ref{cor:F_corollary} with $F = 1$ gives
\[
	\gamma\left(H_{\text{triplet}}^{(2N)}\otimes \iii
		 + \iii \otimes 	H_{\text{unary}}^{(N)}
		 + H_{\text{sync}}\right)\geq\frac{1}{2}\cdot\frac{1}{48}\cdot\frac{1}{4}
	 = \frac{1}{384}.
\]

As the final step, we apply Corollary \ref{cor:F_corollary} again,
now with
\[
	H_{A} = H_{\text{triplet}}^{(2N)}\otimes\iii  
			+  \iii\otimes H_{\text{unary}}^{(N)}
		 + H_{\text{sync}}, \qquad 
	H_{B} = \sum_{i = 1}^{N}g_{2i-1,2i}(\sigma_{i}^{x}),
\]
so that $H_{A} + H_{B} = H_{\text{clock}}^{(N)}.$ Recall the basis for
the groundspace of $H_{A}$ given by the states $|e_{j}\rangle$ from
\eqref{eq:coupled_clocks}. The restriction of $H_{B}$ to
the subspace $S$ spanned by these states is given in \eqref{eq:action_of_g_x}.
Looking at this expression and making use of the fact that $\beta_{j}\geq\frac{1}{4}$
for all $j$, we see that $\gamma\left(H_{B}|_{S}\right)\geq\frac{1}{2}$.
Looking at \eqref{eq:g(U)} we see that the summands in $H_{B}$
are mutually commuting projectors. Since it is also the case that
$g_{2i-1,2i}(\sigma_{i}^{x})|e_{k}\rangle = 0$ whenever $k\notin\{2i-1,2i\}$,
we see that
\[
	g_{2j-1,2j}(\sigma_{j}^{x})g_{2i-1,2i}(\sigma_{i}^{x})|e_{k}\rangle
	 =  \delta_{ij}g_{2i-1,2i}(\sigma_{i}^{x})|e_{k}\rangle,
\]
 and hence
\[
H_{B}^{2}|_{S} = H_{B}|_{S}.
\]
This means
\[
	F  =  \left\Vert H_{B}^{2}|_{S}\right\Vert 
	 =  \left\Vert H_{B}|_{S}\right\Vert \leq\frac{2}{3},
\]
 using \eqref{eq:action_of_g_x} and the fact that $\beta_{2j-1}\leq\frac{1}{3}$.
Now applying Corollary \ref{cor:F_corollary} we get
\[
	\gamma\left(H_{\text{clock}}^{(N)}\right)
	\geq \frac{1}{2}\cdot\frac{1}{384}\cdot\frac{1}{4}\cdot\frac{3}{2}
	 =  \frac{1}{2048}.
\]
\end{proof}

Now we are ready to prove a bound on the smallest nonzero eigenvalue
of $H_{\text{diag}}^{(M)}$ from \eqref{eq:H_diag}.
Recall that this Hamiltonian acts on the Hilbert space of two clock registers, 
and its ground states are represented in Figure~\ref{fig:groundstates}.
\begin{lem}
\label{lem:Hdiag_lemma}
\[
	\gamma\left(H_{\textnormal{diag}}^{(M)}\right) = \Omega\left(\frac{1}{M}\right).
\]
\end{lem}
\begin{proof}
Looking at equations \eqref{eq:H_diag} and \eqref{eq:S_k_k_2}, write
\[
	H_{\text{diag}}^{(M)} = H_{\alpha} + H_{\beta} + H_{\gamma},
\]
where
\begin{align*}
	H_{\alpha} &  =  
		H_{\text{\text{clock}}}^{(9M + 3)}\otimes \iii  +  \iii \otimes H_{\text{clock}}^{(9M + 3)},\\
	H_{\beta} &  = 
		C_{\leq\left(9M + 1\right)}\otimes C_{\geq\left(9M + 3\right)}
		+ \sum_{k = 0}^{M-1}\left(
						C_{\leq\left(9k + 1\right)}	\otimes C_{\geq\left(9k + 3\right)}
				 +  C_{\leq\left(9k + 4\right)} \otimes C_{\geq\left(9k + 6\right)}
 		 	\right)\\
 &  + C_{\geq\left(9M + 3\right)}\otimes C_{\leq\left(9M + 1\right)} 
   +  \sum_{k = 0}^{M-1}\left(
   					C_{\geq\left(9k + 3\right)} \otimes C_{\leq\left(9k + 1\right)}
   				+ C_{\geq\left(9k + 6\right)}\otimes C_{\leq\left(9k + 4\right)}
   \right),\\
H_{\gamma} &  = 
	\sum_{k = 0}^{M-1}\left(
			h_{9k + 1,9k + 2}\otimes C_{\leq(9k + 2)}
	 +  h_{(9k + 2),(9k + 3)}\otimes C_{\geq(9k + 2)}
	 \right)\\
 &  + \sum_{k = 0}^{M-1}\left(
 		  C_{\leq(9k + 2)}\otimes h_{9k + 1,9k + 2}
 		+	C_{\geq(9k + 2)}\otimes h_{(9k + 2),(9k + 3)}
    \right)\\
 &  + \sum_{k = 0}^{M-1}\left(
       h_{9k + 4,9k + 5}\otimes C_{\leq(9k + 5)}
    +  h_{(9k + 5),(9k + 6)}\otimes C_{\geq(9k + 5)}
   \right)\\
 &  +	\sum_{k = 0}^{M-1}\left(
  			C_{\leq(9k + 5)}\otimes h_{9k + 4,9k + 5}
	   +  C_{\geq(9k + 5)}\otimes h_{(9k + 5),(9k + 6)}
	 \right)\\
 &  + h_{9M + 1,9M + 2}\otimes C_{\leq(9M + 2)}
    + h_{(9M + 2),(9M + 3)}\otimes C_{\geq(9M + 2)}\\
 &  + C_{\leq(9M + 2)}\otimes h_{9M + 1,9M + 2}
    + C_{\geq(9M + 2)}\otimes h_{(9M + 2),(9M + 3)}.
\end{align*}
Recall that the ground state subspace $S_{\alpha}$ of $H_{\alpha}$
is spanned by the states 
\begin{equation}
	\{|C_{i}\rangle|C_{j}\rangle\,:\,(i,j)\in[9M + 3]\otimes[9M + 3]\}.\label{eq:C_i_C_j}
\end{equation}
 $H_{\beta}$ satisfies $\gamma\left(H_{\beta}\right)\geq 1$ since
it is a sum of commuting projectors. Applying Corollary \ref{lem:Kitaev's-geometric-lemma}
and using the result of Lemma \ref{lem:clock_lemma} 
which bounds $\gamma\left(H_{\alpha}\right)$,
we get 
\begin{eqnarray}
	\gamma\left(H_{\alpha} + H_{\beta}\right) 
			& \geq & \frac{1}{2048}\cdot\left(1-\sqrt{d}\right),
	\label{eq:bound_alphaplusbeta}\\
	d &  =  & \left\Vert \Pi_{\beta}|_{\Gamma}\right\Vert, \nonumber 
\end{eqnarray}
 where $\Pi_{\beta}$ is the projector onto the groundspace of $H_{\beta}$,
and $\Gamma$ is the subspace of groundstates of $H_{\alpha}$ that
are orthogonal to the groundspace of $H_{\alpha} + H_{\beta}.$ Looking
at the form of $H_{\beta}$ we see that a spanning set of basis vectors
for the groundspace of $H_{\alpha} + H_{\beta}$ is given by a \emph{subse}t
of the basis vectors \eqref{eq:C_i_C_j}. The set $\Gamma$ is the complement of this subset, and in particular can be written as
\begin{eqnarray}
	\Gamma &  =  & \text{span}\{|C_{i}\rangle|C_{j}\rangle:\left(i,j\right)\in P\}\label{eq:basis_Gamma}
\end{eqnarray}
for some set $P\subset[9M + 3]\otimes[9M + 3]$ (we will not need 
an explicit expression for $P$).

We now consider $\Pi_{\beta}|_{\Gamma}$ in the basis \eqref{eq:basis_Gamma}
and bound the norm of the resulting matrix. Recall from 
\eqref{eq:C_leq} and \eqref{eq:C_geq} that the operators $C_{\leq i}$
and $C_{\geq i}$ act only on the last $9M + 3$ qubits of the clock
register, are mutually commuting, and are diagonal in the computational
basis for all $i\in[9M + 3]$. This implies that $\Pi_{\beta}$ has
the form 
\begin{equation}
	\Pi_{\beta} = \sum_{(z_{1},z_{2})\in\mathcal{A}}
		\iii \otimes \left(|z_{1}\rangle\langle z_{1}|\right)
		\otimes \iii \otimes \left(|z_{2}\rangle\langle z_{2}|\right)\label{eq:pi_beta},
\end{equation}
where the two identity operators act on the first $6(9M + 3)-3$ qubits
of the two clock registers and where $\mathcal{A\subseteq}\{0,1\}^{9M + 3}\otimes\{0,1\}^{9M + 3}$
is a set of pairs of $(9M + 3)$-bit 
strings (we will not require an explicit
expression for $\mathcal{A}$).

Now look at the expression for the states $|C_{i}\rangle$ from equation
\eqref{eq:C_i} 
with $N = 9M + 3$ (or inspect Figure~\ref{fig:clockstates}) 
and note that
\[
	\langle C_{i}|\left(\iii \otimes|z\rangle\langle z|\right)|C_{j}\rangle = 0
	\;\text{ whenever }i\neq j,
\]
where the tensor product separates the first 
$6N-3$ qubits
from the last 
$N$
and where $z$ is any 
$N$-bit
string. Using
this fact and \eqref{eq:pi_beta} we get 
\[
	\langle C_{j_{2}}|\langle C_{i_{2}}|\Pi_{\beta}|C_{i_{1}}\rangle|C_{j_{1}}\rangle
	 =  0\;\text{ whenever (}i_{1},j_{1})\neq(i_{2},j_{2}).
\]
 We evaluate the diagonal matrix elements of $\Pi_{\beta}$ using
\eqref{eq:C_i}: 
\[
	\langle C_{j}|\langle C_{i}|\Pi_{\beta}|C_{i}\rangle|C_{j}\rangle
	 =  \frac{1}{4}\left(\delta_{(u(i-1),u(j-1))\in\mathcal{A}}
	 +  \delta_{(u(i),u(j-1))\in\mathcal{A}}
	 +  \delta_{(u(i-1),u(j))\in\mathcal{A}}
	 +  \delta_{(u(i-1),u(j-1))\in\mathcal{A}}\right),
\]
 where $u(i)$ is the unary representation of $i$ with 
 $N$ bits
(the bit string with $i$ ones followed by 
$N-i$ zeros),
and $\delta_{(x,y)\in\mathcal{A}}$
is equal to one if $(x,y)\in\mathcal{A}$ and zero otherwise. Since
$\Pi_{\beta}|_{\Gamma}$ is diagonal in this basis we have 
\begin{align*}
	\left\Vert \Pi_{\beta}|_{\Gamma}\right\Vert  
  &  = \max_{(i,j)\in P}\langle C_{j}|\langle C_{i}|\Pi_{\beta}|C_{i}\rangle|C_{j}\rangle\\
  &  = \max_{(i,j)\in P}\frac{1}{4}\left(\delta_{(u(i-1),u(j-1))\in\mathcal{A}}
    +  \delta_{(u(i),u(j-1))\in\mathcal{A}} 
    +  \delta_{(u(i-1),u(j))\in\mathcal{A}}
    +  \delta_{(u(i-1),u(j-1))\in\mathcal{A}}\right).
\end{align*}
We know that $\left\Vert \Pi_{\beta}|_{\Gamma}\right\Vert $ is not
equal to $1$ since the space $\Gamma$ does not contain any states
which have zero energy for $H_{\beta}$ (by definition). Looking at
the above expression for this quantity we see that this implies 
\[
	\left\Vert \Pi_{\beta}|_{\Gamma}\right\Vert \leq\frac{3}{4}.
\]
Now using this bound in \eqref{eq:bound_alphaplusbeta} we
get
\begin{eqnarray*}
	\gamma\left(H_{\alpha} + H_{\beta}\right) &  =  & \Omega(1).
\end{eqnarray*}
For
the final step in the proof of 
Lemma~\ref{lem:Hdiag_lemma},
we use Corollary \ref{cor:corollary_bound}
with $H_{A} = H_{\alpha} + H_{\beta}$ and $H_{B} = H_{\gamma}.$ As discussed
above, the ground space of $H_{A}$ is spanned by 
the states $|C_{i}\rangle|C_{j}\rangle$
with
$(i,j)\notin P$. $H_{\gamma}$ is block-diagonal in this basis,
with a block for each connected component of the graph in Figure~\ref{fig:groundstates}.
From this Figure we 
also see that blocks have maximum size equal to 7.
Within each block there is a zero energy state given by the uniform
superposition over states $|C_{i}\rangle|C_{j}\rangle$ where $(i,j)$
are in the associated connected component. There are also higher energy
states. The lowest nonzero eigenvalue of $H_{\gamma}$ within the
subspace spanned by ground states of $H_{\alpha} + H_{\beta}$ is equal
to the smallest nonzero eigenvalue for one of these constant-size
blocks, which is a constant independent of $M$. The norm of $H_{B}$
is upper bounded as $\left\Vert H_{B}\right\Vert  = \mathcal{O}(M)$
since it is a sum of $O(M)$ projectors. Furthermore, looking at equations
\eqref{eq:C_leq}, \eqref{eq:C_geq} and \eqref{eq:h_i_explicit}, we see that 
\begin{align*}
	\{C_{\leq j}\}, \qquad
	\{C_{\geq k}\}, \qquad \text{and} \qquad
	\{h_{i,i + 1}\}
\end{align*}
are mutually commuting. 
Thus, $H_{B}$ is a sum of \emph{commuting} projectors
which implies that its eigenvalues are integers, and in particular
$\gamma(H_{B})\geq1$. Now applying Corollary \ref{cor:corollary_bound}
we get 
\[
	\gamma\left(H_{\alpha} + H_{\beta} + H_{\gamma}\right) = \Omega\left(\frac{1}{M}\right),
\]
 which completes the proof.
\end{proof}

The next Lemma bounds the smallest nonzero eigenvalue of $\sum_{j = 0}^{M-1}H_{U}^{j}$
where $H_{U}^{j}$ is defined in \eqref{eq:H_Uj_def}. 
The role of this Hamiltonian is to 
check the application of single-qubit gates 
$U^j$ within the ground state subspace of $H_{\text{diag}}^M$.
\begin{lem}
\label{lem:sumHus} 
\[
	\gamma\left(\sum_{j = 0}^{M-1}H_{U}^{j}\right) = \Omega(1).
\]
\end{lem}
\begin{proof}
Recall that the operator 
\begin{equation}
	\sum_{j = 0}^{M-1}H_{U}^{j} = \sum_{j = 0}^{M-1}\left(h_{9j + 3,9j + 4}(U^{j}) \otimes \iii 
	 +  \iii \otimes \iii \otimes h_{9j + 3,9j + 4}\right)
	 \label{eq:sum_Hu}
\end{equation}
acts on a Hilbert space with an $(n + n_{a})$-qubit computational register
and two clock registers 
\[
	\left(\mathbb{C}^{2}\right)^{(n + n_{a})} \otimes
	\left(\mathbb{C}^{2}\right)^{7N-3} \otimes
	\left(\mathbb{C}^{2}\right)^{7N-3},
\]
 where $N = 9M + 3$ is the number of states for each clock. Each $U^{j}$
is either the Hadamard gate, the $T$ gate, or the identity gate acting
on one of the qubits of the computational register. In \eqref{eq:sum_Hu},
$h_{9j + 3,9j + 4}(U^{j})$ acts nontrivially on a single qubit of the
computational register as well as qubits $6(9j + 3)-1$ and $6(9j + 3)$
of the first clock register (to see this, look at \eqref{eq:h_i_explicit}). 

Define operators 
\begin{equation}
	P_{i,i + 1} = |01\rangle\langle 01|_{6i-1,6i} + |10\rangle\langle 10|_{6i-1,6i}\label{eq:P_i,i + 1}
\end{equation}
 for $i = 1,\dots,N-1$ which act on a single clock register. Note that
we can simultaneously diagonalize \eqref{eq:sum_Hu} along with 
\begin{eqnarray}
	\{\iii \otimes  P_{9j + 3,9j + 4} \otimes \iii \}\quad & \text{and}\quad & \{\iii \otimes \iii \otimes h_{9j + 3,9j + 4}\}\label{eq:P_and_h}
\end{eqnarray}
for $j = 0,\dots,M-1$, since all of these operators are mutually commuting.
We choose to work in a basis of simultaneous eigenvectors for these
operators.

Any eigenstate $\ket{\psi}$ of \eqref{eq:sum_Hu}, which is also a
zero eigenvector of 
\begin{equation}
	\iii \otimes  P_{9j + 3,9j + 4} \otimes \iii \label{eq:P_9j + 3},
\end{equation}
satisfies 
\begin{equation}
	\left( h_{9j + 3,9j + 4}(U^{j}) \otimes \iii \right) \ket{\psi} = 0.
	\label{eq:annihilate_Pi}
\end{equation}
 Likewise, there are eigenstates $|\psi'\rangle$ of \eqref{eq:sum_Hu}
which satisfy 
\begin{equation}
	\left( \iii \otimes \iii \otimes h_{9j + 3,9j + 4} \right) \ket{\psi'} = 0.
	\label{eq:annihilate_secondtyp}
\end{equation}
We are working in a basis where each eigenstate of \eqref{eq:sum_Hu}
has a set of simultaneous eigenvalues for the projectors 
from \eqref{eq:P_and_h}. An eigenstate $|\kappa\rangle$ in this basis has
eigenvalue $0$ or $1$ for each of these projectors. Let us fix $|\kappa\rangle$
and consider the projectors which have eigenvalue zero. Looking at
equations \eqref{eq:annihilate_Pi} and \eqref{eq:annihilate_secondtyp}
we see that for each such projector there is a corresponding term
in \eqref{eq:sum_Hu} that annihilates $|\kappa\rangle.$ Let
us now focus our attention on the other terms in \eqref{eq:sum_Hu}.
Specifically, for a given eigenstate $|\kappa\rangle$ of \eqref{eq:sum_Hu},
consider the terms 
\[
	h_{9j + 3,9j + 4}(U^{j}) \otimes \iii, 
\]
 where 
 $\left( \iii \otimes  P_{9j + 3,9j + 4} \otimes \iii \right) |\kappa\rangle = |\kappa\rangle$,
and the terms 
\[
	\iii \otimes \iii \otimes h_{9k + 3,9k + 4},
\]
 where $\left(\iii \otimes \iii \otimes h_{9k + 3,9k + 4}\right) |\kappa\rangle = |\kappa\rangle$. 
 We'll call these terms the \emph{remaining operators} for the state
$|\kappa\rangle$ (they remain after removing some terms in equation
\eqref{eq:sum_Hu} which have eigenvalue zero).

We classify the eigenstates of \eqref{eq:sum_Hu} into two types. 
The first type of eigenstate is defined by the property that all of
its remaining operators commute. Note that an eigenstate of \eqref{eq:sum_Hu}
is also an eigenstate of the sum of its remaining operators. In this
case the remaining operators are commuting projectors, and so each
eigenvalue of \eqref{eq:sum_Hu} associated with such an eigenstate
is a nonnegative integer.

The second type of eigenstate has at least two remaining operators
which do not commute. Looking at \eqref{eq:sum_Hu} we see
that these two operators must be 
\[
	h_{9j + 3,9j + 4}(U^{j}), \qquad  \text{and} \qquad
	h_{9k + 3,9k + 4}(U^{k}),
\]
 for some $j$ and $k$ with $[U^{j},U^{k}]\neq0.$ Recall that both
$U^{j}$ and $U^{k}$ are either Hadamard, $T$ or the identity acting
on one of the qubits. The condition that $[U^{j},U^{k}]\neq0$ means
that the two gates must act on the same qubit, and that one of them
is the Hadamard gate and the other is the $T$ gate. We now show that
any eigenstate of this second type has eigenvalue $e$ bounded below
by a positive constant. Since each term in \eqref{eq:sum_Hu} is a projector,
we have 
\begin{equation}
	e_{\text{}}
	\geq
	\min_{\ket{\psi}\in S}\langle\psi|h_{9j + 3,9j + 4}(U^{j})
	 +  h_{9k + 3,9k + 4}(U^{k})\ket{\psi}\label{eq:e_bound},
\end{equation}
where $S$ is the subspace of states satisfying 
\[
	\left( \iii \otimes  P_{9j + 3,9j + 4} \otimes \iii \right) \ket{\psi}
	 =  \ket{\psi}
	 \quad\text{and}\quad
	 \left(\iii \otimes  P_{9k + 3,9k + 4} \otimes \iii \right) \ket{\psi}
	 = \ket{\psi}.
\]
 For states $\ket{\psi}\in S$, the Hamiltonian 
\begin{equation}
	h_{9j + 3,9j + 4}(U^{j}) + h_{9k + 3,9k + 4}(U^{k})\label{eq:two_remaining_ops}
\end{equation}
 acts nontrivially within an $8$-dimensional space that is the tensor
product of three 2-dimensional spaces. These are: the states of the
single computational qubit on which the unitaries $U^{j}$ and $U^{k}$
act, the states $|01\rangle,|10\rangle$ on qubits $6(9j + 3)-1,6(9j + 3)$
of the first clock register, and the states $|01\rangle,|10\rangle$
on qubits $6(9k + 3)-1,6(9k + 3)$ of the first clock register (note $j\neq k$).
Identifying this 8-dimensional space with the space of $3$ qubits,
we can write the action of \eqref{eq:two_remaining_ops} as the Hamiltonian
\begin{align}
	&\frac{1}{2}\left(
		\iii \otimes \iii  \otimes \iii 
		-\widehat{H}\otimes\ket{1}\langle 0| \otimes \iii 
		-\widehat{H}\otimes\ket{1}\langle 0| \otimes \iii 
	\right)
	\label{eq:8-d operator}\\
 + &\frac{1}{2}\left(
 		\iii \otimes \iii  \otimes \iii 
 		- T^{\dagger} \otimes \iii \otimes\ket{1}\langle 0|
 		- T \otimes \iii \otimes|0\rangle\langle 1|
 	\right).
 	\nonumber 
\end{align}
 Here we used the fact that $U^{j}$ and $U^{k}$ are the Hadamard
and $T$ gates respectively. Equation \eqref{eq:e_bound} says that
$e_{\text{}}$ is lower bounded by the smallest eigenvalue of this
operator, which is a positive constant equal to $0.00937\dots$. 

Let us now summarize what we have shown regarding the eigenstates
of \eqref{eq:sum_Hu}. We have classified each eigenstate as either
type 1 or type 2 according to the form of its remaining operators.
We showed that every eigenstate of the first type has a nonnegative
integer eigenvalue, and that every eigenstate of the second type has
eigenvalue at least $0.00937\dots$. From this we see that 
\[
	\gamma\left(\sum_{j = 0}^{M-1}H_{U}^{j}\right)
\]
is at least $0.00937\dots$, which completes the proof.
\end{proof}


We now prove a lower bound on the smallest nonzero eigenvalue of $\sum_{j = 0}^{M-1}H_{V}^{j}$
where $H_{V}^{j}$ is defined in \eqref{eq:H_Vj_def}. 
The role of this Hamiltonian is to 
check each of the $M$ applications of the two-qubit gate
$V$ within the ground state subspace of $H_{\text{diag}}^M$.
\begin{lem}
\label{lem:sumofHv}
\[
	\gamma\left(\sum_{j = 0}^{M-1}H_{V}^{j}\right) = \Omega(1).
\]
\end{lem}
\begin{proof}
We use a similar strategy to that used in Lemma \ref{lem:sumHus}.
First, we define a set of mutually commuting projectors which also
commute with 
\begin{equation}
\sum_{j = 0}^{M-1}H_{V}^{j} = \sum_{j = 0}^{M-1}\left(H_{\text{horizontal}}^{j} + H_{\text{vertical}}^{j}\right).\label{eq:sum_HVS}
\end{equation}
 The specific set of projectors reflects the structure of the terms
appearing in \eqref{eq:sum_HVS} (look at equations \eqref{eq:H_horizontalj}
and \eqref{eq:H_verticalj}). In particular, consider the set of projectors
\begin{eqnarray}
	\iii \otimes  h_{9j + 6,9j + 7} \otimes \iii  
		&  & \iii \otimes \iii \otimes h_{9j + 6,9j + 7}
		\label{eq:proj_line1}\\
	\iii \otimes  h_{9j + 6,9j + 7} \otimes C_{\geq 9j + 10} 
		&  & \iii \otimes  C_{\geq 9j + 10} \otimes h_{9j + 6,9j + 7}\\
		\iii \otimes  h_{9j + 8,9j + 9} \otimes C_{\leq 9j + 6} 
		&  & \iii \otimes  C_{\leq 9j + 6} \otimes h_{9j + 8,9j + 9}\\
	\iii \otimes  h_{9j + 9,9j + 10} \otimes \iii  
		&  & \iii \otimes \iii \otimes h_{9j + 9,9j + 10}\\
	\iii \otimes  h_{9j + 9,9j + 10} \otimes C_{\geq 9j + 10} 
		&  & \iii \otimes  C_{\geq 9j + 10} \otimes h_{9j + 9,9j + 10}\\
	\iii \otimes  h_{9j + 8,9j + 9} \otimes C_{\geq 9j + 10} 
		&  & \iii \otimes  C_{\geq 9j + 10} \otimes h_{9j + 8,9j + 9}\\
	\iii \otimes  P_{9j + 7,9j + 8} \otimes \iii  
		&  & \iii \otimes \iii \otimes P_{9j + 7,9j + 8}
	\label{eq:proj_last_line}
\end{eqnarray}
 for $j = 0,\dots,M-1$, where $P_{i,i + 1}$ is defined in \eqref{eq:P_i,i + 1}.
Recall that the three registers separated by tensor products are the
$(n + n_{a})$-qubit computational register, the first clock register
and the second clock register. Using the definitions \eqref{eq:h_i_explicit},
\eqref{eq:C_leq} and \eqref{eq:C_geq} one can check that these projectors
are mutually commuting and that they commute with \eqref{eq:sum_HVS}.
We can therefore simultaneously diagonalize \eqref{eq:sum_HVS} along
with the projectors listed above. We choose to work in a basis of
eigenstates of \eqref{eq:sum_HVS} in which all of these operators are
diagonal.

For a given eigenstate $|\kappa\rangle$ of \eqref{eq:sum_HVS}, each
of the projectors in equations \eqref{eq:proj_line1}-\eqref{eq:proj_last_line}
can have eigenvalue either $0$ or $1$. We defined the projectors
above in such a way that, if one of them has eigenvalue zero then
one of the terms in \eqref{eq:sum_HVS} annihilates $|\kappa\rangle$.
For example, if
\[
	\left( \iii \otimes  P_{7,8} \otimes \iii \right) |\kappa\rangle = 0
\]
 (this is the seventh projector from the top in the left column of equations
\eqref{eq:proj_line1}-\eqref{eq:proj_last_line} with $j = 0$), then 
\[
	\left( h_{7,8}(B_{b_{0}}) \otimes \iii \right) |\kappa\rangle = 0
\]
 (this is the seventh term in \eqref{eq:H_horizontalj} with
$j = 0$). In this way, each of the $7$ projectors on the left-hand
(right-hand) column of equations \eqref{eq:proj_line1}-\eqref{eq:proj_last_line}
is associated with one of the $7$ terms in $H_{\text{horizontal}}^{j}$ ($H_{\text{vertical}}^{j}).$
To see this, compare equations \eqref{eq:proj_line1}-\eqref{eq:proj_last_line}
with equations \eqref{eq:H_horizontalj} and \eqref{eq:H_verticalj}.
Now let us fix an eigenstate $|\kappa\rangle$ of \eqref{eq:sum_HVS},
which has a specific set of eigenvalues for the projectors in equations
\eqref{eq:proj_line1}-\eqref{eq:proj_last_line}. For each of the projectors
(from the list \eqref{eq:proj_line1}-\eqref{eq:proj_last_line} ) for
which $|\kappa\rangle$ has eigenvalue $1$, consider the corresponding
term from equations \eqref{eq:H_horizontalj} and \eqref{eq:H_verticalj}.
We call these terms the \emph{remaining operators} for $|\kappa\rangle$
and we classify eigenstates into two types based on these terms. 

The first type of eigenstate has the property that all of its remaining
operators commute. Any eigenvector of \eqref{eq:sum_HVS} of this type
is an eigenvector of a sum of commuting projectors and therefore has
an integer eigenvalue. 
Therefore, the lowest nonzero eigenvalue for this first
type of eigenstate is at least $1.$ 

Now consider a second type of eigenstate where at least two of its
remaining operators do not commute. Looking at equations \eqref{eq:H_horizontalj}
and \eqref{eq:H_verticalj} we see that possible pairs of noncommuting
remaining operators are\smallskip{}

(1): $\ket{1}\langle 1|_{a_{j}} \otimes \iii \otimes h_{9j + 6,9j + 7}$ and
$h_{9k + 7,9k + 8}(B_{b_{k}})$ with $a_{j} = b_{k}$ (and hence $j\neq k$),
or\smallskip{}

(2): $|0\rangle\langle 0|_{a_{j}} \otimes h_{9j + 6,9j + 7} \otimes \iii $ and
$h_{9k + 7,9k + 8}(B_{b_{k}})$ with $a_{j} = b_{k}$ (and hence $j\neq k$),
or\smallskip{}

(3): $h_{9j + 7,9j + 8}(B_{b_{j}}) \otimes \iii $ and $h_{9k + 7,9k + 8}(\sigma_{b_{k}})$
with $b_{k} = b_{j}$ (it may be that $j = k$). \smallskip{}

We claim that any eigenstate $|\kappa\rangle$ of \eqref{eq:sum_HVS}
of this second type has eigenvalue bounded below by a positive constant.
To see this, first suppose two of the remaining operators are of the
form (1). For these to be remaining operators it must be the case
that $|\kappa\rangle$ is a $ + 1$ eigenvector of both 
\begin{equation}
\iii \otimes \iii \otimes h_{9j + 6,9j + 7}\qquad\text{and}\qquad\iii \otimes \iii \otimes P_{9k + 7,9k + 8}.\label{eq:remain_ops}
\end{equation}
The eigenvalue associated with this state is bounded below by the
smallest eigenvalue of 
\begin{equation}
\ket{1}\langle 1|_{a_{j}} \otimes \iii \otimes h_{9j + 6,9j + 7} + h_{9k + 7,9k + 8}(B_{b_{k}})\label{eq:sum_case1}
\end{equation}
within the joint $ + 1$ eigenspace of the two projectors \eqref{eq:remain_ops}.
Within this space, \eqref{eq:sum_case1} acts nontrivially on an 4-dimensional
space that is the tensor product of two dimensional spaces: the single
qubit space for qubit $a_{j} = b_{j}$ of the computational register
and the space spanned by the two basis states $|01\rangle,|10\rangle$
of qubits $6(9j + 7)-1$ 
and $6(9j + 7)$ of the second clock register.
Within this space \eqref{eq:sum_case1} has the form (identifying the
4-dimensional space with two qubits): 
\[
	\ket{1}\langle 1| \otimes \iii  
	+ \frac{1}{2}\left(
			\iii \otimes \iii 
			- B^{\dagger} \otimes \ket{1}\langle 0|
			- B\otimes |0\rangle\langle 1|
	\right),
\]
 which has smallest eigenvalue $\thickapprox 0.076$. Any eigenstate
$|\kappa\rangle$ which has two remaining operators of the form (1)
has eigenvalue bounded below by this positive constant. In cases (2)
and (3) we can follow the same steps (as in case (1) above) to lower
bound the eigenvalue. In case (2) we bound it by the smallest eigenvalue
of the 2-qubit operator 
\[
	|0\rangle\langle 0| \otimes \iii  
	+ \frac{1}{2}\left(
		\iii \otimes \iii
		- B^{\dagger} \otimes \ket{1}\langle 0|
		- B\otimes |0\rangle\langle 1|
	\right)
\]
 (which is $\thickapprox0.076$) and in case (3) we get a lower bound
from the smallest eigenvalue of the $3$-qubit operator 
\[
	\frac{1}{2}\left(
		\iii \otimes \iii \otimes \iii
		-B^{\dagger} \otimes \ket{1}\langle 0| \otimes \iii 
		-B\otimes |0\rangle\langle 1| \otimes \iii 
	\right) 
		+ \frac{1}{2}\left(
			\iii \otimes \iii \otimes \iii
			-\sigma^{z} \otimes \iii \otimes \ket{1}\langle 0|
			-\sigma^{z} \otimes \iii |\otimes |0\rangle\langle 1|
		\right)
\]
 (it is also $\thickapprox 0.076$).

We have thus shown that any eigenstate (of \eqref{eq:sum_HVS}) of the
first type has a nonnegative integer eigenvalue and any eigenstate
of the second type has smallest eigenvalue bounded below by a positive
constant. This completes the proof. 
\end{proof}
We now add the two-qubit gate terms to $\iii \otimes  H_{\text{diag}}^{(M)}$
and we lower bound the smallest nonzero eigenvalue of the resulting
Hamiltonian.
\begin{lem}
\label{lem:lemdiagplusHv}
\[
	\gamma\left(\iii \otimes  H_{\text{diag}}^{(M)} + \sum_{j = 0}^{M-1}H_{V}^{j}\right) 
	= \Omega\left(\frac{1}{M^{2}}\right).
\]
\end{lem}
\begin{proof}
We apply Corollary \ref{cor:corollary_bound} with 
\begin{eqnarray*}
	H_{A} &  =  & \iii \otimes  H_{\text{diag}}^{(M)},\\
	H_{B} &  =  & \sum_{j = 0}^{M-1}H_{V}^{j}.
\end{eqnarray*}
However, we first need 
to understand the action of $H_{B}$
on the groundspace $S$ of $H_{A}$. We showed in Section \ref{sub:Completeness_Section}
that $H_{B}|_{S}$, written in the unnormalized basis \eqref{eq:phi_J-1}, is block
diagonal. Of course, the matrix has the same block diagonal structure when written in a normalized basis. Furthermore, each nonzero block is identical, equal to a specific $100\times100$ matrix (independent of $M$ or the number
of qubits). The smallest nonzero eigenvalue of $\gamma(H_{B}|_{S})$
is therefore constant, equal to the smallest nonzero eigenvalue of
this matrix. Plugging this into Corollary \ref{cor:corollary_bound}
and using the facts that $\left\Vert H_{B}\right\Vert  = \mathcal{O}(M)$,
$\gamma(H_{A}) = \Omega\left(\frac{1}{M}\right)$ (from Lemma \ref{lem:Hdiag_lemma})
and $\gamma(H_{B}) = \Omega(1)$ (from Lemma \ref{lem:sumofHv}) we
get 
\[
	\gamma\left(\iii \otimes  H_{\text{diag}}^{(M)}
	 +  \sum_{j = 0}^{M-1}H_{V}^{j}\right)
	 =  \Omega\left(\frac{1}{M^{2}}\right).
\]

\end{proof}


\subsection{Proof of Lemma \ref{lem:complicated_gap_bound}}

We are now 
ready to finish the proof of Lemma \ref{lem:complicated_gap_bound}, 
the main result of this Appendix.
We add in the single-qubit gate terms $H_{U}^{j}$ and we use the
bound from Lemma \ref{lem:sumHus}. We prove

\[
	\gamma\left(\iii \otimes 
	H_{\text{diag}}^{(M)}
	 +  \sum_{j = 0}^{M-1}H_{V}^{j}
	 +  \sum_{j = 0}^{M-1}H_{U}^{j}\right)
	 =  \Omega\left(\frac{1}{M^{5}}\right).
\]

\begin{proof}
Let 
\begin{align*}
	H_{A} &  = \iii \otimes  H_{\text{diag}}^{(M)} + \sum_{j = 0}^{M-1}H_{V}^{j}, \\
	H_{B} &  = \sum_{j}H_{U}^{j},
\end{align*}
 and consider the action of $H_{B}$ on the groundspace $S$ of $H_{A}.$
Recall the (unnormalized) basis for the groundspace of $H_{A}$ given
in equations \eqref{eq:Basis_K0} and \eqref{eq:psiV}, where $|\phi\rangle$
ranges over a complete orthonormal basis $\Lambda$ of the $n + n_{a}$-qubit
computational register. As noted in Section \ref{sec:Circuit-to-Hamiltonian-mapping},
$H_{B}|_{S}$ is block diagonal in this basis with $2^{n + n_{a}}$
blocks because $H_{B}$ only connects states associated with the same
$|\phi\rangle\in\Lambda.$ Furthermore, each block of $H_{B}|_{S}$
in this unnormalized basis is equal to $\frac{1}{4}L$ where $L$
is the Laplacian of the graph in Figure~\ref{fig:laplaciangraph}. 
We can write each block of $H_{B}|_{S}$ in an orthonormal basis as
\[
	\frac{1}{4}DLD,
\]
 where 
\[
	D  =  	\text{diag}\left( \frac{1}{\sqrt{7}},\frac{1}{\sqrt{4}},\frac{1}{\sqrt{4}},
	\frac{1}{\sqrt{43}},\frac{1}{\sqrt{4}},\frac{1}{\sqrt{4}},
	\frac{1}{\sqrt{43}},\dots,\frac{1}{\sqrt{4}},\frac{1}{\sqrt{4}},\frac{1}{\sqrt{43}} \right)
\]
 takes care of the normalizations (as per \eqref{eq:norm1}).
Thus, $\gamma(H_{B}|_{S})$ is equal to $\gamma\left(\frac{1}{4}DLD\right)$,
which we now bound. First, note that
\[
	\gamma(L) = \Omega\left(\frac{1}{M^{2}}\right),
\]
 which follows from known bounds on the eigenvalue gap of a Laplacian
matrix (for example from Theorem 4.2 of reference \cite{laplacianeigs}).
To finish the proof we relate $\gamma(L)$ to $\gamma\left(\frac{1}{4}DLD\right).$
Let $|w_{0}\rangle$ be the normalized all ones vector, the unique
zero energy vector for $L$. Then, up to normalization, $\frac{1}{D}|w_{0}\rangle$
is the unique zero energy state for $\frac{1}{4}DLD$, and for some
normalized vector $|v\rangle$ satisfying
\begin{align}
	\langle w_{0}|D^{-1}|v\rangle = 0, \label{eq:orthog}
\end{align}
we have 
\begin{align}
	\gamma\left(\frac{1}{4}DLD\right) &  = \frac{1}{4}\langle v|DLD|v\rangle\nonumber \\
 &  = \frac{1}{4}\frac{\langle q|L|q\rangle}{\langle q|q\rangle}\langle v|D^{2}|v\rangle\nonumber \\
 & \geq\frac{1}{4}\frac{\langle q|L|q\rangle}{\langle q|q\rangle}\frac{1}{43},
 \label{eq:gamma_dld}
\end{align}
 where $|q\rangle = D|v\rangle$ and we used the fact that the smallest
eigenvalue of $D^{2}$ is $\frac{1}{43}$. Write 
\[
	|q\rangle = \alpha|w_{0}\rangle + \beta|w_{\perp}\rangle,
\]
 where $\langle w_{\perp}|w_{0}\rangle = 0$ and $\langle w_{0}|w_{0}\rangle = \langle w_{\perp}|w_{\perp}\rangle = 1$
and note that 
\[
	\langle w_{0}|D^{-1}|v\rangle 
	= \langle w_{0}|D^{-2}|q\rangle 
	= \alpha\langle w_{0}|D^{-2}|w_{0}\rangle 
	+ \beta\langle w_{0}|D^{-2}|w_{\perp}\rangle = 0
\]
by \eqref{eq:orthog} and hence
\[
	\left|\frac{\alpha}{\beta}\right|
	 =  \frac{
			\left|\langle w_{0}|D^{-2}|w_{\perp}\rangle\right|
			}{
			\langle w_{0}|D^{-2}|w_{0}\rangle
		}
	\leq \frac{43}{4},
\]
 where we used the facts that the smallest eigenvalue of $D^{-2}$
is $4$ and its largest eigenvalue is $43.$ Noting that 
\begin{align*}
	\frac{\langle q|L|q\rangle}{\langle q|q\rangle} 
		&  = \frac{\beta^{2}}{\alpha^{2} + \beta^{2}}\langle w_{\perp}|L|w_{\perp}\rangle\\
 & \geq\frac{1}{\left(\left(\frac{43}{4}\right)^{2} + 1\right)}\gamma(L)\\
 &  = \Omega\left(\frac{1}{M^{2}}\right),
\end{align*}
we get 
\[
	\gamma(H_{B}|_{S})  =  \gamma\left(\frac{1}{4}DLD\right)
	  =  \Omega\left(\frac{1}{M^{2}}\right)
\]
 by plugging into \eqref{eq:gamma_dld}. Now applying Corollary
\ref{cor:corollary_bound}, using Lemmas \ref{lem:sumHus} and \ref{lem:lemdiagplusHv}
and using the bound $\left\Vert H_{B}\right\Vert  = \mathcal{O}(M)$,
we finally get the desired result. \end{proof}

\end{document}